\documentclass{lmcs}
\pdfoutput=1

\usepackage{lastpage}
\lmcsdoi{15}{1}{11}
\lmcsheading{}{\pageref{LastPage}}{}{}%
{Mar.~22,~2018}{Feb.~13,~2019}{}

\usepackage{hyperref}
\usepackage{amsmath,amsthm,amsfonts,latexsym,amssymb,amscd}
\usepackage{chemarr}
\usepackage{algorithmic}
\usepackage{algorithm}

\DeclareMathOperator{\Con}{Con}

\newcommand{\bd}{\begin{description}}
\newcommand{\ed}{\end{description}}
\newcommand{\lv}{{[\![}}  
\newcommand{\rv}{{]\!]}}

\renewcommand{\kappa}{m}

\newcommand{\var}[1]{{\mathcal{\uppercase{#1}}}}
\newcommand{\class}[1]{{\mathcal{\uppercase{#1}}}}

\newcommand{\al}[1]{\mathbf{#1}}

\newcommand{\aaa}{\mathsf{a}}
\newcommand{\ddd}{\mathsf{d}}

\newcommand{\mm}{\mathsf{m}}
\newcommand{\sss}{\mathsf{s}}
\newcommand{\ttt}{\mathsf{t}}
\newcommand{\PP}{\mathsf{P}}
\newcommand{\NP}{\mathsf{NP}}

\newcommand{\FP}{\mathsf{FP}}
\newcommand{\FNP}{\mathsf{FNP}}

\newcommand{\HH}{\mathbb{H}}
\newcommand{\SSS}{\mathbb{S}}
\newcommand{\PPP}{\mathbb{P}}

\DeclareMathOperator{\SMP}{\text{\sc SMP}}
\DeclareMathOperator{\SIP}{\text{\sc SIP}}
\DeclareMathOperator{\CompactRep}{\text{\sc CompactRep}}
\DeclareMathOperator{\LocalRep}{\text{\sc LocalRep}}
\DeclareMathOperator{\NeedForkWitnesses}{\text{\sc NeedMoreForkWitnesses}}

\DeclareMathOperator{\IsRepresentable}{\text{\sc IsRepresentable}}
\DeclareMathOperator{\fork}{\text{\sc fork}}

\newcommand{\LLL}{\Lambda}
\newcommand{\FFF}{\Phi}

\newcommand{\dcoh}{\text{$\ddd$-coh}}

\DeclareMathOperator{\SMPd}{\SMP_{\dcoh}}

\DeclareMathOperator{\proj}{proj}
\DeclareMathOperator{\Irr}{Irr}
\newcommand{\interval}[2]{\lv\, {#1},{#2}\,\rv}

\def\PSPACE{\mathrm{PSPACE}}

\let\phi=\varphi
\let\epsilon=\varepsilon
\let\bar=\overline
\let\hat=\widehat
\let\tilde=\widetilde
\def\crc#1{\mathsf{Circ}(#1)}
\newcommand{\crcf}[2]{\mathsf{Circ}_{#1}(#2)}
\newcommand{\crcplus}[1]{\mathsf{Circ}^+(#1)}

\begin{document}

\title[Subpower Membership Problem]{The Subpower Membership Problem\\
for Finite Algebras with Cube Terms}

\author[A.~Bulatov]{Andrei Bulatov\rsuper{a}}
\address{\lsuper{a}School of Computing Science\\
Simon Fraser University\\
Burnaby BC, Canada V5A 1S6}
\email{abulatov@sfu.ca}

\author[P.~Mayr]{Peter Mayr\rsuper{b}}
\address{\lsuper{b}Department of Mathematics\\
University of Colorado\\
Boulder CO, USA, 80309-0395}
\email{\{agnes.szendrei,peter.mayr\}@colorado.edu}

\author[\'A. Szendrei]{\'Agnes Szendrei\rsuper{b}}

\thanks{This material is based upon work supported by
the Austrian Science Fund (FWF) grant no.\ P24285,
the National Science Foundation grant no.\ DMS 1500254, 
the Hungarian National Foundation for Scientific Research (OTKA)
grant no.\ K104251 and K115518, and an NSERC Discovery Grant.
}

\begin{abstract}
The subalgebra membership problem
is the problem of deciding if
a given element belongs to an algebra
given by a set of generators. This
is one of the best established
computational problems in algebra.
We consider a variant of this problem,
which is motivated by recent progress in the Constraint 
Satisfaction Problem, and is
often referred to as the Subpower Membership 
Problem (SMP).
In the SMP we are given a set of tuples in
a direct product of algebras from a fixed
finite set $\class{K}$ of finite algebras, and are
asked whether or not a given tuple belongs to the subalgebra of
the direct product generated by a
given set. 

Our main result is that
the subpower membership problem $\SMP(\class{K})$
is in $\PP$ if $\class{K}$ is a finite set of finite algebras 
with a cube term, provided $\class{K}$ is contained in a residually
small variety.
We also prove that for any finite set of finite algebras
$\class{K}$ in a variety with a cube term,
each one of the problems 
$\SMP(\class{K})$, $\SMP(\HH\SSS\class{K})$, and 
finding compact representations
for subpowers in $\class{K}$, 
is polynomial time reducible to any of the others, and
the first two lie in $\NP$.
\end{abstract}

\maketitle


\section{Introduction}
\label{sec-intro}

The subalgebra membership problem is one of the most well established and
most thoroughly
studied algorithmic problems in algebra. In this problem we are given a 
(finite) set of elements, called generators, of
an algebraic structure (briefly: algebra), and we are
asked to check whether or not a given element belongs to the subalgebra 
generated by the generators. This question has been studied 
within the general algebraic theory 
\cite{Ko:CFP,Ko:LBNP,FHL:PTA,Ma:SMP},
and
it is also an important problem in
computational group theory 
\cite{Si:CMSP,BMS:UPSC,GW:SPC}
where it is referred to as the subgroup membership problem.
In addition to its
theoretical interest, the subalgebra membership problem
for general algebraic structures
has found applications in some learning
algorithms~\cite{DJ:LQF,BCD:LIC,IM:TLAA}, and
the subgroup membership problem has found applications in areas
such as cryptography \cite{Shpilrain05,SZ:USM}. 

The subgroup membership problem was the first to be studied. The celebrated
algorithm by Sims \cite{Si:CMSP} decides, for any given
set of permutations $a_1,\dots, a_k$ on a finite set $X$,
whether or not a given permutation $b$
belongs to the subgroup of the
full symmetric group $\al{S}_X$
generated by $a_1,\dots, a_k$. This algorithm 
has been widely used, and a number of improvements have been suggested
\cite{FHL:PTA,Kn:ERP,Le:AFB}.
While Sims' algorithm is
quite efficient, Kozen \cite{Ko:LBNP} proved that if the
group $\al{S}_X$ is replaced
by the full transformation semigroup on $X$,
the subalgebra membership 
problem is $\PSPACE$-complete.

Kozen \cite{Ko:CFP} also considered the subalgebra membership 
problem for infinite algebras.
Clearly, in this case the elements of the algebras in 
question must be efficiently represented. The natural 
approach is to consider algebras that are finitely presented by a finite set of 
generators and finitely many defining relations.  Kozen suggested a polynomial 
time algorithm for this version of the problem.

Interest in the subalgebra membership problem has been renewed in relation 
to the study of the Constraint Satisfaction Problem
(CSP). Recall that in a CSP the 
goal is to assign values from a certain domain to variables subject
to specified constraints.
Representation of the
constraints is a very important issue in the CSP 
research, especially from the practical perspective, given the wide variety of 
constraint types considered.
Unfortunately, if we do not restrict ourselves to a 
specific constraint type, choices are limited.
The standard approach in the general case
(in finite domain CSPs)
is to represent constraints by constraint relations,
that is, to give an explicit 
list of all the allowed combinations of values.
This is clearly not space efficient, therefore more
concise representations of general constraints are
highly sought.

The algebraic approach to the CSP may offer a possible method
for representing
general constraints in a better way. Under this approach
the constraint relations are considered along with their
\emph{polymorphisms}, that is, operations
on the set of values which preserve all constraint relations.
Hence, each constraint relation can be viewed as (the underlying set of)
a subalgebra of a direct power of a certain finite algebra
associated to the problem; namely, the algebra on 
the set of values whose operations are all polymorphisms described above.
Therefore, knowing the polymorphisms, each
constraint relation can be given by a set of generators,
which is often much smaller than the
constraint relation (i.e., the generated subalgebra)
itself.
In a number of cases the 
representation
of a constraint relation
by a set of generators is exponentially
more concise
(in the arity of the constraint)
than listing all of its elements.
This is especially valuable when
the constraint relations have high arity.
But even if we are not in one of those lucky cases, savings may
be substantial. 

The main difficulty with representations of this kind, however,
is that although 
they are very space efficient, they can be inefficient in other aspects.
For example, it is not clear whether or not  it is possible to
efficiently check that a tuple belongs to a
constraint relation given by generators, but that check
is unavoidable if we want to test
if some mapping is a solution to a CSP. 
Thus we are 
led to the following variation of the subalgebra membership 
problem  suggested by Willard \cite{Wi:FUPC}.
Let $\al A$ be a finite algebra with finitely many fundamental operations. The 
\emph{subpower membership problem} ($\SMP$) for $\al A$  is the 
following decision problem:

\medskip
\pagebreak

\noindent $\SMP(\al{A})$:
\begin{itemize}
\item
INPUT:  
$a_1,\dots,a_k,b\in \al{A}^n$.
\item
QUESTION: Is $b$ in the subalgebra of $\al{A}^n$ 
generated by $\{a_1,\dots,a_k\}$?
\end{itemize}

\smallskip
\noindent
Observe that this problem is different from the problems studied in 
\cite{FHL:PTA,Ko:LBNP}, since the finite algebra $\al A$
is fixed and 
subalgebras are generated in
direct powers of $\al A$.

A naive algorithm for solving $\SMP(\al{A})$
is to compute and list
all elements of
the subalgebra $\al B$ of $\al{A}^n$ generated by $\{a_1,\dots,a_k\}$,
and then check whether $b$ is in $\al B$. Since the size of the input
$a_1,\ldots,a_k,b$ is $(k+1)n$, while the best upper bound for the size of
$\al B$ is $|A|^n$,
and for many algebras $\al{A}$,
some subalgebras $\al{B}$ of $\al{A}^n$ (say, $\al{B}=\al{A}^n$
itself) have generating sets of size
$k$ bounded by a polynomial of $n$,
the time complexity of the naive algorithm is exponential.
It turns out that without further restrictions on $\al A$, one cannot do 
better than the naive algorithm:
it follows from the
main result of M.~Kozik~\cite{Ko:FSFE} that there exists a finite algebra
$\al A$ such that the subpower membership problem for $\al A$ is
EXPTIME-complete.

In contrast, for finite algebras $\al A$ in many
familiar classes, the subpower membership problem is in $\PP$,
that is, there is a polynomial time algorithm for solving the problem.
For example, if $\al A$ is a group, then a variant of Sims' algorithm
(cf.\ \cite{FHL:PTA})
solves the problem in polynomial time. 
Other simple algorithms work if
$\al A$ is a finite lattice or a finite lattice with additional operations
(use the Baker--Pixley Theorem~\cite{BP:PIAT})
or if $\al A$ is a finite semilattice.
A recent result of A.~Bulatov, M.~Kozik, P.~Mayr, and M.~Steindl~\cite{BKMS:SMP} 
extends this observation on semilattices
to any finite commutative semigroup $\al A$ by
showing that if $\al A$ embeds into a direct product of a Clifford semigroup
and a nilpotent semigroup, then the subpower membership problem for $\al A$
is in $\PP$, and for all other finite commutative semigroups $\al A$
the problem is $\NP$-complete.

Extending the result on groups mentioned earlier R.~Willard~\cite{Wi:FUPC} 
proved that the subpower membership problem is in $\PP$ for every
finite algebra $\al A$ that is an expansion of a group by multilinear operations. 
In particular, this is the case for every finite ring, finite module, and finite 
$K$-algebra. It is not known whether this statement remains true if the word
`multilinear' is omitted. 

Returning to the CSP motivation
of the subpower membership problem,
the case in which
representing constraint relations by generators 
is the most advantageous, is when the finite
  algebra $\al{A}$ associated to the CSP has the property that
all subalgebras of finite powers $\al{A}^n$ of $\al{A}$
admit small generating sets;
here `small' means that the size
is bounded by $p(n)$ for some
polynomial $p$ depending only on $\al{A}$.
These algebras have 
been intensively studied within the theory of constraint satisfaction problems,
and 
constitute one of the two known major classes of CSPs
for which general polynomial time algorithms were found several years ago
\cite{BD:SAMC,IM:TLAA}.
It was proved in \cite{BI:VFSP} that a finite algebra $\al{A}$ has the
aforementioned property on small generating sets for subalgebras of powers
if and only if $\al{A}$ \emph{has few subpowers}
in the following sense:
there is a polynomial $q$ depending only on $\al{A}$ such that
$\al{A}^n$ has at most $2^{q(n)}$ subalgebras. 
This number is indeed `small',
as there exist finite
algebras $\al{A}$ (e.g., semilattices or unary algebras) for which
the number of subalgebras of $\al{A}^n$ is doubly exponential in $n$.

Algebras with few subpowers
were also characterized in \cite{BI:VFSP}
by the existence
of a \emph{cube term} (to be defined in Section~\ref{sec-cube}).
Cube terms generalize \emph{Mal'tsev terms} and \emph{near unanimity terms},
which have been studied in algebra since the 60's and 70's, and
more recently have played significant roles in isolating important properties
of constraints in CSPs. 

The principal research problem 
addressed in this paper is
the following.

\smallskip
\noindent
{\bf Question 1} \cite{IM:TLAA}{\bf .}
Is the subpower membership problem for $\al A$ in $\PP$ if
$\al A$ is a finite algebra with a cube term?

\smallskip
\noindent
The answer is known to be `yes' if $\al{A}$ has a near unanimity term (by the
Baker--Pixley Theorem~\cite{BP:PIAT}), and the answer is not known
in general, if $\al{A}$ has a Mal'tsev term; in fact, as we mentioned earlier,
the answer is not known even if $\al{A}$ is a finite group expanded by
further operations.
We will now briefly discuss what the main difficulty seems to be.

Recall from~\cite{BI:VFSP} that if a finite algebra $\al{A}$
has a cube term, then for every subalgebra
$\al{B}$ of any finite power $\al{A}^n$ of $\al{A}$
\emph{there exists} a sort of canonical set of generators,
called a \emph{compact representation} in 
\cite{BD:SAMC,BI:VFSP,IM:TLAA} 
or a \emph{frame} in \cite{DR:ED}.
Compact representations 
played
a crucial role in the CSP applications of finite algebras $\al{A}$ with
cube terms (see \cite{IM:TLAA}). 
For the subpower membership problem $\SMP(\al{A})$, a natural
approach to answering the question 
``Is $b$ in the subalgebra $\al{B}$ of $\al{A}^n$ generated by 
$a_1,\dots,a_k$?'' would be to first find a
compact representation for $\al{B}$, and then use it to answer the question.  
The main components of a compact representation for $\al{B}$
are \emph{witnesses for forks.}
For a
coordinate position $i\in\{1,\dots, n\}$ and
for two elements
$c,d\in A$, the
tuples $(c_1,\dots,c_n),(d_1,\dots,d_n)$ in $\al{B}$ are said to witness
that 
the pair $(c,d)$ is a fork in $\al{B}$ in position $i$
if $c_1=d_1,\dots,c_{i-1}=d_{i-1}$ and $c_i=c,d_i=d$.
It turns out
(see, e.g., Theorem~\ref{thm-smp-cr-equiv} and
  the paragraph preceding Lemma~\ref{lm-char-cr})
that in order to
give a positive answer to Question~1 it would be sufficient 
to be able to decide in polynomial time
--- using  only the specified generators $a_1,\dots,a_k$ ---
whether or not a given pair $(c,d)\in A^2$ is a fork in $\al{B}$ in
a given position $i$ (and if it is, find witnesses for it in $\al{B}$).
This, however, appears to be a significant problem;
cf.~Theorems~\ref{thm-cr-direct}, \ref{thm-short-rep},
  and~\ref{thm-smp-cr-equiv}.

Question~1 will remain unresolved in this paper, but we
will prove (see Theorem~\ref{thm-RScase})
that the answer to Question~1 
is YES in an important special case, namely when
$\al A$ belongs to a residually small variety (to be defined in Section~\ref{sec:cm}).
For a finite algebra $\al A$ with a cube term this condition is equivalent to the requirement that there exists a
natural number $c$ such that every algebra in the variety generated by $\al A$
is a subalgebra of a product of algebras of size less than $c$.
For example, a finite algebra $\al A$ belongs to a residually small variety if it is a lattice, a module, a group with
abelian Sylow subgroups, or a commutative ring with $1$ whose Jacobson radical $J$ satisfies $J^2 =0$.
Although subpower membership for all finite groups, rings and lattices is already known to be tractable by various
algorithms, we give a new unified approach that can handle all of them and also more general algebras in residually
small varieties. So our algorithm is a first step towards finding a polynomial time algorithm for all algebras with
a cube term.

The proof of Theorem~\ref{thm-RScase} does not use compact representations; rather, it
relies on a structure theorem proved in \cite{KS:CAPT} 
for the subalgebras of finite powers of an algebra
$\al A$ with a cube term (or equivalently, parallelogram term). 
The application of this structure theorem leads us to considering subalgebras
of finite products $\al S_1\times\dots\times\al S_n$ 
where the factors $\al S_1,\dots,\al S_n$ come from
the finite collection $\HH\SSS\al A$ of homomorphic images of subalgebras
of $\al A$. 
As a consequence, it is natural for us to expand the scope of the 
{\it subpower membership problem}, and define it for any finite set 
$\class{K}$ of finite algebras 
as follows:

\medskip

\noindent $\SMP(\class{K})$:
\begin{itemize}
\item
INPUT:  
$a_1,\dots,a_k,b\in \al A_1\times\dots\times\al A_n$
with $\al A_1,\ldots,\al A_n\in\class{K}$.
\item
QUESTION: Is $b$ in the subalgebra of $\al A_1\times\dots\times\al A_n$ 
generated by $\{a_1,\dots,a_k\}$?
\end{itemize}

\smallskip
There is an
important issue, which is raised by our
passage from $\SMP(\al A)$ to $\SMP(\HH\SSS\al A)$:
Does the subpower membership problem get harder 
when $\al A$ is replaced by the set $\SSS\al A$ of all its subalgebras or 
by the set $\HH\al A$ of all its homomorphic images?
It is easy to see that for any finite algebra $\al A$, 
$\SMP(\al A)$ and $\SMP(\SSS\al A)$ are essentially the same problem.
However, this is not the case for homomorphic images.
A surprising result of M.~Steindl~\cite{St:SMPB}
shows that there exists a $10$-element semigroup $\al S$
with a $9$-element quotient semigroup $\bar{\al S}$ such that
$\SMP(\al S)$ is in $\PP$ while $\SMP(\bar{\al S})$ is $\NP$-complete.
This shows that the problem $\SMP(\HH\SSS\al A)$ may be harder
than $\SMP(\al A)$ (provided $\PP\not=\NP$), and therefore
poses the following question for us:

\smallskip
\noindent
{\bf Question 2.} 
Are the problems $\SMP(\class{K})$ and $\SMP(\HH\SSS\class{K})$
polynomial time equivalent if $\class{K}$ is a finite set of
finite algebras in a variety with a cube term?

\smallskip
\noindent
We will prove (see Theorem~\ref{thm-smp-hom}) 
that the answer to Question~2 is YES.
The proof uses the techniques of compact representations mentioned above. 
We will also show that, given a finite set of finite
algebras $\class{K}$ in a variety with a cube term,
the subpower membership problem $\SMP(\class{K})$ and 
the problem of finding compact representations for subalgebras
of products of algebras in $\class{K}$
are
polynomial
time reducible to each other
(Theorem~\ref{thm-smp-cr-equiv}), 
and the two problems are in 
$\NP$ and $\FP^\NP$, respectively
(Theorems~\ref{thm-short-rep} and~\ref{thm-smp-cr-equiv}).

 Our results also yield new information on another problem that is closely related to $\SMP$, namely the
 {\it subpower intersection problem} for a finite set 
 $\class{K}$ of finite algebras which is defined as follows:

\medskip

\noindent $\SIP(\class{K})$:
\begin{itemize}
\item
INPUT:  
$a_1,\dots,a_k,b_1,\ldots,b_\ell\in \al A_1\times\dots\times\al A_n$
with $\al A_1,\ldots,\al A_n\in\class{K}$.
\item
OUTPUT: Generators for the intersection of the subalgebras of $\al A_1\times\dots\times\al A_n$ 
generated by $\{a_1,\dots,a_k\}$ and by $\{b_1,\ldots,b_\ell\}$.
\end{itemize}

\medskip

 Assume $\class{K}$ is a finite set of finite algebras in a variety with a cube term.
 Then it turns out that $\SIP(\class{K})$ has a polynomial time reduction to  $\SMP(\class{K})$.
 Indeed, given compact representations for two subalgebras $\al{B}$, $\al{C}$ of 
 $\al{A}_1\times\dots\times\al{A}_n$ with $\al{A}_1,\dots,\al{A}_n\in\class{K}$,
 Dalmau's algorithm~\cite{IM:TLAA} yields a compact representation for $\al B\cap \al C$ in polynomial time.
 Hence $\SIP(\class{K})$ reduces to finding compact representations, which is polynomial time equivalent to
 $\SMP(\class{K})$ (Theorem~\ref{thm-smp-cr-equiv}).
 For the converse, we do not know whether $\SMP(\class{K})$ reduces to $\SIP(\class{K})$ in general.
 However, it clearly does for idempotent algebras because for any elements $a_1,\dots,a_k,b$ in 
$\al{A}_1\times\dots\times\al{A}_k$, the subalgebra generated by
$\{a_1,\dots,a_k\}$ contains $b$ if and only if that subalgebra has a nonempty intersection with $\{b\}$,
the subalgebra generated by $b$. Hence for a finite set of finite algebras $\class{K}$ in an
 idempotent variety with a cube term the problems
 $\SMP(\class{K})$, $\SIP(\class{K})$, and finding compact representations
 are all polynomial time reducible to each other.

\section{Preliminaries}
\label{sec-prelim}

In this section we summarize the concepts and facts
from universal algebra that will be 
used throughout the paper. For more details the reader is referred to 
\cite{BS:CUA}.

For every natural number $m$, we will use the notation $[m]$ for the set
$\{1,2,\dots,m\}$. The collection of
all $k$-element subsets of a set $S$ will be denoted by $\binom{S}{k}$.

\subsection{Algebras and varieties}\label{sec-algebras}

An \emph{algebraic language} is a set $F$ of function symbols, together
with a mapping $F\to \{0,1,\dots\}$ which assigns an
\emph{arity} $k_f$ to every symbol $f\in F$.
An \emph{algebraic structure} (or briefly \emph{algebra})
\emph{in the language $F$} 
is a pair $\al A=(A,F^{\al A})$ where $A$ is a nonempty
set, called
the \emph{universe} of $\al A$, and $F^{\al A}$ is a set of operations
$f^{\al{A}}\colon A^{k_f}\to A$ on $A$,
indexed by elements of $F$, called
the \emph{(basic) operations} of $\al{A}$.
We will say that $\al{A}=(A;F^{\al{A}})$ is a \emph{finite algebra}
if $A$ is finite,
and $\al{A}$ has a \emph{finite language} if $F$ is finite.

Let $\al{A}$ be an algebra in the language $F$.
If $t$ is a \emph{term} in the language $F$
(as defined in first order logic),
we will write $t(x_1,\dots,x_n)$ to indicate that all variables occurring
in $t$ are among $x_1,\dots,x_n$ (but some of the variables
$x_1,\dots,x_n$ may not occur in $t$).
We may refer to a term $t(x_1,\dots,x_n)$
as an \emph{$n$-ary term}.
Each such term $t$
induces an $n$-ary \emph{term operation}
$t^{\al{A}}\colon A^n\to A$ of $\al{A}$.
Thus, the term operations of $\al{A}$
are exactly those operations that can be obtained
from the basic operations of $\al{A}$
and from projection operations by composition. A function $g\colon A^n\to A$
is called an $n$-ary
\emph{polynomial operation} of $\al{A}$ if it has the form
$g(x_1,\dots,x_n)=t^{\al{A}}(x_1,\dots,x_n,c_1,\dots,c_k)$ for some
$(n+k)$-ary term operation $t^{\al{A}}$ and some elements $c_1,\dots,c_k$
of $\al{A}$.

As in the last two paragraphs, algebras
will usually be denoted by boldface letters, and
their universes by the 
same letters in italics.
However, we will omit the superscript $\al A$
from $F^{\al{A}}$, $f^{\al{A}}$, and $t^{\al{A}}$ whenever
there is no danger of confusion.

Terms in a language $F$ encode all computations that are possible in
any algebra $\al{A}$ of the language $F$. Since we are interested in
the complexity of certain computations, we will now discuss
what measure of complexity we will use for terms.
Assuming that the language $F$ is finite, and hence there is a constant bound
on the arities of all symbols in $F$, the complexity of a term $t$
can be captured by the number of function
symbols in $t$, or equivalently, by the number
of nodes in the term tree of $t$.
This number will be referred to as
\emph{the length of $t$}.
If such a term $t$
is short, say, is polynomially long in some 
parameter, then it
can be efficiently evaluated in the straightforward way
in any algebra, for any given input.
However, even if the term itself is long, it may admit a more concise 
representation by an \emph{$F$-circuit}, i.e., a circuit with gates whose types correspond to
the symbols in $F$. Note that a circuit may evaluate identical subterms of
$t$ only once, and pass on the value to several gates in the circuit.
This is why we will consider the number of gates in such a circuit to measure the amount of computation
needed to evaluate $t$, which
may be exponentially
smaller than the length of $t$. 
More formally, let $\crcf{F}{t}$
denote a minimal $F$-circuit 
representing $t$,
that is,
no two gates compute identical
subterms.
We define the \emph{size of $\crcf{F}{t}$} to be the
number of gates in $\crcf{F}{t}$, and will use this number to measure the time
needed to evaluate $t$.
For $F = \{f\}$ a singleton, we also abbreviate $\{f\}$-circuit to $f$-circuit. Further we will omit the reference
to $F$ and write $\crc t$ when the set of gate types is clear from context.

Let $\al A,\al B$, and $\al{A}_j$ ($j\in J$)
be algebras in the same language $F$. A mapping
$\phi\colon A\to B$
is said to be a \emph{homomorphism}
$\phi\colon\al{A}\to\al{B}$ if
\[
\phi(f^{\al A}(a_1,\dots,a_{k_f}))=f^{\al B}(\varphi(a_1),\dots,
\phi(a_{k_f}))
\quad \text{for all $f\in F$ and $a_1,\dots,a_{k_f}\in A$.}
\]
An invertible (i.e., bijective)
homomorphism is called an \emph{isomorphism}.
We say that $\al{B}$ is a \emph{homomorphic image} of $\al{A}$ if there
exists an onto homomorphism $\al{A}\to\al{B}$.
Furthermore,
we say that
$\al{B}$ is 
  a \emph{subalgebra} of $\al{A}$, and write $\al{B}\le\al{A}$,
if $B\subseteq A$ and the
inclusion map $B\to A$ is a homomorphism $\al{B}\to\al{A}$;
equivalently, $\al{B}$ is a subalgebra of $\al{A}$ iff
$B\subseteq A$ and $f^{\al{B}}$ is the restriction of
$f^{\al{A}}$ to the set $B$ for all $f\in F$.
The \emph{direct product} $\prod_{j\in J}\al{A_j}$ of the algebras $\al{A}_j$
($j\in J$) is the unique algebra $\al{A}$ with universe $A:=\prod_{j\in J} A_j$
for which all projection maps $\proj_\ell\colon A\to A_\ell$,
$(a_j)_{j\in J}\mapsto a_\ell$ ($\ell\in J$) are homomorphisms; equivalently,
$\al{A}$ is the direct product of the algebras $\al{A}_j$ ($j\in J$) iff
$\al{A}$ has universe $A:=\prod_{j\in J} A_j$, and its operations $f^{\al{A}}$
act coordinatewise, via $f^{\al{A}_j}$ in coordinate $j\in J$, for all
$f\in F$.
A subalgebra $\al{B}$ of a direct product $\prod_{j\in J}\al{A}_j$
will be called a \emph{subdirect subalgebra} if
$\proj_\ell(B)=A_\ell$ for every $\ell\in J$.  

A \emph{congruence} of $\al{A}$ is the kernel of a homomorphism
with domain $\al{A}$; equivalently, a congruence of $\al{A}$ 
is an equivalence relation on $A$
that is invariant under the
operations of $\al{A}$.
If $\vartheta$ is a congruence of $\al{A}$, then there is a unique
algebra, denoted $\al{A}/\vartheta$, whose universe is
the set 
$A/\vartheta$ of equivalence classes of $\vartheta$, and whose operations
are defined so that the natural map $\nu\colon A\to A/\vartheta$,
$a\mapsto a/\vartheta$ is a homomorphism $\al{A}\to\al{A}/\vartheta$.
Here $a/\vartheta$ denotes
the equivalence class of $\vartheta$ containing $a$.
The algebra $\al{A}/\vartheta$ is called a
\emph{quotient algebra} of $\al{A}$.
By the first
isomorphism theorem (\cite[Theorem II.6.12]{BS:CUA}), 
if $\phi\colon\al{A}\to\al{B}$
is an onto homomorphism with kernel $\vartheta$, then $\al{B}$ is isomorphic
to $\al{A}/\vartheta$.

For any algebra $\al A$, the congruences of $\al A$ form a lattice
with respect to inclusion, which is called  the
\emph{congruence lattice} of $\al A$, and
is denoted by $\Con(\al A)$. The \emph{meet} operation
$\wedge$ of
$\Con(\al A)$ is simply the set-theoretic intersection of 
congruences, and the \emph{join} operation $\vee$ is the transitive closure of 
the union of congruences.
The top and bottom elements of $\Con(\al A)$ are 
denoted $1$ and $0$, respectively.
For $\alpha,\beta\in\Con(\al{A})$ with $\alpha\le\beta$ we define the
  \emph{interval} $\interval{\alpha}{\beta}$ by
  $\interval{\alpha}{\beta}:=\{\gamma\in\Con(\al{A}):\alpha\le\gamma\le\beta\}$.

A congruence $\vartheta\in\Con(\al A)$ 
is said to be \emph{completely meet irreducible} if
for any family $\{\vartheta_j:j\in J\}$
of congruences of $\al{A}$,
$\vartheta=\bigwedge_{j\in J}\vartheta_j$
implies that $\vartheta=\vartheta_j$ for some $j\in J$.
If $\vartheta$ is completely meet irreducible
and $\vartheta\not=1$, then $\vartheta$
has a unique \emph{upper cover}
$\vartheta^+$ in $\Con(\al{A})$, that is, there is a unique
congruence $\vartheta^+\in\Con(\al{A})$
such that $\vartheta<\vartheta^+$ and
$\interval{\vartheta}{\vartheta^+}=\{\vartheta,\vartheta^+\}$.
We will use the notation $\Irr(\al A)$ for the
set of all completely
meet irreducible congruences of $\al A$, excluding $1$.
An algebra $\al A$ is said to be
\emph{subdirectly irreducible} if
$0\,(\not=1)$ is completely meet irreducible in $\Con(\al A)$.
In this case
the unique minimal non-trivial congruence $0^+$ is
called the \emph{monolith} of $\al A$.  
It follows from the second isomorphism theorem
(\cite[Theorem II.6.15]{BS:CUA})
that
the subdirectly irreducible quotients of $\al A$
are exactly the algebras $\al A/\sigma$, $\sigma\in\Irr(\al A)$.

Let $\vartheta$ be a congruence of an algebra $\al A$.
We will often write $a\equiv_{\vartheta} b$ instead of $(a,b)\in\vartheta$.
If $\al B$ is a subalgebra of $\al A$,
we will say that $\al B$ is \emph{saturated with respect
to $\vartheta$}, or $\al B$ is a \emph{$\vartheta$-saturated subalgebra}
of $\al A$, if $b\in B$ and $b\equiv_{\vartheta}a$ imply 
$a\in B$ for all $a\in A$.
In other words, $\al B$ is $\vartheta$-saturated if and only if
its universe is a union of $\vartheta$-classes of $\al A$.
For an arbitrary subalgebra $\al B$ of $\al A$ there exists a smallest
$\vartheta$-saturated subalgebra of $\al A$ that contains $\al B$, which we
denote by $\al B[\vartheta]$; the universe of $\al B[\vartheta]$
is $B[\vartheta]:=\bigcup_{b\in B}b/\vartheta$. 
Denoting the restrictions of $\vartheta$ to $\al B$ and $\al B[\vartheta]$ by
$\vartheta_{\al B}$ and $\vartheta_{\al B[\vartheta]}$, respectively,
we have by
the third isomorphism theorem (\cite[Theorem II.6.18]{BS:CUA}) 
that the map 
$\al B/\vartheta_{\al B}\to\al B[\vartheta]/\vartheta_{\al B[\vartheta]}$,
$b/\vartheta_{\al B}\mapsto b/\vartheta_{\al B[\vartheta]}(=b/\vartheta)$
is an isomorphism.

For a product $\al A_1\times\dots\times\al A_n$ of algebras and for
any set $I\subseteq[n]$, the projection homomorphism
\[
\al A_1\times\dots\times\al A_n=\prod_{i\in[n]}\al A_i\to\prod_{i\in I}\al A_i,
\qquad
(a_i)_{i\in[n]}\mapsto(a_i)_{i\in I}
\]
will be denoted by $\proj_I$. For a subalgebra $\al B$ or for an element $b$
of $\prod_{i\in[n]}\al A_i$, we will write $\al B|_I$ or $b|_I$
for $\proj_I(\al B)$ or $\proj_I(b)$, respectively.
If $I=\{j_1,\dots,j_k\}$, then 
the notation $|_{\{j_1,\dots,j_k\}}$ will be simplified
to $|_{j_1,\dots,j_k}$.

For any class $\class K$ of algebras in the same language,
$\HH\class K$, $\SSS\class K$, and $\PPP\class K$
denote the classes of all homomorphic images, subalgebras, and direct products
of members of $\class K$, respectively.

A \emph{variety} is a class of algebras in the same language that is
closed under taking direct products, subalgebras, and homomorphic 
images. By Tarski's Theorem
(\cite[Theorem II.9.5]{BS:CUA}),
the smallest variety containing a class
$\class K$ of algebras in the same 
language equals $\HH\SSS\PPP\class K$.
Birkhoff's Theorem
(\cite[Theorem 11.9]{BS:CUA}) characterizes varieties as 
classes of algebras defined by identities.
In more detail, an \emph{identity} in a
language $F$ is an expression
of the form
$s(x_1,\dots,x_n)=t(x_1,\dots,x_n)$ where $s(x_1,\dots,x_n)$ and
$t(x_1,\dots,x_n)$ are terms in
the language $F$.
An algebra $\al A$ in the
language $F$ \emph{satisfies an identity}
$s(x_1,\dots,x_n)=t(x_1,\dots,x_n)$
if the $n$-ary term operations $s^{\al{A}}$ and $t^{\al{A}}$ are equal.
Birkhoff's Theorem is the statement that for a fixed language,
a class $\class{V}$ of algebras
is a variety if and only if there exists a set $\Sigma$ of identities
such that an algebra lies in $\class{V}$ if and only if
it satisfies the identities in $\Sigma$.

Occasionally, when we want to describe sets of identities where
the identities have a similar form,
it will be convenient to represent the identities in
matrix form as follows. For any class $\class{K}$ of algebras 
in the same language, and for any term $t=t(x_1,\dots,x_n)$,
if $M$ is an $m\times n$ matrix of variables
and $\vec{v}$ is an $m\times 1$ matrix of variables,
\begin{equation}
\label{eq-matrix-id}
\class{K}\models t(M)=\vec{v}
\end{equation}
will denote that the $m$ 
identities represented by the rows in \eqref{eq-matrix-id}
are true in $\class{K}$. For example, 
\begin{equation}
\label{eq-maltsev}
\class{K}\models t
\left(
\begin{matrix}
x & x & y\\
y & x & x
\end{matrix}
\right)
=
\left(
\begin{matrix}
y\\
y
\end{matrix}
\right)
\end{equation}
expresses that 
the identities $t(x,x,y)=y$ and $t(y,x,x)=y$ hold in $\class K$,
that is, $t$
is a \emph{Mal'tsev} term for $\class{K}$.

\subsection{Cube Terms and Parallelogram Terms} \label{sec-cube}

Let us fix an integer $\ddd\ (>1)$.
A \emph{$\ddd$-cube term} for $\class{K}$ is a 
term $t$ satisfying a set of identities 
of the form \eqref{eq-matrix-id} in two variables $x,y$,
where $M$ is a matrix with $\ddd$ rows such that every
column of $M$ contains at least one $x$, and $\vec{v}$
consists of $y$'s only. As \eqref{eq-maltsev} shows,
a Mal'tsev term is a $2$-cube term.

As we mentioned in the introduction, cube
terms were introduced in \cite{BI:VFSP} to show that
a finite algebra $\al A$ has few subpowers
if and only if $\al A$
has a cube term.
More manageable terms that are equivalent to cube terms (e.g., edge terms,
star terms) were also found in \cite{BI:VFSP}.
In this paper we will use another
family of equivalent terms, called parallelogram terms,
which were introduced
in \cite{KS:CAPT}.

Let $m$ and $n$ be positive integers
and let $\ddd= m+n$. 
An {\em $(m,n)$-parallelogram term} 
for $\class{K}$ is
a $(\ddd +3)$-ary term 
$P_{m,n}$
such that 
\begin{equation}
\label{p}
\class{K}\models
P_{m,n}
\left(
\begin{array}{ccc|}
x&x&y\\
x&x&y\\
&\vdots&\\
x&x&y\\
\hline
y&x&x\\
&\vdots&\\
y&x&x\\
y&x&x\\
\end{array}
\begin{array}{cccccccc}
z&y&\cdots&y&y&\cdots&y&y\\
y&z&&y&y&&y&y\\
\vdots&&\ddots& &&&&\vdots\\
y&y&&z&y&&y&y\\
y&y&&y&z&&y&y\\
\vdots&&&&&\ddots&&\vdots\\
y&y&&y&y&&z&y\\
y&y&\cdots&y&y&\cdots&y&z\\
\end{array}
\right) = 
\left(\begin{matrix}
y\\
y\\
\vdots\\
y\\
y\\
\vdots\\
y\\
y
\end{matrix}
\right).
\end{equation}
Here the rightmost block
of variables is a $\ddd\times \ddd$ array, the upper left
block is $m\times 3$ and the lower left block is $n\times 3$.

It is easy to see from these definitions that an $(m,n)$-parallelogram term
that is independent of its last $\ddd$ variables is a Mal'tsev term, and
an $(m,n)$-parallelogram term that is independent of its first $3$ variables 
is a $\ddd$-ary near unanimity term.

The theorem below summarizes the facts we will need later on about cube terms
and parallelogram terms. 

\begin{thm}[{See \cite{BI:VFSP,KS:CAPT}}]
\label{thm-cube-par}
Let $\var{V}$ be a variety, and let $\ddd\,(>1)$ be an integer.
\begin{enumerate}
\item \label{it:cpp} 
The following conditions are equivalent:
\begin{enumerate}
\item
$\var{V}$ has a $\ddd$-cube term,
\item
$\var{V}$ has an $(m,n)$-parallelogram term for all
$m,n\ge1$ with $m+n=\ddd$,
\item
$\var{V}$ has an $(m,n)$-parallelogram term for some
$m,n\ge1$ with $m+n=\ddd$.
\end{enumerate}
\item \label{it:cubecm} 
If $\var{V}$ has a cube term, then $\var{V}$ is congruence modular.
\end{enumerate}
\end{thm}

In statement~\eqref{it:cubecm}
the phrase $\var{V}$ is \emph{congruence modular}
means that the congruence lattice of every algebra
in $\var{V}$
is \emph{modular}; see, \cite{FM:CTC}. 
The equivalence of the conditions (a)--(c)
 in statement~\eqref{it:cpp} follows by 
combining results from \cite[Theorem~4.4]{BI:VFSP} and 
\cite[Theorem~3.5]{KS:CAPT}.
Proofs for statement~\eqref{it:cubecm} can be found in
\cite[Theorem~2.7]{BI:VFSP}
and \cite[Theorem~3.2]{DKS:ETCM}.

In view of the equivalence of conditions (a) and (c) in 
Theorem~\ref{thm-cube-par}, when we consider classes of algebras
with a $\ddd$-cube term, we will work with a
$(1,\ddd-1)$-parallelogram term $P=P_{1,\ddd-1}$, and we will also use the
following terms derived from $P$:
\begin{align}\label{eq-ps}
s(x_1,\ldots,x_\ddd) &{}:=P(x_1,x_2,x_2,x_1,\ldots,x_\ddd),\\ \notag
p(x,u,y) &{}:=P(x,u,y,x,y,\ldots,y). \notag
\end{align}
For any class $\class{K}$ of algebras with a $(1,\ddd-1)$-parallelogram term $P$
one can easily deduce 
from the $(1,\ddd-1)$-parallelogram identities that
\begin{align}\label{eq-Pps}
\class{K}\ \models\ \ 
& y=p(x,x,y),\notag\\
& \phantom{y={}}p(x,y,y)=s(x,y,y,\ldots,y),\notag\\
& \phantom{y=p(x,y,y)={}}s(y,x,y,\ldots,y)=y,\\
& \phantom{y=p(x,y,y)=s(y,x,y,\ldots,y)}\,\,\,\vdots \notag\\ 
& \phantom{y=p(x,y,y)={}}s(y,y,y,\ldots,x)=y. \notag 
\end{align}
To simplify notation, we also define
\begin{equation}
  \label{eq-xy}
x^y :=p(x,y,y)\ (=s(x,y,\dots,y)),
\end{equation}
and
\begin{equation}
  \label{eq-sl}
s^\ell(x_1,\ldots,x_\ddd) := s(s^{\ell-1}(x_1,x_2\ldots,x_\ddd),x_2,\ldots,x_\ddd)
\quad
\text{for all $\ell\ge1$},
\end{equation}
where $s^0:=x_1$. So, $s^1=s$ and $s^{\ell}$ is the $\ell$-th iterate of $s$
in the first variable.

\subsection{Congruence Modular Varieties: 
the Commutator and Residual Smallness} \label{sec:cm}

Let $\var{V}$ be an arbitrary congruence modular variety.
There is a well-behaved commutator operation
$[\phantom{n},\phantom{n}]$ on the congruence lattices of algebras in
any such variety $\var{V}$, which extends
--- and shares many important properties of ---
the group theoretic commutator for normal subgroups of groups.
For the definition and basic properties of the commutator operation
in $\var{V}$ the reader is referred to \cite{FM:CTC}.
A congruence $\alpha\in\Con(\al A)$ of an algebra $\al A\in\var{V}$ is called
\emph{abelian} if $[\alpha,\alpha]=0$, and the \emph{centralizer}
of a congruence $\alpha\in\Con(\al A)$, denoted $(0:\alpha)$, 
is the largest congruence $\gamma\in\Con(\al A)$ such that 
$[\alpha,\gamma]=0$.

Recall that $\var{V}$ has a \emph{difference term} 
(see \cite[Theorem~5.5]{FM:CTC}), which we will denote by $d$.
In the last two sections of this paper we will need 
some properties of abelian congruences, which can be summarized informally
as follows: the difference term $d$ induces abelian groups
on the blocks of all abelian congruences $\alpha$ of all algebras 
$\al A\in\var{V}$; moreover, 
the term operations of $\al A$
are `linear between the blocks' of $\alpha$ with respect to these
abelian groups. 
The theorem below gives a more precise formulation of these facts.

\begin{thmC}[{\cite[Section~9]{FM:CTC}}]
\label{thm-abelian-congr}
Let $\var{V}$ be a congruence modular variety with a difference term $d$,
let $\al A\in\var{V}$, and let $\alpha$ be an abelian congruence of $\al A$.
\begin{enumerate}
\item \label{it:abelian1}
For every $o\in A$, the $\alpha$-class containing $o$ forms 
an abelian group $(o/\alpha;+_o,-_o,o)$
with zero element $o$ for the operations
$+_o$ and $-_o$ defined by
\[
\qquad
x +_o y := d(x,o,y)
\quad\text{and}\quad
-_o x := d(o,x,o)
\quad\text{for all $x,y\in o/\alpha$.}
\]
\item \label{it:abelian2} 
For every term $g(x_1,\dots,x_k)$ in the language of $\var{V}$,
for arbitrary elements $o_1,\dots,o_k,o\in A$
such that $g(o_1,\dots,o_k)\equiv_\alpha o$, and for any tuple
$(a_1,\dots,a_k)\in(o_1/\alpha)\times\dots\times(o_k/\alpha)$,
\begin{multline*}
\qquad\quad
g(a_1,a_2,\dots,a_k)=g(a_1,o_2,\dots,o_k)+_o g(o_1,a_2,o_3,\dots,o_k)+_o\dots\\
+_o g(o_1,\dots,o_{k-1},a_k)-_o (k-1)g(o_1,o_2,\dots,o_k).
\quad
\end{multline*}
\end{enumerate}
\end{thmC}  

We will refer to the abelian groups described 
in statement (1) as the \emph{induced
abelian groups} on the $\alpha$-classes of $\al A$. 

For a cardinal $c$, 
a variety $\var{V}$ is called \emph{residually less than $c$} 
if every subdirectly irreducible algebra in $\var{V}$ has cardinality
$<c$; $\var{V}$ is called \emph{residually small} if it is residually 
less than some cardinal.

\begin{thmC}[{{}\cite{FM:CTC}}]
\label{thm-rs}
If $\al A$ is
a finite algebra that generates a congruence modular variety
$\var{V}(\al A)$, then
the following conditions are equivalent:
\begin{enumerate}
\item \label{it:rs}
 $\var{V}(\al A)$ is residually small.
\item \label{it:rs2}
$\var{V}(\al A)$ is residuall less than some natural number.
\item \label{it:si}
For every subdirectly irreducible algebra $\al S\in\HH\SSS(\al A)$
with abelian monolith $\mu$, the centralizer $(0:\mu)$ of $\mu$
is an abelian congruence of $\al S$.
\end{enumerate}
\end{thmC}

\begin{proof}
By~\cite[Theorem~10.15]{FM:CTC}, \eqref{it:rs} and \eqref{it:rs2} are both equivalent to
the condition that the commutator
identity
  $[x\wedge y,y]= x\wedge[y,y]$
holds in the congruence lattice of every subalgebra of $\al A$.
The latter is equivalent to~\eqref{it:si} by~\cite[p. 422]{FM:RSV}.
\end{proof}

Now let $\class{K}=\{\al A_1,\dots,\al A_n\}$ be a finite set of finite
algebras in a congruence modular variety,
and let $\var{V}(\class{K})$ be the variety generated by $\class{K}$.
Then $\var{V}(\class{K})$ is the join of the varieties $\var{V}(\al A_i)$
for $i\leq n$.
Hence, by~\cite[Theorem 11.1]{FM:CTC},
$\var{V}(\class{K})$ is residually small iff $\var{V}(\al A_i)$
is residually small for all $i\leq n$.
Thus we get the following corollary.
 
\begin{cor}
\label{cor-rs}
Let $\class{K}$ be a finite set of finite algebras in a congruence
modular variety. 
The variety generated by $\class{K}$ is residually small if and only if
for every subdirectly irreducible algebra $\al S\in\HH\SSS(\class{K})$
with abelian monolith $\mu$, the centralizer $(0:\mu)$ of $\mu$
is an abelian congruence of $\al S$.
\end{cor}

In conclusion of this section,
we introduce a slightly technical notion, \emph{similarity},
which is an important equivalence relation on the class of subdirectly
irreducible algebras in a congruence modular variety and will
play a role in Sections~\ref{sec-subpowers}--\ref{sec-str-alg}.
However, a reader of this paper may safely regard the definition
of similarity as a `black box', since similarity only occurs as a result
of using the structure theorem from \cite{KS:CAPT}
(see Theorem~\ref{thm-paralg}), 
and in the algorithms in Section~\ref{sec-str-alg},
similarity is applied only to
a finite set of finite subdirectly irreducible algebras, therefore
checking similarity between them does not affect the complexity of
the algorithms. 

To define similarity, let $\var{V}$ be a congruence modular variety.
\emph{Similarity} is a binary relation defined on the class of subdirectly
irreducible algebras in $\var{V}$.
The definition may be found in 
\cite[Definition~10.7]{FM:CTC}, and shows that similarity is
an equivalence relation.
To decide whether two subdirectly irreducible algebras in $\var{V}$
are similar it is more convenient to use the following
characterization given in \cite[Theorem~10.8]{FM:CTC}:
two subdirectly irreducible algebras 
$\al B, \al C\in \var{V}$ are similar if and only if
\label{p-similar-alg}
there exists an algebra $\al E\in \mathcal V$ (which can be taken
to be a subdirect subalgebra of $\al B\times\al C$) and there exist
congruences $\beta, \gamma, \delta, \epsilon\in\Con(\al E)$ such that 
$\al E/\beta\cong \al B$, $\al E/\gamma\cong \al C$ and there
is a projectivity
$\interval{\beta}{\beta^+}\searrow
  \interval{\delta}{\epsilon}\nearrow\interval{\gamma}{\gamma^+}$
in $\Con(\al E)$, 
where $\beta^+$ and $\gamma^+$ are the unique upper covers
of $\beta$ and $\gamma$, respectively.
Here $\interval{\beta}{\beta^+}\searrow
  \interval{\delta}{\epsilon}\nearrow\interval{\gamma}{\gamma^+}$
  denotes that $\beta\wedge\epsilon = \delta, \beta\vee\epsilon = \beta^+$ and
  $\epsilon\wedge\gamma = \delta, \epsilon\vee\gamma = \gamma^+$.

\section{Compact Representations}
\label{sec-cRep}

Throughout this section
we will use the following global assumptions:

\begin{asm} \hfill
  \begin{itemize}
  \item
    $\var{V}$ is a variety with a $\ddd$-cube term ($\ddd>1$),
  \item
    $P$ is a $(1,\ddd-1)$-parallelogram term in $\var{V}$ (the existence
    of such a term is ensured by Theorem~\ref{thm-cube-par}), and
  \item
    $s(x_1,\ldots,x_\ddd)$, $p(x,u,y)$, $x^y$, and $s^\ell(x_1,\ldots,x_\ddd)$
    are the terms defined in \eqref{eq-ps} and \eqref{eq-xy}--\eqref{eq-sl}.
  \end{itemize}  
\end{asm}

We will not assume that the algebras we are considering are finite,
or have a finite language.
Therefore, unless finiteness is
assumed explicitly, all statements and
results hold for arbitrary algebras.

Terms in the language of $\var{V}$ which can be expressed 
using $P$ only, will be referred to as \emph{$P$-terms}.
For example, the terms $s(x_1,\ldots,x_\ddd)$, $p(x,u,y)$,
$x^y$, and $s^\ell(x_1,\ldots,x_\ddd)$ are $P$-terms. 
By a \emph{$P$-subalgebra} of an algebra $\al A\in\var{V}$
we mean a subalgebra of the reduct of $\al A$ to the language $\{P\}$.
We will say that an algebra $\al A\in\var{V}$ is \emph{$P$-generated}
by $R\subseteq A$ if $R$ is a generating set for the reduct of $\al A$ 
to the language $\{P\}$; or equivalently, if
every element of $\al A$ is of the form
$t(r_1,\dots,r_m)$ for some $m\ge0$, some elements 
$r_1,\dots,r_m\in R$, and some $m$-ary $P$-term $t$.
The $P$-subalgebra of an algebra $\al A\in\var{V}$ generated by a set
$S\,(\subseteq A)$ will be denoted by $\langle S\rangle_P$.
Recall that circuit representations of $P$-terms are called \emph{$P$-circuits},
as they use only gates of type $P$.

Now we will introduce a variant of the concept of `compact 
representation' from \cite{BI:VFSP}. One difference is that
we will use a less restrictive notion 
of `fork' than `minority fork', because we want to avoid assuming 
finiteness of the algebras considered unless finiteness is necessary for the
conclusions. 
Another difference is that we
will consider subalgebras of products of algebras,
rather than subalgebras of powers of a single algebra.
Let $\al A_1,\dots,\al A_n\in\var{V}$, and let 
$B\subseteq A_1\times\dots\times A_n$.
For $i\in[n]$ and $\gamma,\delta\in A_i$
we will say that
$(\gamma,\delta)$ is a \emph{fork in the $i$-th coordinate of $B$} if
there exist $b,\hat{b}\in B$ such that
\begin{equation}
\label{eq-fork}
b|_{[i-1]}=\hat{b}|_{[i-1]}
\quad\text{and}\quad
b|_i=\gamma,\ \ \hat{b}|_i=\delta.
\end{equation}
The set of all forks in the $i$-th
coordinate of $B$ will be denoted by $\fork_i(B)$.
Tuples $b,\hat{b}\in B$ satisfying \eqref{eq-fork}
will be referred to as
\emph{witnesses} for the fork $(\gamma,\delta)\in\fork_i(B)$. 
For each $i$ and $B$ as above and for every positive integer $e$, we define
\[
\fork^e_i(B):=\{(\gamma,\delta^{\gamma^e}):(\gamma,\delta)\in\fork_i(B)\}
\]
where $\delta^{\gamma^e}$ is a short notation for
$(\dots((\delta^\gamma)^\gamma)\dots)^\gamma$ with
$e$ occurrences of $\gamma$.
The elements of $\fork^e_i(B)$ will be called \emph{$e$-derived forks
in the $i$-th coordinate of $B$}. In the case when $e=1$ we will use the 
notation $\fork'_i(B)$ instead of $\fork^1_i(B)$, and will call the
elements of $\fork'_i(B)$ \emph{derived forks
in the $i$-th coordinate of $B$}.
The next lemma shows
that derived forks are indeed forks, and they are `transferable',
which does not hold for forks in general.

\begin{lem}
\label{lm-der-forks}
If $\al A_1,\ldots,\al A_n\in \var{V}$ and
$\al B$ is a $P$-subalgebra of 
$\al A_1\times\dots\times\al A_n$, then
\begin{enumerate}
\item \label{it:fe}
$\fork_i(B)\supseteq\fork'_i(B)\supseteq 
\dots \supseteq\fork^e_i(B)\supseteq\fork^{e+1}_i(B)\supseteq\dots$ 
for all $i\in[n]$ and $e\ge1$; moreover,
\item  \label{it:df} 
for every $(\gamma,\delta)\in\fork'_i(B)$ and for every
$b\in B$ with $b|_i=\gamma$,
there is an element $\hat{b}\in B$ such that \eqref{eq-fork} holds,
that is, $b$ and $\hat{b}$ witness that $(\gamma,\delta)\in\fork_i(B)$.
\end{enumerate} 
\end{lem}

\begin{proof}
For~\eqref{it:df} let $(\gamma,\delta)\in\fork'_i(B)$ and let
$b\in B$ with $b|_i=\gamma$. Then there exists 
$(\gamma,\beta)\in\fork_i(B)$ such that $\delta=\beta^\gamma$.
Let $c,\hat{c}\in B$ be witnesses for $(\gamma,\beta)\in\fork_i(B)$; thus,
$c|_{[i-1]}=\hat{c}|_{[i-1]}$ and $c|_i=\gamma$, 
$\hat{c}|_i=\beta$. It follows from the identities in \eqref{eq-Pps}
that for the element $\hat{b}:=p(\hat{c},c,b)\in B$ 
we have
\[
\hat{b}|_{[i-1]}=
p(\hat{c}|_{[i-1]},c|_{[i-1]},b|_{[i-1]})=b|_{[i-1]},
\] 
and
\[
\hat{b}|_i=
p(\hat{c}|_i,c|_i,b|_i)=
p(\beta,\gamma,\gamma)=\beta^\gamma=\delta.
\]
This proves~\eqref{it:df}, and also the inclusion
$\fork_i(B)\supseteq\fork'_i(B)$
in~\eqref{it:fe}. The inclusion $\fork^e_i(B)\supseteq\fork^{e+1}_i(B)$
for any $e\ge1$
follows by the same argument, using 
$\delta=\beta^{\gamma^{e+1}}=(\beta^{\gamma^e})^\gamma$ and $\beta^{\gamma^e}$
in place of $\delta=\beta^\gamma$ and $\beta$.
\end{proof}

The following `weak transitivity rule' for forks will also be
useful.

\begin{lem}
\label{lm-newfork}
Let $\al A_1,\dots,\al A_n\in\var{V}$, and let $\al B$
be a $P$-subalgebra of $\al A_1\times\dots\times\al A_n$.
If $(\gamma,\delta)$ and 
$(\beta,\delta)$ are forks in $\fork_i(B)$
witnessed by the pairs
$(v,\hat{v})$ and $(u,\hat{u})$ in $\al B$,
respectively, 
then the pair
\begin{equation}
\label{eq-newfork}
\Bigl( p(p(v,\hat{v},\hat{u}),p(v,\hat{v},\hat{v}),v),\ \ p(u,v,v)\Bigr)
\end{equation}
in $\al B$
is a witness for the fork $(\gamma,\beta^\gamma)\in\fork_i(B)$.
\end{lem}

\begin{proof}
The choice of $u,\hat{u},v,\hat{v}$ implies that $u|_{[i-1]}=\hat{u}|_{[i-1]}$,
$v|_{[i-1]}=\hat{v}|_{[i-1]}$, and $u|_i=\beta$, $\hat{u}|_i=\delta=\hat{v}|_i$,
$v|_i=\gamma$. Hence,
\begin{multline*}
p(p(v,\hat{v},\hat{u}),p(v,\hat{v},\hat{v}),v)|_{[i-1]}\\
=p(p(v|_{[i-1]},\hat{v}|_{[i-1]},\hat{u}|_{[i-1]}),p(v|_{[i-1]},\hat{v}|_{[i-1]},\hat{v}|_{[i-1]}),
                                         v|_{[i-1]})\\
=p(u|_{[i-1]},v|_{[i-1]},v|_{[i-1]})=p(u,v,v)|_{[i-1]}
\end{multline*}
and
\begin{align*}
p(p(v,\hat{v},\hat{u}),p(v,\hat{v},\hat{v}),v)|_i
&{}=p(p(\gamma,\delta,\delta),p(\gamma,\delta,\delta),\gamma)=\gamma,\\
p(u,v,v)|_i
&{}=p(\beta,\gamma,\gamma)=\beta^\gamma.
\end{align*}
Clearly, $u,\hat{u},v,\hat{v}\in\al B$
implies that the pair \eqref{eq-newfork}
also lies in $\al{B}$,
so the proof of the lemma is complete.
\end{proof}

\begin{defi}
\label{df-rep}
For two sets $B,R\subseteq A_1\times\dots\times A_n$,
we will say that $R$ is a \emph{$(\ddd,e)$-representation for $B$}
if the following three conditions are met:
\begin{enumerate}
\item \label{it:B} 
$R\subseteq B$;
\item \label{it:I} 
$R|_I=B|_I$ for all $I\subseteq[n]$ with $|I|<\ddd$;
\item \label{it:Fi} 
$\fork_i(R)\supseteq\fork^e_i(B)$ for all $i\in[n]$.
\end{enumerate}
If $e=1$ and the parameter $\ddd$ of the cube term of $\var{V}$
is clear from the context, then reference to $(\ddd,e)$
will be omitted.
\end{defi}

In the special case when $\al A_1,\dots,\al A_n$ are members of
a fixed finite set $\class{K}$ of finite algebras
in $\var{V}$, and the maximum size of an algebra in $\class{K}$
is $\aaa_{\class{K}}$, then
it is easy to see that every set $B\subseteq A_1\times\dots\times A_n$
has a $(\ddd,e)$-representation $R$ of size
\[
|R|\le\sum_{I\subseteq[n],|I|<d}\bigl|B|_I\bigr|+\sum_{i\in[n]}\fork^e_i(B)
\le\binom{n}{\ddd-1}\aaa_{\class{K}}^{\ddd-1}+2n\aaa_{\class{K}}^2.
\]
A $(\ddd,e)$-representation $R$ for $B$ of size
$|R|\le \binom{n}{\ddd-1}\aaa_{\class{K}}^{\ddd-1}+2n\aaa_{\class{K}}^2$
is called a 
\emph{compact $(\ddd,e)$-representation} for $B$.

\begin{rem}
  \label{rm-minfork}
If $\class{K}$ is a finite set of finite algebras in $\var{V}$, then
there exists an $e\ge1$ such that $\class{K}\models (x^{y^e})^{y^e}=x^{y^e}$.
For such an $e$, the terms $d(x,y)=y^{x^e}$, 
$p(z,y,x)^{x^{e-1}}$ and $s^e(x_1,\dots,x_\ddd)$ satisfy the identities
in \cite[Lemma~2.13]{BI:VFSP}. 
Consequently, for this $e$,
our $e$-derived forks and
(compact) $(\ddd,e)$-representations 
are exactly the minority forks (called minority indices), 
and the (compact) representations with minority forks defined in 
\cite[Definitions~3.1--3.2]{BI:VFSP}.
\end{rem}

It was proved in \cite[Lemma 3.4, Theorem~3.6]{BI:VFSP} that for any
subalgebra $\al B$ of a finite power of a finite algebra $\al A\in\var{V}$,
a compact representation (with minority forks) generates $\al B$.  
The next theorem
is essentially the same result, with several significant 
differences, which we will discuss 
in Remark~\ref{rm-differences} below.

\begin{thm}
\label{thm-rep-generates}
Let $e$ be a fixed positive integer.
For every integer $n\ge\ddd$
there exists a $P$-term $T_n$ satisfying
the following conditions:
\begin{enumerate}
\item \label{it:repgen1}
  If $\al B$ is a $P$-subalgebra of a product 
$\al A_1\times\dots\times\al A_n$ of finitely many algebras
 $\al A_1,\dots,\al A_n\in\var{V}$,
 and $R$ is a $(\ddd,e)$-representation for $\al B$, then
 every element of $\al B$ is produced by a single application 
 of $T_n$ to some elements of $R$.
 Consequently, $\al B$ is $P$-generated by $R$.
\item \label{it:repgen2}
  The size of the $P$-circuit
  $\crc{T_n}$ is $O(en^{\ddd+1})$,
 and there is an algorithm that runs in time $O(en^{\ddd+1})$,
 which outputs $\crc{T_n}$ for any given $e$ and $n$.
\end{enumerate}
\end{thm}

Note that the $P$-terms $T_n$ depend on the parameter $e$, but to simplify
notation, we decided to suppress this dependence in the notation
of $T_n$ (and in the notation of the terms
needed to build up $T_n$, see Lemma~\ref{lm-terms}).
In this paper we will use Theorem~\ref{thm-rep-generates} for
$e=1$ and $e=2$ only.

\begin{rem}
  \label{rm-differences}
The main differences between 
\cite[Lemma 3.4, Theorem~3.6]{BI:VFSP} and
Theorem~\ref{thm-rep-generates} are as follows:
\begin{itemize}
\item
As we explained at the beginning of this section, we consider
subalgebras of direct products of finitely many different algebras 
rather than just subalgebras of finite direct powers
of the same algebra.
In addition, we don't assume that the algebras are finite,
therefore we use a less restrictive notion of a fork than
the notion used in \cite{BI:VFSP}, cf.~Remark~\ref{rm-minfork}.
\item
Given $e$, we produce a $P$-term $T_n$ for every $n$,
which describes a uniform way for obtaining
each element of a $P$-subalgebra $\al B$ of
a product $\al A_1\times\dots\times\al A_n$ in one step
from members of an arbitrary $(\ddd,e)$-representation for $\al{B}$.  
Since we are interested in algorithmic applications, 
the advantage of finding these terms $T_n$ explicitly is that ---
via the minimal circuit representation of $T_n$ --- we can
get a useful upper bound on the length of computations needed to generate 
elements of $\al B$ from 
members of a $(\ddd,e)$-representation for $\al B$.
\end{itemize}
\end{rem}

First we will prove the following lemma.

\begin{lem}
\label{lm-terms}
Let $e$ be a fixed positive integer.
For every integer $n\ge \ddd$ there exists a $P$-term
$t_n=t_n\bigl(x,y,z,\bar{w_I}\bigr)$,
where
$\bar{w_I}:=(w_I)_{I\in\binom{[n]}{\ddd-1}}$
is a tuple of variables indexed by
all $(\ddd-1)$-element subsets of $[n]$,
such that the following hold:
\begin{enumerate}
\item\label{it:terms1}
For every subset $R$ of a product 
$\al A_1\times\dots\times\al A_n$
with $\al A_1,\dots,\al A_n\in\var{V}$, and for every
tuple
$b=(b_1,\dots,b_{n-1},\gamma)$ in $\al A_1\times\dots\times\al A_n$,
if
\begin{enumerate}
\item \label{itbI}
for each $I\in\binom{[n]}{\ddd-1}$ the set $R$ contains a tuple
$b^I$ 
satisfying
$b^I|_I=b|_I$, and
\item \label{itfm}
for some element $b'=(b_1,\dots,b_{n-1},\beta)$ of
the $P$-subalgebra
$\al R^*:=\langle R\rangle_P$
of $\al A_1\times\dots\times\al A_n$,
the set $R$ contains tuples
$u=(u_1,\dots,u_{n-1},\gamma)$ and $\hat{u}=(u_1,\dots,u_{n-1},\beta^{\gamma^e})$ 
which are witnesses for the fork
 $(\gamma,\beta^{\gamma^e})\in\fork_n(R)$,
\end{enumerate}
then 
\begin{equation}
\label{eq-paral}
b=t_n\bigl(b',\hat{u},u,\bar{b^I}\bigr)
\quad\text{where}\quad \bar{b^I}:=(b^I)_{I\in\binom{[n]}{\ddd-1}},
\end{equation}
and therefore $b$ is in $\al R^*$.
\item
  \label{it:terms2}
The size of the $P$-circuit $\crc {t_n}$ is $O(en^\ddd)$, and
there is an algorithm that runs in time $O(en^\ddd)$, which
outputs $\crc{t_n}$ for any given $e$ and $n$.
\end{enumerate}
\end{lem}

\begin{proof}
We use a dynamic programming approach to build the tuple
$b$ and keep a record of the sequence of operations we perform so that we 
obtain the term $t_n$ in the end. For every $\ell\in[n]\setminus[\ddd-2]$ and
for every set $V'\in\binom{[\ell]}{\ddd-1}$ we construct
a $P$-term $t_{\ell,V'}$ 
which `approximates'
$t_n$ in the sense that for $t_{\ell,V'}$ in place of $t_n$
the equality \eqref{eq-paral} holds in all coordinates in 
$V'\cup([n]\setminus[\ell])$ (but may fail in other coordinates).

\begin{clm}
\label{clm-terms}
For $\ell=n,n-1,\dots,\ddd-1$ and for every set $V:=V'\cup([n]\setminus[\ell])$
with $V'\in\binom{[\ell]}{\ddd-1}$ 
there exists a $P$-term $t_{\ell,V'}=t_{\ell,V'}(x,y,z,\bar{w_I}^V)$ with
$\bar{w_I}^V:=(w_I)_{I\in\binom{V}{\ddd-1}}$
such that whenever the algebras $\al A_1,\dots,\al A_n$, 
$\al R^*$ and the elements
$b$, $b^I$ $\bigl(I\in\binom{[n]}{\ddd-1}\bigr)$, $b'$, $u$, $\hat{u}$
  satisfy the assumptions of Lemma~\ref{lm-terms}~\eqref{it:terms1},
  we have that
  \[
  b|_V = t_{\ell,V'}\bigl(b',\hat{u},u,\bar{b^I}^V\bigr)|_V
  \quad
  \text{where}\quad \bar{b^I}^{V}:=(b^I)_{I\in\binom{V}{\ddd-1}}.
  \]
\end{clm}

\begin{proof}[Proof of Claim~\ref{clm-terms}.]
We proceed by induction on $n-\ell$.
For $n=\ell$ we have that $V=V'\in\binom{[n]}{\ddd-1}$.
So we can choose $t_{n,V'}(x,y,z,w_V):=w_V$, and our claim is
trivial.

Assume now that $\ell<n$ and that our
claim is true for $\ell+1$.
To prove the claim for $\ell$, let
$V=V'\cup([n]\setminus[\ell])$ with
$V'=\{i_1,\dots,i_{\ddd-1}\}\in\binom{[\ell]}{\ddd-1}$,
$i_1<\dots<i_{\ddd-1}$.
For $j\in [\ddd-1]$ let
$V'_j=\{i_1,\dots,i_{j-1},i_{j+1},\dots,i_{\ddd-1},\ell+1\}$ and
$V_j := V\setminus\{i_j\} = V'_j\cup ([n]\setminus[\ell+1])$.
We will prove that the $P$-term
\begin{multline}
   \label{eq-t-V}
t_{\ell,V'}\bigl(x,y,z,\bar{w_I}^V\bigr) :=
P\Bigl(s^{e+1}\bigl(x,(t_{\ell+1,V'_j}(x,y,z,\bar{w_I}^{V_j}))_{j\in[\ddd-1]}\bigr), \\
p\bigl(y,z,t_{\ell+1,V'_1}(x,y,z,\bar{w_I}^{V_1})\bigr),
t_{\ell+1,V'_1}(x,y,z,\bar{w_I}^{V_1}), 
x,
\bigl(t_{\ell+1,V'_j}(x,y,z,\bar{w_I}^{V_j})\bigr)_{j\in[\ddd-1]}\Bigr)
\end{multline}
has the desired properties,
where the $P$-terms $t_{\ell+1,V_j'}$ ($j\in[\ddd-1]$) are supplied
by the induction hypothesis.
To prove this, let
$\al A_1,\dots,\al A_n$, $\al R^*$, and 
$b$, $b^I$ $\bigl(I\in\binom{[n]}{\ddd-1}\bigr)$, $b'$, $u$, $\hat{u}$
  satisfy the assumptions of Lemma~\ref{lm-terms}~\eqref{it:terms1}.
By the induction hypothesis,
the elements
$b^{\ell+1,V'_j}:=t_{\ell+1,V'_j}\bigl(b',\hat{u},u,\bar{b^I}^{V_j}\bigr)$ 
satisfy the condition
$b^{\ell+1,V'_j}|_{V_j}=b|_{V_j}$ for all $j\in[\ddd-1]$;
that is, $b^{\ell+1,V'_j}|_V$ has the form 
\[
b^{\ell+1,V'_j}|_V=(b_{i_1},\dots,b_{i_{j-1}},\zeta_j,b_{i_{j+1}},\dots,b_{i_{\ddd-1}},
b_{\ell+1},\dots,b_{n-1},\gamma)
\]
for some element $\zeta_j\in A_{i_j}$ in the $j$-th coordinate.
To compute 
$t_{\ell,V'}(b',\hat{u},u,\bar{b^I}^V)|_V$,
we start with
evaluating the first two arguments of $P$
in the expression for $t_{\ell,V'}$ in \eqref{eq-t-V}.

Using the identities for $s$ in \eqref{eq-Pps} we get that
\begin{multline*}
s\bigl(b',\bigl(t_{\ell+1,V'_j}(b',\hat{u},u,\bar{b^I}^{V_j})\bigr)_{j\in[\ddd-1]}\bigr)\big|_V
=s(b'|_V,b^{\ell+1,V'_1}|_V,\ldots,b^{\ell+1,V'_{\ddd-1}}|_V) \\
=s\left(
\begin{matrix}
b_{i_1} &  \zeta_1  & b_{i_1} & \dots & b_{i_1}\\
b_{i_2} & b_{i_2} &  \zeta_2  & \dots & b_{i_2}\\
\vdots & \vdots & \vdots & \ddots & \vdots\\
b_{i_{\ddd-1}} & b_{i_{\ddd-1}} & b_{i_{\ddd-1}} &  \dots & \zeta_{\ddd-1} \\ 
b_{\ell+1} & b_{\ell+1} & b_{\ell+1}  &   \dots & b_{\ell+1}\\
\vdots & \vdots & \vdots &   &  \vdots\\
b_{n-1} & b_{n-1} & b_{n-1} & \dots & b_{n-1}\\ 
\beta  & \gamma  &  \gamma  &  \dots &  \gamma 
\end{matrix}
\right)
=\left(
\begin{matrix}
b_{i_1}\\
b_{i_2}\\
\vdots\\
b_{i_{\ddd-1}}\\
b_{\ell+1}\\
\vdots\\
b_{n-1}\\
\beta^\gamma
\end{matrix}
\right).
\end{multline*}
By repeating the same computation $e$ more times so that every time
the tuple just obtained is placed in the first argument of $s$ in the next 
computation, we obtain that
\begin{multline}
\label{eq-s2-calc}
s^{e+1}\bigl(b',\bigl(t_{\ell+1,V'_j}(b',\hat{u},u,\bar{b^I}^{V_j})\bigr)_{j\in[\ddd-1]}\bigr)\big|_V
=s^{e+1}(b'|_V,b^{\ell+1,V'_1}|_V,\ldots,b^{\ell+1,V'_{\ddd-1}}|_V) \\
=\left(
\begin{matrix}
b_{i_1}\\
b_{i_2}\\
\vdots\\
b_{i_{\ddd-1}}\\
b_{\ell+1}\\
\vdots\\
b_{n-1}\\
\beta^{\gamma^{e+1}}
\end{matrix}
\right).
\end{multline}
The identities for $p$ in \eqref{eq-Pps} yield that
\begin{multline}
\label{eq-p-calc}
p\bigl(\hat{u},u,t_{\ell+1,V'_1}(b',\hat{u},u,\bar{b^I}^{V_1})\bigr)|_V
=p\bigl(\hat{u}|_V,u|_V,b^{\ell+1,V'_1}|_V\bigr)\\
=p\left(
\begin{matrix}
u_{i_1} &  u_{i_1}  & \zeta_1 \\
u_{i_2} &  u_{i_2} &  b_{i_2}\\
\vdots & \vdots & \vdots\\
u_{i_{\ddd-1}} & u_{i_{\ddd-1}} & b_{i_{\ddd-1}}\\ 
u_{\ell+1} & u_{\ell+1} & b_{\ell+1} \\
\vdots & \vdots & \vdots \\
u_{n-1} & u_{n-1} & b_{n-1} \\ 
\beta^{\gamma^e}  & \gamma  &  \gamma 
\end{matrix}
\right)
=\left(
\begin{matrix}
\zeta_1\\
b_{i_2}\\
\vdots\\
b_{i_{\ddd-1}}\\
b_{\ell+1}\\
\vdots\\
b_{n-1}\\
\beta^{\gamma^{e+1}}
\end{matrix}
\right).
\end{multline}
Combining~\eqref{eq-s2-calc} and~\eqref{eq-p-calc} with the definition of $t_{\ell,V'}$,
and using the $(1,\ddd-1)$-parallelogram identities, we obtain that
\begin{multline*}
t_{\ell,V'}(b',\hat{u},u,\bar{b^I}^V)|_V\\
=P\left(
\begin{matrix}
b_{i_1} & \zeta_1 & \zeta_1 & b_{i_1} &  \zeta_1  & b_{i_1} & \dots & b_{i_1}\\
b_{i_2} & b_{i_2} & b_{i_2} & b_{i_2} & b_{i_2} &  \zeta_2  & \dots & b_{i_2}\\
\vdots & \vdots & \vdots & \vdots & \vdots & \vdots & \ddots &
\vdots\\
b_{i_{\ddd-1}} & b_{i_{\ddd-1}} & b_{i_{\ddd-1}} & b_{i_{\ddd-1}} & b_{i_{\ddd-1}} & b_{i_{\ddd-1}} & 
                                                        \dots & \zeta_{\ddd-1} \\ 
b_{\ell+1} & b_{\ell+1} & b_{\ell+1} & b_{\ell+1} & b_{\ell+1} & b_{\ell+1} & 
                                                        \dots & b_{\ell+1}\\
\vdots & \vdots & \vdots & \vdots & \vdots & \vdots &   &
                                                            \vdots\\
b_{n-1} & b_{n-1} & b_{n-1} & b_{n-1} & b_{n-1} & b_{n-1} & \dots & b_{n-1}\\ 
\beta^{\gamma^{e+1}}  & \beta^{\gamma^{e+1}}  & \gamma  &  \beta  &  \gamma  &  \gamma  
                                                        & \dots &  \gamma 
\end{matrix}
\right)=
\left(
\begin{matrix}
b_{i_1}\\
b_{i_2}\\
\vdots\\
b_{i_{\ddd-1}}\\
b_{\ell+1}\\
\vdots\\
b_{n-1}\\
\gamma
\end{matrix}
\right)=b|_V.
\end{multline*}
This completes the proof of
Claim~\ref{clm-terms}.
 \renewcommand{\qedsymbol}{$\diamond$}
\end{proof}

The term $t_n:=t_{\ddd-1,[\ddd-1]}$ constructed in Claim~\ref{clm-terms}
for 
\[
V=[n]=\{1,\ldots,\ddd-1\}\cup([n]\setminus[\ddd-1])
\] 
clearly satisfies
the requirement in statement~\eqref{it:terms1} of Lemma~\ref{lm-terms}.
 
To estimate
the size of $\crc{t_n}$,
note that the term $t_n$ has $\sum_{\ell=\ddd-1}^{n-1}\binom{\ell}{\ddd-1}=O(n^\ddd)$
  distinct subterms
of the form
$t_{\ell,V'}$ where $\ell\in\{\ddd-1,\dots,n-1\}$
and $V'\in\binom{[\ell]}{\ddd-1}$.
By \eqref{eq-t-V}, obtaining each
one of these terms
$t_{\ell,V'}$ from variables and subterms
of the form
$t_{\ell+1,W}$ requires one
application of each of $s^{e+1}$, $p$, and $P$,
that is, altogether $e+3$ applications of $P$.
Hence $\crc{t_n}$ has size
$O(en^{\ddd})$. This proves the first part of statement~\eqref{it:terms2}
of Lemma~\ref{lm-terms}. The second part is an immediate consequence
of this estimate and the definition of $t_n$.
\end{proof}

Now let $n\ge\ddd$, and let $t_n=t_n(x,y,z,\bar{w_I})$ be the $P$-term
from Lemma~\ref{lm-terms}.
The analogous $P$-terms for
$\ddd\le m\le n$ are $t_m=t_m(x,y,z,\bar{w_I}^{[m]})$ with
$\bar{w_I}^{[m]}:=(w_I)_{I\in\binom{[m]}{\ddd-1}}$.
We will use these terms to define new $P$-terms
\[
T_m(z^{(\ddd)},\hat{z}^{(\ddd)},\dots,z^{(m)},\hat{z}^{(m)},\bar{w_I}^{[m]})
\]
for each $m=\ddd-1,\ddd,\dots,n$ by recursion as follows:
$T_{\ddd-1}(w_{[\ddd-1]}):=w_{[\ddd-1]}$, and
for all $m$ with $\ddd\le m\le n$,
\begin{multline*}
T_{m}(z^{(\ddd)},\hat{z}^{(\ddd)},\dots,z^{(m)},\hat{z}^{(m)},\bar{w_I}^{[m]})\\
{}:= 
t_{m}\bigl(T_{m-1}(z^{(\ddd)},\hat{z}^{(\ddd)},\dots,z^{(m-1)},\hat{z}^{(m-1)},
\bar{w_I}^{[m-1]}), z^{(m)}, \hat{z}^{(m)}, \bar{w_I}^{[m]}\bigr).
\end{multline*}
In particular, for $m=n$, the term is 
$T_n(z^{(\ddd)},\hat{z}^{(\ddd)},\dots,z^{(n)},\hat{z}^{(n)},\bar{w_I})$, because
$\bar{w_I}^{[n]}=\bar{w_I}$.

Note that by this
definition, $T_n$ has
$n-\ddd+1$ subterms
of the form $T_m$ ($m=\ddd,\dots,n$) other than the variable
$T_{\ddd-1}$.
Furthermore, each $T_m$
is obtained from $T_{m-1}$ by applying $t_m$
to $T_{m-1}$ and some variables.
Since by Lemma~\ref{lm-terms}, the
$P$-circuit $\crc{t_m}$ has size $O(em^\ddd)$,
we get that
\begin{equation} \label{eq:cTn}
\text{the $P$-circuit }  \crc{T_n} \text{ has size }
O(en^{\ddd+1}).
\end{equation}

\def\itbI{{\rm (a)}}
\def\itfm{{\rm (b)}}

\begin{lem}
\label{lm-Ts}
Let $e$ be a fixed positive integer.
Assume that $\al A_1,\dots,\al A_n\in\var{V}$ $(n\ge\ddd)$,
and let $b$ be
an element and $R$ a subset of
$\al A_1\times\dots\times\al A_n$ such that
\begin{enumerate}
\item[\itbI]
for each $I\in\binom{[n]}{\ddd-1}$ the set $R$
contains a tuple
$b^I$ satisfying $b^I|_I=b|_I$, 
and
\item[\itfm]
for every $m$ with $\ddd\le m\le n$ the set $R$ contains
tuples
$u^{(m)}, \hat{u}^{(m)}$ which are witnesses for the fork
$(\gamma,\beta^{\gamma^e})\in\fork_m(R)$ 
where
\begin{equation}
\label{eq-betagamma}
\gamma:=b|_m\quad\text{and}\quad 
\beta:=
T_{m-1}(u^{(\ddd)},\hat{u}^{(\ddd)},\dots,u^{(m-1)},\hat{u}^{(m-1)},
\bar{b^I}^{[m-1]})|_m.
\end{equation}
\end{enumerate}
Then the following equalities hold:
\begin{equation}
\label{eq-rep-bm}
b|_{[m]}=
T_m(u^{(\ddd)},\hat{u}^{(\ddd)},\dots,u^{(m)},\hat{u}^{(m)},
\bar{b^I}^{[m]})|_{[m]}
\quad
\text{for all $m=\ddd-1,\ddd,\dots,n$;}
\end{equation}
in particular,
\begin{equation}
\label{eq-rep-b}
b=
T_n(u^{(\ddd)},\hat{u}^{(\ddd)},\dots,u^{(n)},\hat{u}^{(n)},
\bar{b^I}).
\end{equation}
\end{lem}

\begin{proof}
Let $\al R^*$ denote the $P$-subalgebra of 
$\al A_1\times\dots\times\al A_n$
generated by $R$.
For $m=\ddd-1,\ddd,\dots,n$ let 
\[
b^{(m)}:=
T_m(u^{(\ddd)},\hat{u}^{(\ddd)},\dots,u^{(m)},\hat{u}^{(m)},
\bar{b^I}^{[m]}).
\]
We will proceed by induction to prove that \eqref{eq-rep-bm} holds,
that is, $b|_{[m]}=b^{(m)}|_{[m]}$ for all $m=\ddd-1,\ddd,\dots,n$.
Then, for the case when $m=n$, \eqref{eq-rep-bm} yields the equality 
\eqref{eq-rep-b}.

To start the induction, let $m:=\ddd-1$.
Since $T_{\ddd-1}(w_{[\ddd-1]})=w_{[\ddd-1]}$, we have that
$b^{(\ddd-1)}=b^{[\ddd-1]}$, so 
$b|_{[\ddd-1]}=b^{(\ddd-1)}|_{[\ddd-1]}$ is clearly true by assumption \itbI.

Now assume that $m\ge\ddd$ and that the equality 
$b|_{[m-1]}=b^{(m-1)}|_{[m-1]}$ holds.
Then
the tuples $b|_{[m]}\in\al A_1\times\dots\times\al A_m$ and
$b^{(m-1)}|_{[m]}\in\al R^*|_{[m]}$ have the form $(b_1,\dots,b_{m-1},\gamma)$
and $(b_1,\dots,b_{m-1},\beta)$, respectively, 
where $\gamma$ and $\beta$ are defined by \eqref{eq-betagamma}.
Assumption~\itfm\  implies that $R$ contains tuples
$u^{(m)},\hat{u}^{(m)}\in R$ which are witnesses for the fork
$(\gamma,\beta^{\gamma^e})\in\fork_m(R)$.
Then $u^{(m)}|_{[m]}$ and $\hat{u}^{(m)}|_{[m]}$ are in $R|_{[m]}$,
and they have the form
$u^{(m)}|_{[m]}=(u_1,\dots,u_{m-1},\gamma)$ and 
$\hat{u}^{(m)}|_{[m]}=(u_1,\dots,u_{m-1},\beta^{\gamma^e})$ 
for some $u_i\in\al A_i$ ($i\in[m-1]$). 
The fact that $R$ satisfies assumption \itbI\ also implies that
the elements $b^I|_{[m]}\in R|_{[m]}$
have the property  
$(b^I|_{[m]})|_I=(b|_{[m]})|_I$ for all $I\in\binom{[m]}{\ddd-1}$.

This shows that the assumptions
\itbI--\itfm\ of Lemma~\ref{lm-terms}~\eqref{it:terms1}
hold for
\begin{itemize}
\item
the subset
$R|_{[m]}$ and element
$b|_{[m]}=(b_1,\dots,b_{m-1},\gamma)$ of $\al A_1\times\dots\times\al A_m$, 
\item
the element
$b^{(m-1)}|_{[m]}
=(b_1,\dots,b_{m-1},\beta)$ in $\al R^*|_{[m]}$, and
\item
the elements
$u^{(m)}|_{[m]}=(u_1,\dots,u_{m-1},\gamma)$, 
$\hat{u}^{(m)}|_{[m]}=(u_1,\dots,u_{m-1},\beta^{\gamma^e})$, and
$b^I|_{[m]}$~$\bigl(I\in\binom{[m]}{\ddd-1}\bigr)$ 
in $R|_{[m]}$.
\end{itemize}
Thus, Lemma~\ref{lm-terms}~\eqref{it:terms1}
--- combined with the definitions of $b^{(m-1)}$ and $T_m$
--- implies that
\begin{align*}
b|_{[m]}
&{}=t_m(b^{(m-1)}|_{[m]},\hat{u}^{(m)}|_{[m]},u^{(m)}|_{[m]},\bar{b^I|_{[m]}}^{[m]})\\
&{}=t_m(b^{(m-1)},\hat{u}^{(m)},u^{(m)},\bar{b^I}^{[m]})|_{[m]}\\
&{}=t_m\bigl(T_{m-1}(u^{(\ddd)},\hat{u}^{(\ddd)},\dots,u^{(m-1)},\hat{u}^{(m-1)},
\bar{b^I}^{[m-1]}),\hat{u}^{(m)},u^{(m)},\bar{b^I}^{[m]}\bigr)|_{[m]}\\
&{}=T_m(u^{(\ddd)},\hat{u}^{(\ddd)},\dots,u^{(m)},\hat{u}^{(m)},
\bar{b^I}^{[m]})|_{[m]},
\end{align*}
which is what we wanted to prove.
\end{proof}

\begin{lem}\label{lm-rep-generates}
Let $e$ be a fixed positive integer, and let
$\al B$ be a $P$-subalgebra of a product 
$\al A_1\times\dots\times\al A_n$ of finitely many algebras
$\al A_1,\dots,\al A_n\in\var{V}$.
If $R$ is a $(\ddd,e)$-representation for $\al B$, then
for every element $b$ of $\al{B}$,
\begin{enumerate}
  \item
there exist elements
$b^I$ $(I\in\binom{[n]}{\ddd-1})$ and
$u^{(m)},\hat{u}^{(m)}$ $(m=\ddd,\dots,n)$ in $R$ such that
conditions~{\itbI--\itfm} of Lemma~\ref{lm-Ts} are satisfied,
\item
  hence the equality
  $b=T_n(u^{(\ddd)},\hat{u}^{(\ddd)},\dots,u^{(n)},\hat{u}^{(n)},\bar{b^I})$
  in \eqref{eq-rep-b} holds for $b$.
\end{enumerate}
\end{lem}

\begin{proof}
  Let $R$ be a $(\ddd,e)$-representation for $\al B$, and let $b\in B$.
As before, let $\al{R}^*$ denote the $P$-subalgebra of
    $\al{A}_1\times\dots\times\al{A_n}$ generated by $R$.
    Our statements
  will follow from
  Lemma~\ref{lm-Ts}
  if we show that
  $b$ and $R$ satisfy the assumptions~\itbI--\itfm\ 
  of that lemma.  
Condition~\itbI\  clearly follows from our assumptions that 
$b\in B$ and $R$ is a $(\ddd,e)$-representation for $\al B$.
To verify condition~\itfm\  we proceed by induction on $m$
to show that for every $m$ ($\ddd\le m\le n$) 
\begin{enumerate}[labelindent=4mm]
\item[\itfm$_m$]
$R$ contains tuples
$u^{(m)}, \hat{u}^{(m)}$ witnessing the fork
$(\gamma,\beta^{\gamma^e})\in\fork_m(R)$ 
where $\gamma$ and $\beta$ are defined by \eqref{eq-betagamma}.
\end{enumerate}

Let $\ddd\le m\le n$, and assume that 
condition \itfm$_{i}$ holds for $i=\ddd,\dots,m-1$;
note that this assumption is vacuously true for the base case
$m=\ddd$. 
Our goal is to show that~\itfm$_m$ also holds.
Let
\[
b^{(m-1)}:=T_{m-1}(u^{(\ddd)},\hat{u}^{(\ddd)},\dots,u^{(m-1)},\hat{u}^{(m-1)},
\bar{b^I}^{[m-1]}),
\]
and let $\gamma$ and $\beta$ be defined by \eqref{eq-betagamma};
that is, $\gamma=b|_m$ and $\beta=b^{(m-1)}|_m$.
Since $b^{(m-1)}$ involves only the elements 
$b^I$ $\bigl(I\in\binom{[m-1]}{\ddd-1}\bigr)$ and
$u^{(\ddd)},\hat{u}^{(\ddd)},\dots,u^{(m-1)},\hat{u}^{(m-1)}$ of $R$,
and since our induction hypothesis ensures that these elements
satisfy the assumptions of Lemma~\ref{lm-Ts},
we get from Lemma~\ref{lm-Ts}
that $b|_{[m-1]}=b^{(m-1)}|_{[m-1]}$.
Here $b\in B$ and $b^{(m-1)}\in R^*\subseteq B$, therefore
$b$ and $b^{(m-1)}$ are witnesses in $\al B$ for the fork 
$(\gamma,\beta)\in\fork_m(B)$.
Hence, by Lemma~\ref{lm-der-forks},
$(\gamma,\beta^{\gamma^e})\in\fork^e_m(B)$. 
So, the fact that
$R$ is a $(\ddd,e)$-representation for $\al B$ implies that
$(\gamma,\beta^{\gamma^e})\in\fork_m(R)$.
Thus, \itfm$_m$ holds, as we wanted to show.
\end{proof}

\begin{proof}[Proof of Theorem~\ref{thm-rep-generates}]
Statement~\eqref{it:repgen1} of Theorem~\ref{thm-rep-generates} 
is an immediate consequence of \linebreak Lemma~\ref{lm-rep-generates}.
Statement~\eqref{it:repgen2} of Theorem~\ref{thm-rep-generates}
follows directly from \eqref{eq:cTn} and the definition of $T_n$.
\end{proof}

We close this section by discussing the following question:
given a generating set for a subalgebra $\al B$
of a product $\al B_1\times\dots\times\al B_n$ (in a variety with a cube term)
and a product congruence
$\theta=\theta_1\times\dots\times\theta_n$
of $\al B_1\times\dots\times\al B_n$
(where $\theta_i\in\Con(\al{B}_i)$ for all $i\in[n]$),
how can one construct a generating set for the $\theta$-saturation
$\al B[\theta]$ of $\al B$?
This is a crucial step for giving a positive answer to Question~2
in the introduction, as we explain now.

Suppose that we are given
a subalgebra $\al{D}$ of 
$\al{B}_1/\theta_1\times\dots\times\al{B}_n/\theta_n$
by a generating set
$\{c_1,\dots,c_k\}$, along with another element $v$ of
$\al{B}_1/\theta_1\times\dots\times\al{B}_n/\theta_n$,
and our task is
to decide whether or not $v\in\al{D}$, but we are only allowed to
do computations in $\al{B}_1\times\dots\times\al{B}_n$.
For the product congruence $\theta=\theta_1\times\dots\times\theta_n$
on $\al{B}_1\times\dots\times\al{B}_n$, we can use the natural isomorphism
\[
\al{B}_1/\theta_1\times\dots\times\al{B}_n/\theta_n
\cong(\al{B}_1\times\dots\times\al{B}_n)/\theta,
\quad(b_1/\theta_1,\dots,b_n/\theta_n)\mapsto(b_1,\dots,b_n)/\theta
\]
to view the elements $c_1,\dots,c_k,v$ and the subalgebra $\al{D}$ of
$\al{B}_1/\theta_1\times\dots\times\al{B}_n/\theta_n$
as those of
$(\al{B}_1\times\dots\times\al{B}_n)/\theta$.
Now choose representatives $a_1,\dots,a_k,u\,(\in B_1\times\dots\times B_n)$
from the $\theta$-classes $c_1,\dots,c_k,v$, respectively, and let
$\al{B}$ be the subalgebra of $\al{B}_1\times\dots\times\al{B}_n$
generated by the set $\{a_1,\dots,a_k\}$.
Since the definition
of $\al B[\theta]$ in Section~\ref{sec-prelim} implies
that $\al B[\theta]$ is the full inverse image of
$\al{D}$ under the natural homomorphism 
\[
\al B_1\times\dots\times\al B_n\twoheadrightarrow
(\al B_1\times\dots\times\al B_n)/\theta\cong
\al B_1/\theta_1\times\dots\times\al B_n/\theta_n,
\]
we get that $v\in\al{D}$
if and only if $u\in\al{B}[\theta]$.
Hence, we can decide whether $v$ is in $\al{D}$
by a subpower membership calculation in $\al B_1\times\dots\times\al B_n$,
provided we can
construct a generating set for $\al{B}[\theta]$ from the
generating set $\{a_1,\dots,a_k\}$ of $\al{B}$ .
We now give such a generating set explicitly.
  
\begin{thm}
\label{thm-gen-sat}
Let $\al B$ be a subalgebra of $\al B_1\times\dots\times\al B_n$
with $\al B_1,\dots,\al B_n\in\var{V}$ $(n\ge\ddd)$,
let $\theta_i\in\Con(\al B_i)$ for all $i\in[n]$,
and let 
$\theta:=\theta_1\times\dots\times\theta_n\in
\Con(\al B_1\times\dots\times\al B_n)$.
Assume that $G$ is a generating set for $\al B$, 
and $\LLL,\FFF\subseteq B[\theta]$ satisfy the following conditions:
\begin{enumerate}
\item \label{it:L} 
  For every $I\in\binom{[n]}{\ddd-1}$ and for every $\bar{b}\in B|_I$
  and $\bar{c}\in B[\theta]|_I$ with
  $\bar{b} \equiv_{\theta|_I} \bar{c}$ there exist
  $b_{I,\bar{b}}\in B$ and $c_{I,\bar{c}}\in \LLL$ such that
\begin{equation*}
  b_{I,\bar{b}}|_I=\bar{b},
\quad  
c_{I,\bar{c}}|_I=\bar{c},
\quad\text{and}\quad 
b_{I,\bar{b}}
\equiv_{\theta}
c_{I,\bar{c}}.
\end{equation*}
\item \label{it:F} 
For every $i\in[n]$ and $(\beta,\gamma)\in\theta_i$ with $\beta\in B|_i$,
there exist $b_{i,\beta}\in B$ and $c_{i,\gamma}\in \FFF$ such that
\begin{equation*}
b_{i,\beta}|_i=\beta,\quad
c_{i,\gamma}|_i=\gamma,\quad 
b_{i,\beta}|_{[i-1]} = c_{i,\gamma}|_{[i-1]},
\quad\text{and}\quad 
b_{i,\beta}\equiv_{\theta} c_{i,\gamma}.
\end{equation*}
\end{enumerate}
Then $G\cup \LLL\cup \FFF$ is a generating set for $\al B[\theta]$.
\end{thm}

\begin{proof}
 Let $\al C$ denote the subalgebra of 
$\al B_1\times\dots\times\al B_n$ generated by $G\cup \LLL\cup \FFF$.
Clearly, $\al B\le\al C\le\al B[\theta]$.
To show that $\al C=\al B[\theta]$, choose an arbitrary
element $d$
in $\al B[\theta]$. 
Then there exists
$b$
in $\al B$ such that 
$b\equiv_{\theta}d$.
Using that $\LLL$ satisfies condition~\eqref{it:L}, choose and fix elements 
$b^I\in B$ and $d^I\in \LLL$ such that 
\[
b^I|_I=b|_I,\quad
d^I|_I=d|_I,\quad \text{and}\quad
b^I \equiv_{\theta} d^I\quad
\text{for each}\quad I\in\textstyle{\binom{[n]}{\ddd-1}}.
\]
We will use the $P$-terms in Lemma~\ref{lm-Ts} with $e=2$ to show that
\begin{equation}
\label{eq-d}
d=
T_n(v^{(\ddd)},\hat{v}^{(\ddd)},\dots,v^{(n)},\hat{v}^{(n)},
\bar{d^I}^{[n]})
\end{equation}
holds for appropriately chosen 
elements $v^{(\ddd)},\hat{v}^{(\ddd)},\dots,v^{(n)},\hat{v}^{(n)}$
in $\al C$. Since all $d^I$ are in $\LLL\ (\subseteq\al C)$, this will 
show that $d\in\al C$, and hence will complete the proof of the theorem.

By the definition of the term $T_{\ddd-1}:=w^{[\ddd-1]}$
and by the choice of $b^{[\ddd-1]}$ and $d^{[\ddd-1]}$ we have that
\[
b|_{[\ddd-1]}=b^{[\ddd-1]}|_{[\ddd-1]},\quad 
d|_{[\ddd-1]}=d^{[\ddd-1]}|_{[\ddd-1]}, 
\]
and
\[
T_{\ddd-1}(b^{[\ddd-1]})=b^{[\ddd-1]}\equiv_{\theta} d^{[\ddd-1]}
=T_{\ddd-1}(d^{[\ddd-1]}).
\]
Now we proceed by induction to show that for every $m=\ddd,\dots,n$
there exist tuples $u^{(\ddd)},\hat{u}^{(\ddd)},\dots,u^{(m)},\hat{u}^{(m)}$
in $\al B$ 
and tuples $v^{(\ddd)},\hat{v}^{(\ddd)},\dots,v^{(m)},\hat{v}^{(m)}$ in $\al C$
such that
\begin{align}
b|_{[m]}
&{}=T_m(u^{(\ddd)},\hat{u}^{(\ddd)},\dots,u^{(m)},\hat{u}^{(m)},
\bar{b^I}^{[m]})|_{[m]},\label{eq-bm} \\ 
d|_{[m]}
&{}=T_m(v^{(\ddd)},\hat{v}^{(\ddd)},\dots,v^{(m)},\hat{v}^{(m)},
\bar{d^I}^{[m]})|_{[m]},\label{eq-dm}  
\end{align}
and
\begin{equation}
\label{eq-theta-rel-forks}
u^{(j)}\equiv_{\theta}v^{(j)},
\quad
\hat{u}^{(j)}\equiv_{\theta}\hat{v}^{(j)}
\qquad\text{for all $j=\ddd,\dots,m$,}
\end{equation}
hence also
\begin{multline}
\label{eq-theta-related}
\qquad
(B\ni)\ T_m(u^{(\ddd)},\hat{u}^{(\ddd)},\dots,u^{(m)},\hat{u}^{(m)},
\bar{b^I}^{[m]})\\
\equiv_{\theta}
T_m(v^{(\ddd)},\hat{v}^{(\ddd)},\dots,v^{(m)},\hat{v}^{(m)},
\bar{d^I}^{[m]})\ (\in C).
\qquad
\end{multline}
Then, equality \eqref{eq-dm} for $m=n$ yields the  desired equality
\eqref{eq-d}.

Our induction hypothesis is that the statement in the preceding paragraph
is true for $m-1$,
that is, there exist tuples 
$u^{(\ddd)},\hat{u}^{(\ddd)},\dots,u^{(m-1)},\hat{u}^{(m-1)}$ in $\al B$ and
tuples
$v^{(\ddd)},\hat{v}^{(\ddd)},\dots,v^{(m-1)}$, $\hat{v}^{(m-1)}$ in $\al C$
such that \eqref{eq-bm}--\eqref{eq-theta-related} hold for $m-1$ in place of 
$m$.
To simplify notation, let 
\begin{align*}
b^{(m-1)}
&{}:=T_{m-1}(u^{(\ddd)},\hat{u}^{(\ddd)},\dots,u^{(m-1)},\hat{u}^{(m-1)},
\bar{b^I}^{[m-1]}),\\
d^{(m-1)}
&{}:=T_{m-1}(v^{(\ddd)},\hat{v}^{(\ddd)},\dots,v^{(m-1)},\hat{v}^{(m-1)},
\bar{d^I}^{[m-1]}).
\end{align*}
Then we have that
$b^{(m-1)}|_{[m-1]}=b|_{[m-1]}$ and $d^{(m-1)}|_{[m-1]}=d|_{[m-1]}$. 
Let $\gamma:=b|_m$, $\beta:=b^{(m-1)}|_m$,
$\tau:=d|_m$, and $\sigma:=d^{(m-1)}|_m$.
If we can show the existence of a pair of witnesses
$u^{(m)},\hat{u}^{(m)}$  in $\al B$ 
for the fork
$(\gamma,\beta^{\gamma^2})\in\fork_m(B)$
and a pair of witnesses
$v^{(m)},\hat{v}^{(m)}$  in $\al C$ 
for the fork
$(\tau,\sigma^{\tau^2})\in\fork_m(C)$
such that $u^{(m)}\equiv_{\theta} v^{(m)}$ and
$\hat{u}^{(m)}\equiv_{\theta} \hat{v}^{(m)}$, then
\eqref{eq-theta-rel-forks}--\eqref{eq-theta-related}
will follow for $m$, and by Lemma~\ref{lm-Ts},
\eqref{eq-bm}--\eqref{eq-dm} will also hold for $m$. 

To prove the existence of such
$u^{(m)},\hat{u}^{(m)}$
and
$v^{(m)},\hat{v}^{(m)}$,
notice first that
our induction hypothesis that \eqref{eq-theta-related} holds
for $m-1$ in place of $m$ has the effect that 
$(B\ni)\ b^{(m-1)}\equiv_{\theta} d^{(m-1)}\ (\in C)$,
so in particular, $(B|_m\ni)\ \beta\equiv_{\theta_m}\sigma$.
By our choice of $b$ and $d$, we also have that 
$(B|_m\ni)\ \gamma=b|_m\equiv_{\theta_m} d|_m=\tau$.
Thus, by our assumption \eqref{it:F}, there exist witnesses
$b_\beta:=b_{m,\beta}$
and
$c_\sigma:=c_{m,\sigma}\in \FFF$
for the fork $(\beta,\sigma)\in\fork_m(C)$ such that
\begin{equation}
\label{eq-bc1}
b_{\beta}|_m=\beta,\quad
c_{\sigma}|_m=\sigma,\quad
b_{\beta}|_{[m-1]} = c_{\sigma}|_{[m-1]},\quad\text{and}\quad
b_{\beta}\equiv_{\theta} c_{\sigma}.
\end{equation}
There also exist witnesses 
$b_\gamma:=b_{m,\gamma}\in B$ and
$c_\tau:=c_{m,\tau}\in \FFF$
for the fork $(\gamma,\tau)\in\fork_m(C)$ such that
\begin{equation}
\label{eq-bc2}
b_{\gamma}|_m=\gamma,\quad
c_{\tau}|_m=\tau,\quad
b_{\gamma}|_{[m-1]} = c_{\tau}|_{[m-1]},\quad\text{and}\quad
b_{\gamma}\equiv_{\theta} c_{\tau}.
\end{equation}
Let $b_{\beta^\gamma}:=p(b_{\beta},b_{\gamma},b_{\gamma})$,
$c_{\beta^\tau}:=p(b_{\beta},c_{\tau},c_{\tau})$, and
$c_{\sigma^\tau}:=p(c_{\sigma},c_{\tau},c_{\tau})$.
Clearly, $b_{\beta^\gamma}\in B$ and
$c_{\beta^\tau}, c_{\sigma^\tau}\in C$.
Furthermore, the relations in 
\eqref{eq-bc1}--\eqref{eq-bc2}
imply that
\begin{multline}
\label{eq-combined-bc}
\quad
b_{\beta^\gamma}|_m=\beta^\gamma,\quad
c_{\beta^\tau}|_m=\beta^\tau,\quad
c_{\sigma^\tau}|_m=\sigma^\tau,\\ 
b_{\beta^\gamma}|_{[m-1]} = c_{\beta^\tau}|_{[m-1]} 
= c_{\sigma^\tau}|_{[m-1]},\quad\text{and}\quad
b_{\beta^\gamma} \equiv_{\theta} c_{\beta^\tau} 
\equiv_{\theta} c_{\sigma^\tau}.
\qquad
\end{multline}
It follows that $c_{\beta^\tau}$ and $c_{\sigma^\tau}$
are witnesses in $\al C$ for the fork 
$(\beta^{\tau},\sigma^{\tau})\in\fork_m(C)$.

Since $b$ and $b^{(m-1)}$ are witnesses in $B$ for the fork 
$(\gamma,\beta)\in\fork_m(B)\subseteq\fork_m(C)$, we can apply 
Lemma~\ref{lm-newfork}
first to $(\tau,\gamma)$ and $(\beta,\gamma)$ to obtain witnesses
\begin{align*}
z&{}:=p(p(c_\tau,b_\gamma,b),p(c_\tau,b_\gamma,b_\gamma),c_\tau)\ (\in C)\\
\hat{z}&{}:=p(b^{(m-1)},c_\tau,c_\tau)\ (\in C)
\end{align*}
for the fork $(\tau,\beta^\tau)\in\fork_m(C)$. 
The analogous construction for the forks $(\gamma,\gamma)$
and $(\beta,\gamma)$ yields witnesses
\begin{align*}
w&{}:=p(p(b_\gamma,b_\gamma,b),p(b_\gamma,b_\gamma,b_\gamma),b_\gamma)=
p(b,b_\gamma,b_\gamma)\ (\in B)\\
\hat{w}&{}:=p(b^{(m-1)},b_\gamma,b_\gamma)\ (\in B)
\end{align*}
for the fork $(\gamma,\beta^\gamma)\in\fork_m(B)$. 
Since $w,\hat{w}$ are obtained from $z,\hat{z}$ by replacing $c_\tau$ with
$b_\gamma$,
the relation $b_\gamma \equiv_{\theta} c_\tau$ 
in \eqref{eq-bc2}
implies that
$z\equiv_{\theta} w$ and $\hat{z}\equiv_{\theta} \hat{w}$.
Applying Lemma~\ref{lm-newfork}
again, now to the forks $(\tau,\beta^\tau)$ and
$(\sigma^\tau,\beta^\tau)$, we obtain witnesses
\begin{align*}
v^{(m)}&{}:=p(p(z,\hat{z},c_{\beta^\tau}),p(z,\hat{z},\hat{z}),z)\ (\in C)\\
\hat{v}^{(m)}&{}:=p(c_{\sigma^\tau},z,z)\ (\in C)
\end{align*}
for the fork $(\tau,\sigma^{\tau^2})=(\tau,(\sigma^\tau)^\tau)\in\fork_m(C)$.
Similarly, for the corresponding forks 
$(\gamma,\beta^\gamma)$ and
$(\beta^\gamma,\beta^\gamma)$, we get witnesses
\begin{align*}
u^{(m)}&{}:=p(p(w,\hat{w},b_{\beta^\gamma}),p(w,\hat{w},\hat{w}),w)\ (\in B)\\
\hat{u}^{(m)}&{}:=p(b_{\beta^\gamma},w,w)\ (\in B)
\end{align*}
for the fork $(\gamma,\beta^{\gamma^2})=(\gamma,(\beta^\gamma)^\gamma)\in\fork_m(B)$.
Since $u^{(m)}$ and $\hat{u}^{(m)}$ are obtained from 
$v^{(m)}$ and $\hat{v}^{(m)}$ by replacing $z$ with $w$, $\hat{z}$ with $\hat{w}$, and
$c_{\sigma^\tau}, c_{\beta^\tau}$ with $b_{\beta^\gamma}$, the relations
$z\equiv_{\theta} w$, $\hat{z}\equiv_{\theta} \hat{w}$ proved earlier, and
the relations 
$b_{\beta^\tau}\equiv_{\theta} c_{\beta^\tau}\equiv_{\theta}c_{\sigma^\tau}$ 
in \eqref{eq-combined-bc} 
imply that
$v^{(m)}\equiv_{\theta} u^{(m)}$ and
$\hat{v}^{(m)}\equiv_{\theta} \hat{u}^{(m)}$.
This completes the induction for~\eqref{eq-bm}--\eqref{eq-theta-related} and the proof of the theorem.
\end{proof}

\section{Algorithms Involving Compact Representations}
\label{sec-cr-algorithms}

Throughout this section we will work under the following global
assumptions.

\begin{asm}\quad
  \begin{itemize}
  \item
    $\var{V}$ is a variety in a finite language $F$
    with a $\ddd$-cube term ($\ddd>1$),
  \item
    $P$ is a $(1,\ddd-1)$-parallelogram term in $\var{V}$ (the existence
    of such a term is ensured by Theorem~\ref{thm-cube-par}), and
  \item
    $\class{K}$ is a finite set of finite algebras in $\var{V}$.
  \end{itemize}  
\end{asm}    

Our aim in this section is to use the results of Section~\ref{sec-cRep}
to show that 
\begin{enumerate}
  \item[$\blacktriangleright$]
the problems $\SMP(\class{K})$ and $\SMP(\HH\SSS\class{K})$
are polynomial time equivalent (Theorem~\ref{thm-smp-hom}),
  \item[$\blacktriangleright$]
$\SMP(\class{K})$ is in $\NP$ (Theorem~\ref{thm-short-rep}), 
and
\item[$\blacktriangleright$]
$\SMP(\class{K})$
and the problem
of finding compact representations
for algebras in $\SSS\PPP_{\rm fin}\class{K}$ given by their generators
are polynomial time reducible to each other
(Theorem~\ref{thm-smp-cr-equiv}).
\end{enumerate}

Recall that if $\al{B}$ is a subalgebra of \mbox{$\al A_1\times\dots\times\al A_n$}
with $\al A_1,\dots,\al A_n\in\class{K}$, then a representation for $\al{B}$ is a subset $R$ of $B$ satisfying
  conditions~\eqref{it:I}--\eqref{it:Fi} of Definition~\ref{df-rep} (with $e=1$).
Condition~\eqref{it:I} in the definition
makes sure that for every 
$I\subseteq[n]$ with $|I|=\min\{n,\ddd-1\}$ and for every
tuple $\bar{b}\in B|_I$ there exists an element $r_{I,\bar{b}}$ in $R$ such that $r_{I,\bar{b}}|_I=\bar{b}$.
Similarly, condition~\eqref{it:Fi} makes sure that for every derived fork $(\gamma,\delta)\in\fork'_i(B)$ ($i\in[n]$)
there exist elements $u_{i,\gamma,\delta},\hat{u}_{i,\gamma,\delta}$ in $R$ which witness that $(\gamma,\delta)\in\fork_i(R)$.
For some of the algorithms we are going to discuss, it will be convenient to
fix these choices for witnesses in $R$.
This motivates our definition of a standardized representation below.
 
 Before the definition, let us note that if $n<\ddd$,
  condition~\eqref{it:I} forces $R=B$, so finding representations
  for $\al{B}$ in this case is trivial. Therefore, we will assume that
  $n\ge\ddd$. Then, by condition~\eqref{it:I}, we have
  $R|_{[\ddd-1]}=B|_{[\ddd-1]}$, hence one can easily find witnesses
  among the elements $r_{[\ddd-1],\bar{b}}\in R$
  for the forks in coordinates $i\in[\ddd-1]$.
  Therefore, witnesses for forks are of interest
  only in coordinates $i=\ddd,\ddd+1,\dots,n$.
  (See also Lemma~\ref{lm-rep-generates}.)

\begin{defi}
\label{df-strep}
 For a subalgebra $\al{B}$ of \mbox{$\al A_1\times\dots\times\al A_n$}
with $\al A_1,\dots,\al A_n\in\class{K}$ ($n\geq\ddd$),
we will call a representation $R$ for $\al{B}$
 a \emph{standardized representation for $\al{B}$}
if the following conditions are satisfied:
\begin{enumerate}
\item
  \label{it:strep1}
  For every $I\subseteq[n]$ with $|I|=\ddd-1$ and for every tuple $\bar{b}\in B|_I$, an element
$r_{I,\bar{b}}$ of $R$ is fixed so that $r_{I,\bar{b}}|_I=\bar{b}$;
this element will be referred to as \emph{the designated witness in $R$ for $\bar{b}\in R|_I$}.
\item
  \label{it:strep2}
  For every $m$ ($\ddd\leq m\leq n$) there is a set $F_m$ with 
$\fork'_m(B)\subseteq F_m\subseteq\fork_m(R)$ such that
for every $(\gamma,\delta)\in F_m$ a pair 
$(u_{m,\gamma,\delta},\hat{u}_{m,\gamma,\delta})$
in $R^2$ is fixed
which witnesses that $(\gamma,\delta)\in\fork_m(R)$;
these elements
$u_{m,\gamma,\delta},\hat{u}_{m,\gamma,\delta}$
will be referred to as 
\emph{the designated witnesses in $R$ for $(\gamma,\delta)\in\fork_m(R)$}.

\item
    \label{it:strep3}
Every element of $R$ has at least one designation.
\end{enumerate}
We will refer to the designated witnesses in \eqref{it:strep1}
and \eqref{it:strep2}
as \emph{local witnesses for $\al{B}$} and \emph{fork witnesses for $\al{B}$,
respectively.}
\end{defi}

  Formally, a standardized representation for $\al{B}$
  is the representation $R$ together with
    a \emph{designation function}
    \begin{multline}\label{eq-desig-fn}
      \Bigl(\bigl\{(I,\bar{b}):I\in{\textstyle\binom{[n]}{\ddd-1}},\,
      \bar{b}\in B|_I\bigr\}\sqcup
    \bigcup_{m=\ddd}^n \bigl(\{m\}\times F_m\times[2]\bigr)\Bigr)\to R,\\
          (I,\bar{b})\mapsto r_{I,\bar{b}},\
          (m,(\gamma,\delta),1)\mapsto u_{m,\gamma,\delta},\
          (m,(\gamma,\delta),2)\mapsto \hat{u}_{m,\gamma,\delta} 
    \end{multline}
    where the notation is the same as in
    Definition~\ref{df-strep}~\eqref{it:strep1}--\eqref{it:strep2},
    and $\sqcup$ denotes disjoint union.
    Condition~\eqref{it:strep3} in the definition requires
    the designation function to be onto.
    For our purposes in this paper it will be sufficient to treat the
    designation function of a standardized representation informally, as we
    did in Definition~\ref{df-strep}.

\begin{rem}
\label{rm-representable}
It is clear that under the assumptions of Definition~\ref{df-strep}
  a standardized representation for $\al{B}$
is compact.
Moreover, it follows from
Lemma~\ref{lm-rep-generates}
that if
$R$ is a standardized representation for $\al B$, then
the equality \eqref{eq-rep-b} in Lemma~\ref{lm-Ts} (with $e=1$)
holds provided
\begin{itemize}
\item
$\bar{b^I}$ is an enumeration of the elements $b^I$ 
of $R$ designated to witness $b^I|_I=b|_I$ for all
$I\in\binom{[n]}{\ddd-1}$; and
\item
for every $\ddd\le m\le n$, the pair $(u^{(m)}, \hat{u}^{(m)})$
is the designated witness in $R$ for the fork 
$(\gamma,\beta^\gamma)\in\fork_m(R)$ where
$\gamma$ and $\beta$ are determined by \eqref{eq-betagamma}.
\end{itemize}
\end{rem}

 Now we are ready to describe in more detail what we mean by
  the problem of finding compact representations
for algebras in $\SSS\PPP_{\rm fin}\class{K}$ given by their generators:

\medskip

\noindent $\CompactRep(\class{K})$:
\begin{itemize}
\item
INPUT:  
$a_1,\dots,a_k\in \al A_1\times\dots\times\al A_n$
with $\al A_1,\ldots,\al A_n\in\class{K}$ ($n\ge\ddd$).
\item
OUTPUT: A standardized representation for the subalgebra of 
$\al A_1\times\dots\times\al A_n$ 
generated by $\{a_1,\dots,a_k\}$.
\end{itemize}

\smallskip

In Subsections~\ref{ssec-finding-crs1} and \ref{ssec-finding-crs2} below
we will discuss two different algorithms for
solving $\CompactRep(\class{K})$.
The idea of the first algorithm (Algorithm~\ref{alg-cr-direct})
will be key to showing in Subsection~\ref{ssec-smp-np}
that $\SMP(\class{K})$ is in $\NP$.
The second algorithm
(Algorithm~\ref{alg-smp-red-cr}),
on the other hand, will use $\SMP(\class{K})$ to solve
$\CompactRep(\class{K})$.
Prior to our first algorithm for $\CompactRep(\class{K})$ we will
introduce two auxiliary algorithms (Subsection~\ref{ssec-aux}), and
then prove in Subsection~\ref{ssec-HS} that the problems
$\SMP(\class{K})$ and $\SMP(\HH\SSS\class{K})$ are polynomial time equivalent.

\subsection{Two auxiliary algorithms}
\label{ssec-aux}

When we construct standardized representations, the following concepts
will be useful.

\begin{defi}
\label{df-partial-strep}
Let $\al A_1,\ldots,\al A_n\in\class{K}$ with $n\ge\ddd$, and let 
$R\subseteq\al A_1\times\dots\times\al A_n$.
We will call $R$ a \emph{partial standardized representation}
if every element of $R$ is designated to witness
either $\bar{b}\in R|_I$ for some $I\in\binom{[n]}{\ddd-1}$ or a fork
in $\fork_m(R)$ for some $m$ ($\ddd\le m\le n$).

If $R$ is a partial standardized representation that is contained
in a subalgebra $\al{B}$ of $\al A_1\times\dots\times\al A_n$,
we may refer to $R$ as a \emph{partial standardized representation for
$\al{B}$}.
If, in addition, $R$ satisfies condition \eqref{it:strep1} of
Definition~\ref{df-strep}, we will say that \emph{$R$ contains a full set of
  designated local witnesses for $\al{B}$}.
\end{defi}

  Similarly to standardized representations, a partial standardized
  representation is formally defined to be a set $R$ satisfying
  the requirements of Definition~\ref{df-partial-strep}, together with
  a \emph{designation function}, which is a function mapping onto $R$,
  similar to
  \eqref{eq-desig-fn}, but the domain is a subset of the set
  \[
   \Bigl\{(I,\bar{b}):I\in{\textstyle\binom{[n]}{\ddd-1}},\,
      \bar{b}\in \prod_{i\in I}\al{A}_i\Bigr\}\sqcup
    \bigcup_{m=\ddd}^n \bigl(\{m\}\times A_m^2\times[2]\bigr).
    \]
 If $R$ is a partial standardized representation for
 a subalgebra $\al{B}$ of $\al A_1\times\dots\times\al A_n$, then
 the domain of its designation function is a subset of the
 domain indicated in \eqref{eq-desig-fn}. 
 Although we will usually treat the designation functions of partial
 standardized representations $R$ informally, we will consider the size of $R$
 to be the size of the domain of its designation function (not just the
 cardinality of the set $R$). In particular, if $R$ is a partial
 standardized representation for
 a subalgebra $\al{B}$ of $\al A_1\times\dots\times\al A_n$
 such that $R$ contains
 a full set of designated local witnesses for $\al{B}$, then
 its size is $\Theta(n^{\ddd-1})$.

\begin{defi}
\label{df-representable}
 Let $\al A_1,\ldots,\al A_n\in\class{K}$ with $n\ge\ddd$, let
  $R\subseteq\al A_1\times\dots\times\al A_n$ be a partial standardized
  representation, and let 
  $b\in\al A_1\times\dots\times\al A_n$.
We will say that \emph{$b$ is representable by $R$} if
the following conditions are met:
\begin{enumerate}
\item
  \label{it:representable1}
for each $I\in\binom{[n]}{\ddd-1}$,
$R$ contains elements $b^I$ designated to witness $b^I|_I=b|_I$, and 
\item
  \label{it:representable2}
for every $\ddd\le m\le n$, 
$R$ contains designated witnesses $u^{(m)},\hat{u}^{(m)}$ 
for the fork $(\gamma,\beta^\gamma)$
where $\beta,\gamma$ are as defined in
\eqref{eq-betagamma} (with $e=1$).
\end{enumerate}
We will say that \emph{$b$ is completely representable by $R$} if
  condition \eqref{it:representable1} and the following stronger version of
  \eqref{it:representable2} hold:
\begin{enumerate}
\item[\eqref{it:representable2}${}_{\rm c}$]
for every $\ddd\le m\le n$, 
$R$ contains designated witnesses $u^{(m)},\hat{u}^{(m)}$ 
and $v^{(m)},\hat{v}^{(m)}$ 
for the forks $(\gamma,\beta^\gamma)$
and $(\gamma,\beta)$ in $\fork_m(R)$ 
where $\beta,\gamma$ are as defined in
\eqref{eq-betagamma} (with $e=1$).
\end{enumerate}

\end{defi}

The following fact is an easy consequence of Lemma~\ref{lm-Ts}.

\begin{cor}
\label{cor-representable}
Let
$R\subseteq\al A_1\times\dots\times\al A_n$
$(\al A_1,\ldots,\al A_n\in\class{K},\ n\ge\ddd)$
be a partial standardized representation.
If an element $b\in\al A_1\times\dots\times\al A_n$ is representable by $R$,
then the equality in
\eqref{eq-rep-b} holds for $b$ $($with $e=1)$, and
hence $b$ is in the $P$-subalgebra of $\al A_1\times\dots\times\al A_n$
generated by $R$.
\end{cor}  

The proof of Lemma~\ref{lm-Ts} can easily be turned into a
polynomial time algorithm for solving the following problem.

\medskip

\begin{samepage}
\noindent $\IsRepresentable(\class{K})$: 
\begin{itemize}
\item
INPUT: $b\in \al A_1\times\dots\times\al A_n$ ($n\geq\ddd$)
and a partial standardized representation\\
$R\subseteq \al A_1\times\dots\times\al A_n$
($\al A_1,\ldots,\al A_n\in\class{K}$) such that
$R$ contains elements $b^I$ designated to witness $b^I|_I=b|_I$
for each $I\in\binom{[n]}{\ddd-1}$.
\item
OUTPUT: (YES, $\emptyset$, $S$) or (NO, $S'$, $S$) where
$S',S\,(\subseteq\langle R\cup\{b\}\rangle_P)$ are lists of 
designated witnesses 
(missing from $R$) for the derived forks and for the forks that 
are not derived forks,
respectively, so that $b$ becomes representable(completely representable) by 
the partial standardized representation $R\cup S'$ ($R\cup S'\cup S$, respectively).
\end{itemize}
\end{samepage}

\smallskip

Note that the designated witnesses (for forks that are not derived forks) 
collected in the set $S$ do not play a role in determining whether or not
$b$ is representable, but they will be useful in other algorithms that call
$\IsRepresentable(\class{K})$.

\def\Input{\text{\bf Input:}}
\def\Output{\text{\bf Output:}}
\def\Question{\text{\bf Question:}}
\def\for{\text{\bf for}}
\def\endfor{\text{\bf end for}}
\def\do{\text{\bf do}}
\def\ifif{\text{\bf if}}
\def\endif{\text{\bf end if}}
\def\then{\text{\bf then}}
\def\els{\text{\bf else}}
\def\while{\text{\bf while}}
\def\endwhile{\text{\bf end while}}
\def\return{\text{\bf return}}

\begin{lem}
\label{lm-isrep}
Algorithm~\ref{alg-ir} solves $\IsRepresentable(\class{K})$ in polynomial time.
\end{lem}

\begin{algorithm}
\caption{For $\IsRepresentable(\class{K})$}
\label{alg-ir} 
\begin{algorithmic}[1] 
\REQUIRE $b\in \al A_1\times\dots\times\al A_n$
and a partial standardized representation
$R\subseteq \al A_1\times\dots\times\al A_n$
($\al A_1,\ldots,\al A_n\in\class{K}$, $n\geq\ddd$) such that
$R$ contains elements $b^I$ designated to witness $b^I|_I=b|_I$
for each $I\in\binom{[n]}{\ddd-1}$
\ENSURE (YES, $\emptyset$, $S$) or (NO, $S'$, $S$) where
$S',S\,(\subseteq\langle R\cup\{b\}\rangle_P)$ are lists of 
designated witnesses 
(missing from $R$) for the derived forks and for the forks that 
are not derived forks,
respectively, so that $b$ becomes representable
(completely representable) by 
the partial standardized representation $R\cup S'$
($R\cup S'\cup S$, respectively)
\STATE $S:=\emptyset$, $S':=\emptyset$, 
\STATE $b^{(\ddd-1)}:=b^{[\ddd-1]}$
\FOR{$m=\ddd ,\dots,n$}
\STATE $\beta:=b^{(m-1)}|_m$, $\gamma:=b|_{m}$, $c:=p(b^{(m-1)},b,b)$
\IF{$R\cup S'$ has no designated witnesses for 
$(\gamma,\beta^\gamma)\in\fork_m(R\cup S')$} 
\STATE add $b,c$ to $S'$ as designated witnesses for 
$(\gamma,\beta^\gamma)\in\fork_m(R\cup S')$
\ENDIF
\IF{$R\cup S'\cup S$ has no designated witnesses for 
$(\gamma,\beta)\in\fork_m(R\cup S'\cup S)$}
\STATE add $b,b^{(m-1)}$ to $S$ as designated witnesses for 
$(\gamma,\beta)\in\fork_m(R\cup S'\cup S)$
\ENDIF
\STATE let $u,\hat{u}\in R\cup S'$ be the designated witnesses for
$(\gamma,\beta^\gamma)\in\fork_m(R\cup S')$
\STATE $b^{(m)}=t_m(b^{(m-1)},\hat{u},u,\bar{b^I}^{[m]})$
\ENDFOR 
\IF{$S'=\emptyset$} 
\STATE \return\ (YES, $S'$, $S$)
\ELSE 
\STATE \return\ (NO, $S'$, $S$)
\ENDIF
\end{algorithmic}
\end{algorithm}

\begin{proof}
First we prove that Algorithm~\ref{alg-ir} is correct.
In Step~2, $b^{[\ddd-1]}\in R$ is the designated witness for
$b^{[\ddd-1]}|_{[\ddd-1]}=b|_{[\ddd-1]}$.
The loop in Steps~3--13
follows the induction in the
proof of Lemma~\ref{lm-Ts} (with $e=1$).
For each $m=\ddd,\dots,n$, if $b^{(m-1)}$,
$S'$ and $S$ (with $S'=S=\emptyset$ for $m=\ddd$ by Step~1)
have
been constructed
such that
  $S',S\subseteq\langle R\cup\{b\}\rangle_P$,
$b^{(m-1)}\in\langle R\cup S'\rangle_P$,%
and 
$b^{(m-1)}|_{[m-1]}=b|_{[m-1]}$, then to construct 
$b^{(m)}=t_m(b^{(m-1)},\hat{u},u,\bar{b^I}^{[m]})$ we need 
designated witnesses
$u$, $\hat{u}$ for the fork $(\gamma,\beta^\gamma)\in\fork_m(R\cup S')$ 
where
$\gamma=b|_m$ and $\beta=b^{(m-1)}|_m$.
  Furthermore, for $b$ to be completely representable by $R\cup S'\cup S$,
  we need designated witnesses
$v$, $\hat{v}$ for the fork $(\gamma,\beta)\in\fork_m(R\cup S'\cup S)$.

Clearly, $b,b^{(m-1)}$ 
witness the fork
$(\gamma,\beta)\in\fork_m(\langle R\cup\{b\}\rangle_B)$.
It follows that $b$ and 
$c=p(b^{(m-1)},b,b)$ witness the derived fork
$(\gamma,\beta^\gamma)\in\fork'_m(\langle R\cup\{b\}\rangle_P)$, because
$c|_{[m-1]}=p(b^{(m-1)},b,b)|_{[m-1]}=p(b,b,b)|_{[m-1]}=b|_{[m-1]}$
and
$b|_m=\gamma$, 
$c|_m=p(b^{(m-1)},b,b)|_m=p(\beta,\gamma,\gamma)=\beta^\gamma$.
  Therefore, if the required designated witnesses $u,\hat{u}$ do not exist
  in $R\cup S'$, the algorithm (in Steps~5--7) correctly adds $b,c$
  to be the designated witnesses.
  Similarly, if the required designated witnesses $v,\hat{v}$ do not exist
  in $R\cup S'\cup S$, the algorithm (in Steps~8--10) correctly adds
  $b,b^{(m-1)}$
  to be the designated witnesses.
  Notice also that the updated versions $S'$ and $S$ maintain the property
  $S',S\subseteq\langle R\cup\{b\}\rangle_P$ when the loop starts over
  (if $m<n$) or ends (if $m=n$)
  after computing $b^{(m)}$ in Step~12.

  Thus, the standardized representations $R\cup S'$ and $R\cup S'\cup S$ 
  satisfy the requirements that $b$ is representable by $R\cup S'$ and
  completely representable by $R\cup S'\cup S$. Since new designated
  fork witnesses
  (at most one pair for each $m=\ddd,\dots,n$)
  are added to $S'$ or $S$ only if such witnesses do not exist, but
  are necessary for the (complete) representability of $b$, we see that
  $b$ is representable by $R$ if and only if $S'=\emptyset$.  
This proves the correctness of Algorithm~\ref{alg-ir}.

To estimate the running time of Algorithm~\ref{alg-ir}, notice that 
Step~1
requires constant time, while 
Steps~2 and 14--18
can be done in time $O(n|R|)$
and $O(n|S|+n|S'|)$, respectively, where $|R|+|S|+|S'|\le O(n^{\ddd-1})$.
Finally, by Lemma~\ref{lm-terms}~\eqref{it:terms2}
(with $e=1$),
if $b$ is representable by $R$,
the computations in the loop in 
Steps~3--13
require
$O(n^{\ddd+1})$
applications of $P$.
So Steps~3--16
(including the search for witnesses for forks in $R$) 
can be done in $O(n^{\ddd+2})$ time.
If $b$ is not representable by $R$, 
essentially the same computation is performed, 
so the bound $O(n^{\ddd+2})$ applies in this case as well.
This proves that Algorithm~\ref{alg-ir} runs in time $O(n^{\ddd+2})$.
\end{proof}

\begin{rem}
  \label{rem-alg1}
The definition of the $P$-term $T_n$ shows that
Algorithm~\ref{alg-ir} may be viewed as doing the
following two computations {\bf simultaneously},
when it is run on the inputs $b$ and $R$:
\begin{itemize}
\item
it evaluates the $P$-term $T_n$ (for $e=1$)
step-by-step on the tuple $\bar{b^I}$ of designated local witnesses in $R$ and 
on appropriately chosen designated fork witnesses
$u^{(\ddd)},\hat{u}^{(\ddd)},\dots,u^{(n)},\hat{u}^{(n)}$ to obtain
the elements
\[
b^{(m)}=T_m(u^{(\ddd)},\hat{u}^{(\ddd)},\dots,u^{(m)},\hat{u}^{(m)},\bar{b^I}^{[m]}),
\quad m=\ddd,\dots,n,
\]
which, for $m=n$, yields the desired representation
\[
b=T_n(u^{(\ddd)},\hat{u}^{(\ddd)},\dots,u^{(n)},\hat{u}^{(n)},\bar{b^I})
\]
of $b$;
\item
  at each step, to be able to evaluate $T_m$,
  it produces (namely, it finds in $R\cup S'$ or, if missing,
  adds to $S'$)
the designated fork witnesses $u^{(m)},\hat{u}^{(m)}$
for derived forks needed for the evaluation of $T_m$, and does the same with
the designated fork witnesses for the corresponding non-derived forks.
\end{itemize}
So, in addition to some book-keeping to ensure that the necessary
designated fork witnesses (and nothing else) get into the sets $S'$ and $S$,
the computations done by Algorithm~\ref{alg-ir}
in the algebra
$\al{A}_1\times\dots\times\al{A}_n$ can be performed by a $P$-circuit
$\crc{T_n}$ for $T_n$, which is expanded at each node
that produces an element $b^{(m-1)}$ ($m=\ddd,\dots,n$)
by a single $P$-gate that also computes $p(b^{(m-1)},b,b)$
(see the definition of $p$ in \eqref{eq-ps}).
We will denote this expanded circuit by $\crcplus{T_n}$.
Since $\crcplus{T_n}$ is obtained from $\crc{T_n}$ by
adding at most $n-\ddd+1$ $P$-gates,
it follows from \eqref{eq:cTn} (and $e=1$)
that $\crcplus{T_n}$ has
size $O(n^{\ddd+1})+O(n)=O(n^{\ddd+1})$.
\end{rem}

Now we define another problem which will be useful in constructing
  standardized representations for subalgebras of
  $\al{A}_1\times\dots\times\al{A}_n$ ($\al{A}_1\dots,\al{A}_n\in\class{K}$),
  and also for the polynomial time reduction of $\SMP(\HH\SSS\class{K})$  
  to $\SMP(\class{K})$: 

\smallskip

\noindent $\LocalRep(\class{K})$:
\begin{itemize}
\item
INPUT: $a_1,\dots,a_k\in\al{B}_1\times\dots\times\al{B}_n$ with
$\al{B}_1,\dots,\al{B}_n\in\SSS\class{K}$ ($n\ge\ddd$), and
$\theta=\theta_1\times\dots\times\theta_n$ with $\theta_i\in\Con(\al{B}_i)$
for all $i\in[n]$. (Let $\al{B}$ denote the subalgebra of
$\al{B}_1\times\dots\times\al{B}_n$ generated by $\{a_1,\dots,a_k\}$.)
\item
OUTPUT: Partial standardized representation $R$ for the
$\theta$-saturation  $\al{B}[\theta]$ of $\al{B}$ such that
$R$ is a full set of designated local witnesses
for $\al{B}[\theta]$, and
whenever $r_{I,\bar{b}},r_{I,\bar{d}}\in R$ are
designated to witness $r_{I,\bar{b}}|_I=\bar{b}$
and $r_{I,\bar{d}}|_I=\bar{d}$ where $\bar{b}\in\al{B}|_I$ and
$\bar{b}\equiv_{\theta|_I}\bar{d}$, then $r_{I,\bar{b}}\equiv_\theta r_{I,\bar{d}}$.
\end{itemize}

\smallskip

In particular, if $\theta_i=0$ for all $i\in[n]$, then
the latter condition holds vacuously for $R$, therefore
the output of $\LocalRep(\class{K})$ is the subset $R$ of
a standardized representation for $\al{B}$, which 
is a full set of designated local witnesses
for $\al{B}$.

\begin{lem}
\label{lm-locrep}
Algorithm~\ref{alg-locrep} solves $\LocalRep(\class{K})$ in polynomial time.
\end{lem}

\begin{algorithm}
\caption{For $\LocalRep(\class{K})$}
\label{alg-locrep} 
\begin{algorithmic}[1] 
\REQUIRE
$a_1,\dots,a_k\in\al{B}_1\times\dots\times\al{B}_n$ with
$\al{B}_1,\dots,\al{B}_n\in\SSS\class{K}$, and
$\theta=\theta_1\times\dots\times\theta_n$ with $\theta_i\in\Con(\al{B}_i)$
for all $i\in[n]$ ($n\ge\ddd$). (Let $\al{B}$ denote the subalgebra of
$\al{B}_1\times\dots\times\al{B}_n$ generated by $\{a_1,\dots,a_k\}$.)
\ENSURE
$R$ is a partial standardized representation for $\al{B}[\theta]$
such that
$R$ is a full set of designated local witnesses
for $\al{B}[\theta]$, and
whenever $r_{I,\bar{b}},r_{I,\bar{d}}\in R$ are
designated to witness $r_{I,\bar{b}}|_I=\bar{b}$
and $r_{I,\bar{d}}|_I=\bar{d}$ where $\bar{b}\in\al{B}|_I$ and
$\bar{b}\equiv_{\theta|_I}\bar{d}$, then $r_{I,\bar{b}}\equiv_\theta r_{I,\bar{d}}$.
\STATE $R:=\emptyset$
\FOR{$I\in\binom{[n]}{\ddd-1}$}
\STATE generate $\al B|_I$ by $\{a_1|_I,\dots,a_k|_I\}$, and simultaneously,
\FOR{each $\bar{b}\in\al B|_I \setminus R|_I$}
\STATE find an element $r_{I,\bar{b}}\in\al B$ 
such that $r_{I,\bar{b}}|_I=\bar{b}$, include it in $R$,
  and designate it to witness $\bar{b}\in R|_I$
\FOR{all $\bar{d}\in\prod_{i\in I}\al{B}_i$ such that $\bar{d}\not=\bar{b}$ and 
$\bar{d}|_j\equiv_{\theta_j}\bar{b}|_j$ for all $j\in I$}
\STATE include in $R$ the tuple $r_{I,\bar{d}}$ 
defined by $r_{I,\bar{d}}|_I=\bar{d}$  and
$r_{I,\bar{d}}|_{[n]\setminus I}=r_{I,\bar{b}}|_{[n]\setminus I}$
as a designated witness for $\bar{d}\in R|_I$
\ENDFOR 
\ENDFOR 
\ENDFOR

\end{algorithmic}
\end{algorithm}

\begin{proof}
The correctness of Algorithm~\ref{alg-locrep} is straightforward to check.
To estimate its running time, notice that
Step~1 requires constant time, and in
Steps~3--9, each subalgebra $\al B|_I$ can be generated in a constant
number of steps that depends on $\class{K}$ only (and is
independent of the size of the input).
Therefore, the running time of the algorithm
is determined by the number $\binom{n}{\ddd-1}=O(n^{\ddd-1})$ of
iterations of the outermost loop (Steps~2--10)
and the time needed for computing the designated witnesses 
that are added to $R$ in each iteration (Steps~3--9) of the loop, 
which is bounded above by $O(n)$.
Thus, Algorithm~\ref{alg-locrep} runs in $O(n^\ddd)$ time.
\end{proof}

\subsection{The equivalence of $\SMP(\class{K})$ and ~$\SMP(\HH\SSS\class{K})$}
\label{ssec-HS}
In this subsection we prove that the answer to Question~2 in the Introduction
is YES.

\begin{thm}
\label{thm-smp-hom}
The
decision problems $\SMP(\class{K})$ and $\SMP(\HH\SSS\class{K})$ are
polynomial time equivalent.
\end{thm}

\begin{algorithm}
\caption{Reduction of $\SMP(\HH\SSS\class{K})$ to $\SMP(\class{K})$}
\label{alg-smp-red-smp} 
\begin{algorithmic}[1] 
\REQUIRE $c_1,\dots,c_k,c_{k+1}\in\al C_1\times\dots\times\al C_n$ 
with $\al C_1,\dots,\al C_n\in\HH\SSS\class{K}$, $n\geq\ddd$
\ENSURE Is $c_{k+1}$ in the subalgebra $\al D$ of
$\al C_1\times\dots\times\al C_n$ generated by $c_1,\dots,c_k$?
\FOR{$i=1,\dots,n$}
\STATE find $\al A_i\in\class{K}$, $\al B_i\le\al A_i$ 
and $\theta_i\in\Con(\al B_i)$
such that $\al C_i=\al B_i/\theta_i$
\ENDFOR
\FOR{$j=1,\dots,k+1$}
\STATE find $a_j\in B_1\times\dots\times B_n$ such that
$a_j/(\theta_1\times\dots\times\theta_n) = c_j$
\ENDFOR
\STATE $G := \{a_1,\dots,a_k\}$ (let $\al B$ denote the subalgebra of $\al{B}_1\times\dots\times \al{B}_n$ generated by $G$)
\STATE Run $\LocalRep(\class{K})$ with the input
  $a_1,\dots,a_k\in\al{B}_1\times\dots\times\al{B}_n$,
  $\theta_i\in\Con(\al{B}_i)$
  ($i\in[n]$) to get output $R$ (a set; the designations will not play a role)
\STATE $\LLL:=R$, $\FFF:=\emptyset$
\FOR{$i=1,\dots,n$}
\STATE generate $\al B|_i$ by $G|_i=\{a_1|_i,\dots,a_k|_i\}$,
and simultaneously,
\FOR{each new $\beta$ in $\al B|_i$}
\STATE find an element $b_{i,\beta}$
generated by $a_1,\dots,a_k$
satisfying $b_{i,\beta}|_i=\beta$ 
\FOR{all $\gamma\equiv_{\theta_i}\beta$ ($\gamma\in B_i$)}
\STATE add to $\FFF$ the tuple $c_{i,\gamma}$ defined by
$c_{i,\gamma}|_i=\gamma$ and
$c_{i,\gamma}|_{[n]\setminus\{i\}}= b_{i,\beta}|_{[n]\setminus\{i\}}$ 
\ENDFOR
\ENDFOR
\ENDFOR
\STATE run $\SMP(\class{K})$ with the input 
$G\cup \LLL\cup \FFF\subseteq\al A_1\times\dots\times\al A_n$ and 
$a_{k+1}\in\al A_1\times\dots\times\al A_n$ 
($\al A_1,\dots,\al A_n\in\class{K}$), to
get an answer $\textsf{A}=\textrm{YES}$ or $\textsf{A}=\textrm{NO}$
\STATE \return\  $\textsf{A}$
\end{algorithmic}
\end{algorithm}

\begin{proof}
$\SMP(\class{K})$ is a subproblem of $\SMP(\HH\SSS\class{K})$, so
$\SMP(\class{K})$ is clearly polynomial time reducible to 
$\SMP(\HH\SSS\class{K})$. For the converse we will show that
Algorithm~\ref{alg-smp-red-smp} reduces $\SMP(\HH\SSS\class{K})$
to $\SMP(\class{K})$ in polynomial time.

In Steps~1--3 Algorithm~\ref{alg-smp-red-smp} finds the algebras $\al A_i\in\class{K}$, 
their subalgebras $\al B_i$ and their congruences $\theta_i$ such that 
the algebras $\al C_i$
in the input are $\al C_i=\al B_i/\theta_i$ ($i\in[n]$). 
In Steps~4--6 tuples 
$a_1,\dots,a_{k+1}\in\al B_1\times\dots\times\al B_n$~%
$(\le\al A_1\times\dots\times\al A_n)$ 
are found such that 
for the product congruence 
$\theta:=\theta_1\times\dots\times\theta_n$ we have
$c_j=a_j/\theta$ for all $j\in[k+1]$. 
Thus, as we explained in the paragraph preceding Theorem~\ref{thm-gen-sat},
the set $G=\{a_1,\dots,a_k\}$ obtained in Step~7 is a generating set for a subalgebra
$\al B$ of $\al B_1\times\dots\times\al B_n$ 
with the following properties: $\al B[\theta]/\theta {}\cong{} \al D$,
and  $c_{k+1}\in\al D$ if and only if $a_{k+1}\in \al B[\theta]$.
Therefore, Algorithm~\ref{alg-smp-red-smp} gives the correct answer in 
Steps~19--20,
provided
the set $G\cup \LLL\cup \FFF$ produced earlier in the process is a generating set
for $\al B[\theta]$. 

By Theorem~\ref{thm-gen-sat} it suffices to check that the set 
$\LLL$ constructed in Steps~8--9
satisfies condition~\eqref{it:L}, 
while the set $\FFF$ constructed in Steps~9--18
satisfies condition~\eqref{it:F} 
in Theorem~\ref{thm-gen-sat}.
For $\LLL$ this is straightforward to check. For $\FFF$ note that, given
$i\in[n]$ and $\beta\in B|_i$ as in 
the loop~10--18,
the tuple $b_{i,\beta}$
obtained in line~13
belongs to $\al B$.
Hence, when the \for\ loop in 
Steps~14--16
is performed for $\gamma=\beta$,
we get the tuple
$c_{i,\beta}=b_{i,\beta}\in B$,
which is added to
$\FFF$. This tuple can serve as the tuple denoted
$b_{i,\beta}$
in condition~\eqref{it:F} for every
$c_{i,\gamma}$
added to $\FFF$ in 
Steps~14--16.
This shows that the set $\FFF$ constructed in 
Steps~9--18
satisfies condition~\eqref{it:F}
in Theorem~\ref{thm-gen-sat}, and hence finishes
the proof of the correctness of Algorithm~\ref{alg-smp-red-smp}.

Steps~1--7 run in $O(kn)$ time, 
Steps~8--9
in $O(n^\ddd)$ time, while
Steps~10--18
in $O(n^2)$ time. So, the reduction of 
$\SMP(\HH\SSS\class{K})$ to $\SMP(\class{K})$ takes $O(kn^\ddd)$ time.
For an input of size $O(kn)$ of $\SMP(\HH\SSS\class{K})$
we get an input of size $O(kn^\ddd)$ for $\SMP(\class{K})$.
\end{proof}


\subsection{Finding compact representations: A direct algorithm}
\label{ssec-finding-crs1}

Our aim in this subsection is to present an algorithm for 
$\CompactRep(\class{K})$, which does not rely on $\SMP(\class{K})$.
The idea of the algorithm will be used in the next subsection to prove
that $\SMP(\class{K})\in\NP$.
A different algorithm for $\CompactRep(\class{K})$, which does rely on
$\SMP(\class{K})$, will be presented in Subsection~\ref{ssec-finding-crs2}.

Recall from Definitions~\ref{df-rep} and \ref{df-strep} that if $\al{B}$
is a subalgebra of 
$\al{A}_1\times\dots\times\al{A}_n$
($\al{A}_1,\dots,\al{A}_n\in\class{K}$, $n\ge\ddd$),
then for a set $R\subseteq\al{A}_1\times\dots\times\al{A}_n$ to be
a standardized representation for $\al{B}$, it has to satisfy three
conditions:
(i)~$R\subseteq\al{B}$ and each element in $R$ has a designation,
(ii)~$R$ contains a full set of local witnesses for $\al{B}$,
and (iii)~$R$ contains witnesses for all derived forks (and possible some non-derived forks) of $\al{B}$
in coordinates~$\ge\ddd$.
Algorithm~\ref{alg-locrep} (with all $\theta_i=0$) yields, in polynomial time,
a partial standardized representation $R_0$ consisting of a full set
of designated
local witnesses for $\al{B}$. Therefore, the task we are left with is
to find a way to expand $R_0$ to $R$ so that both conditions (i) and (iii)
are met. To enforce (i) without relying on $\SMP(\class{K})$ we have
to make sure that every new element we are adding to $R$ is obtained by
applying operations of $\al{B}$ to elements of $R$ or the given
generators of $\al{B}$. To achieve (iii), we need to ensure that we add
sufficiently many elements to get designated witnesses for all derived forks
of $\al{B}$ in coordinates~$\ge\ddd$.
The next lemma shows how this can be done.

\begin{lem}
  \label{lm-char-cr}
  Let $\al{B}$ be a subalgebra of a product $\al{A}_1\times\dots\times\al{A}_n$
  with $\al{A}_1,\dots,\al{A}_n\in\class{K}$, $n\ge\ddd$, and assume that
  $R_1,R_2\subseteq\al{A}_1\times\dots\times\al{A}_n$ are partial standardized
  representations which
  satisfy conditions~\eqref{it:char-cr1}--\eqref{it:char-cr4}
  below: 
  \begin{enumerate}
  \item\label{it:char-cr1}
    $R_1\subseteq R_2\subseteq\al{B}$
    and the designation function for $R_2$ extends the designation function for $R_1$.
  \item\label{it:char-cr3}
    Every element of $\al{B}$ is completely representable by $R_1$
    (in particular, $R_1$ contains a full set of designated local witnesses for $\al{B}$).
  \item\label{it:char-cr4}  
    $R_2$ contains designated witnesses for the forks
    that are obtained by the
    `weak transitivity rule' (cf.~Lemma~\ref{lm-newfork})
    from forks witnessed in $R_1$;    
    that is, for every $i$ $(\ddd\le i\le n)$,
    if $R_1$ has designated witnesses for
    $(\gamma,\delta), (\beta,\delta)\in\fork_i(R_1)$, then
    $(\gamma,\beta^\gamma)\in\fork_i(R_2)$ and $R_2$ has
    designated witnesses for it.
  \end{enumerate}    
Then $R_2$ is a standardized representation for $\al{B}$.
\end{lem}  

\begin{proof}
Let $R_1,R_2\subseteq\al{A}_1\times\dots\times\al{A}_n$
be partial standardized representations, which satisfy
conditions~\eqref{it:char-cr1}--\eqref{it:char-cr4}.
Conditions~\eqref{it:char-cr1}--\eqref{it:char-cr3} imply that
$R_2$ is a partial standardized representation for $\al{B}$, which
contains a full set of
designated local witnesses for $\al{B}$.
Therefore, it remains to show that
$R_2$ contains designated witnesses for all derived forks of $\al{B}$ in
coordinates~$\ge\ddd$.

To prove this, let $m\ge\ddd$, and 
let $(\gamma,\sigma)\in\fork'_m(B)$.
Then there exists $(\gamma,\beta)\in\fork_m(B)$
such that $\sigma=\beta^\gamma$.
Hence, there exist $b,b'\in B$ such that 
$b|_{[m-1]}=b'|_{[m-1]}$ and $b|_m=\beta$, $b'|_m=\gamma$.
By assumption~\eqref{it:char-cr3},
both $b$ and $b'$ are completely
representable by $R_1$. Since $b|_{[m-1]}=b'|_{[m-1]}$,
it follows from Lemma~\ref{lm-Ts} that
\[
b^{(m-1)}=T_{m-1}(u^{(\ddd)},\hat{u}^{(\ddd)},\dots,u^{(m-1)},\hat{u}^{(m-1)},
\bar{b^I}^{[m-1]})=(b')^{(m-1)}.
\]
Let $\delta:=b^{(m-1)}|_m\,\bigl(=(b')^{(m-1)}|_m\bigr)$.
Since $b$ and $b'$ are completely representable  
by $R_1$, $R_1$ contains designated witnesses $(u,u')$ and $(v,v')$ 
for the forks $(\beta,\delta),(\gamma,\delta)\in\fork_m(R_1)$.
Now, assumption~\eqref{it:char-cr4}
makes sure that in this situation, $R_2$ contains designated 
witnesses for the fork $(\gamma,\beta^\gamma)=(\gamma,\sigma)$.
This completes the proof.
\end{proof}

  Recall that $\FNP$ denotes the function problem version of $\NP$. More precisely, a binary relation $R(x,y)$,
  where the size of $y$ is polynomial in the size of $x$, is in $\FNP$ if and only if there is a deterministic
  polynomial time algorithm that can determine whether $R(x,y)$ holds given both $x$ and $y$.

\begin{thm}
  \label{thm-cr-direct}
  Algorithm~\ref{alg-cr-direct} solves $\CompactRep(\class{K})$
  in polynomial time by making multiple calls to an oracle
  $\NeedForkWitnesses(\class{K})$ in $\FNP$.
\end{thm}

\begin{proof}
First we define the oracle.

\smallskip

\noindent $\NeedForkWitnesses(\class{K})$:
\begin{itemize}
\item
INPUT:  
Partial standardized representation
$R\subseteq \al A_1\times\dots\times\al A_n$
($\al A_1,\ldots,\al A_n\in\class{K}$, $n\ge\ddd$)
such that
$R$ contains designated local witnesses for the
subalgebra $\al{R}^*$ of
$\al A_1\times\dots\times\al A_n$ generated by $R$
\item
OUTPUT:\\
  $\bigl(\textrm{YES},f,
  \bigl((u_j^{(\ddd)},\hat{u}_j^{(\ddd)},\dots,u_j^{(n)},\hat{u}_j^{(n)},
  \bar{w_j^I})\bigr)_{j=1}^{\ell},b_1,\dots,b_\ell,b,\tilde{R}\bigr)$
  where
  \begin{enumerate}
  \item\label{it:needforks1}
    $f\in F$ is an $\ell$-ary operation symbol
    (for some $\ell$),
  \item\label{it:needforks2}  
    each pair $(u_j^{(m)},\hat{u}_j^{(m)})$ ($j\in[\ell]$, $\ddd\le m\le n$)
    is designated to witness a fork in $\fork_m(R)$,
  \item\label{it:needforks3}
    the tuples $\bar{w_j^I}=(w_j^I)_{I\in\binom{[n]}{I}}$ ($j\in[\ell]$)
    are such that 
    each $w_j^I$ is a designated local witness of $b_j|_I\in R|_I$,
  \item\label{it:needforks4}
    $b_j=T_n(u_j^{(\ddd)},\hat{u}_j^{(\ddd)},\dots,u_j^{(n)},\hat{u}_j^{(n)},
    \bar{w_j^I})$ (with $e=1$) for each $j\in[\ell]$,
  \item\label{it:needforks5}  
    $b=f(b_1,\dots,b_{\ell})$,
  \item\label{it:needforks6}
    $b$ is not completely representable by $R$,
  \item\label{it:needforks7}
    $\tilde{R}=R\cup S'\cup S$
    is the partial standardized representation 
    obtained from the output of $\IsRepresentable(\class{K})$ run on the
    input $b,R$ (hence, $b$ is completely representable by $\tilde{R}$).
  \end{enumerate}  
NO, if $f$,
$(u_j^{(\ddd)},\hat{u}_j^{(\ddd)},\dots,u_j^{(n)},\hat{u}_j^{(n)},\bar{w_j^I})$
($j\in[\ell]$), $b_1,\dots,b_\ell$, and $b$ with these properties
do not exist.
\end{itemize}

To show that $\NeedForkWitnesses(\class{K})$ is in $\FNP$, we have to verify
that the length of the output $(\textrm{YES},\dots,\tilde{R})$
is bounded by a polynomial of the length of the input $R$, and
there is a polynomial time algorithm which, when $R$ and
$(\textrm{YES},\dots,\tilde{R})$ are both given, determines whether or not
conditions
\eqref{it:needforks1}--\eqref{it:needforks7} hold.
The first requirement holds, because the length of the output is dominated by
the length of $\tilde{R}$, and 
both $R$ and $\tilde{R}$ have length
$\Theta(n^\ddd)$, as they are partial standardized representations for 
$\al{R}^*$ which contain a full set of designated local witnesses for
$\al{R}^*$.

To prove that the second requirement holds as well, 
suppose that we are given a partial standardized representation
$R\subseteq \al A_1\times\dots\times\al A_n$
($\al A_1,\ldots,\al A_n\in\class{K}$, $n\ge\ddd$)
such that
$R$ contains designated local witnesses for $\al{R}^*$, and we are
given a tuple 
\[
\bigl(f,
  \bigl((u_j^{(\ddd)},\hat{u}_j^{(\ddd)},\dots,u_j^{(n)},\hat{u}_j^{(n)},
  \bar{w_j^I})\bigr)_{j=1}^{\ell},b_1,\dots,b_\ell,b,\tilde{R}\bigr).
\]
Conditions~\eqref{it:needforks1}--\eqref{it:needforks3}
can be checked in $O(1)+O(n^2)+O(n^\ddd)=O(n^\ddd)$ time
  as $\ell$ is bounded by a constant that depends only on $\class{K}$.
By \eqref{eq:cTn}, the circuit complexity of $T_n$ (with $e=1$)
is $O(n^{\ddd+1})$, hence checking condition~\eqref{it:needforks4} requires
$\ell\cdot O(n^{\ddd+2})=O(n^{\ddd+2})$ time.
Condition~\eqref{it:needforks5} can be checked in $O(n)$ time.
Finally, conditions~\eqref{it:needforks6}--\eqref{it:needforks7}  
can be checked by running Algorithm~\ref{alg-ir} for
$\IsRepresentable(\class{K})$ with input $b$ and $R$,
which requires $O(n^{\ddd+2})$ time, and comparing $\tilde{R}$ with its output.
Note here that $b$ and $R$ form a correct input for
Algorithm~\ref{alg-ir} because,
by construction, $b\in\al{R}^*$,
so our requirement on $R$ guarantees that appropriate
designated local witnesses for $b$ exist in $R$.

\begin{algorithm}
\caption{Algorithm for $\CompactRep(\class{K})$ (direct)}
\label{alg-cr-direct}
\begin{algorithmic}[1] 
\REQUIRE $a_1,\dots,a_k\in \al A_1\times\dots\al A_n$ with 
$\al A_1,\dots,\al A_n\in\class{K}$ ($n\ge\ddd$)
\ENSURE Standardized representation $R$ for the subalgebra $\al B$
of $\al A_1\times\dots\times\al A_n$ generated by $a_1,\dots,a_k$
\STATE
Run Algorithm~\ref{alg-locrep} for
  $\LocalRep(\class{K})$ with input
  $a_1,\dots, a_k\in\al{A}_1\times\dots\times\al{A}_n$ and
  $\theta_i=0$ ($i\in[n]$) to get output $R$
\FOR{$b\in\{a_1,\dots,a_k\}$}
\STATE run Algorithm 1 for $\IsRepresentable(\class{K})$ with 
input $b$ and $R$ to get output (YES, $S'$, $S$) or (NO, $S'$, $S$)
\STATE $R:=R\cup S'\cup S$
\ENDFOR
\WHILE{$\NeedForkWitnesses(\class{K})$ returns $(\textrm{YES},\dots,\tilde{R})$}
\STATE $R:=\tilde{R}$  
\ENDWHILE
\FOR{$m=\ddd,\dots,n$}
\FOR{all $(\gamma,\delta), (\beta,\delta)\in \fork_m(R)$ 
  which have designated witnesses $(v,\hat{v})$, $(u,\hat{u})$, respectively,
  in $R$}
\IF{$R$ contains no designated witnesses for 
$(\gamma,\beta^\gamma)$}
\STATE
add $p(p(v,\hat{v},\hat{u}),p(v,\hat{v},\hat{v}),v)$ and $p(u,v,v)$
to $R$ and designate them to witness the fork 
$(\gamma,\beta^\gamma)\in\fork_m(R)$
\ENDIF
\ENDFOR
\ENDFOR
\STATE \return\ $R$
\end{algorithmic}
\end{algorithm}

Now we turn to the analysis of Algorithm~\ref{alg-cr-direct}.
To prove its correctness,
let 
$a_1,\ldots,a_k\in\al A_1\times\dots\times\al A_n$ 
($\al A_1,\dots,\al A_n\in\class{K}$) be an
input for $\CompactRep(\class{K})$ with $n\ge\ddd$, and let $\al B$ denote
the subalgebra of $\al A_1\times\dots\times\al A_n$ generated by
$a_1,\dots,a_k$.
We have to show that the set $R$ returned in Step~16 of
Algorithm~\ref{alg-cr-direct}
is a standardized representation for $\al B$.

Notice that during Algorithm~\ref{alg-cr-direct}
we only add elements to $R$,
and every element added to $R$ has a designation.
Therefore, at every point in the algorithm, the current version of $R$
is a partial standardized representation, which extends
all previous versions.
The description of $\LocalRep(\class{K})$ implies that
the version of $R$ produced in Step~1, which we will denote by $R_0$,
is a standardized representation for $\al{B}$ that is a full set
of designated local witnesses for $\al{B}$. In particular,
$R_0\subseteq\al{B}$.

Next we want to argue that
\begin{enumerate}
\item\label{it:forwhile1}
  throughout all iterations of the \for\ loop in Steps~2--5 and
  the \while\ loop in Steps~6--8,
  $R$ satisfies
  \begin{equation}\label{eq-RinB}
     R_0\subseteq R\subseteq\al{B};
  \end{equation}  
\item\label{it:forwhile2}
  in every iteration of these two loops,
  the algorithms for $\IsRepresentable(\class{K})$ and
  $\NeedForkWitnesses(\class{K})$, respectively, have a correct input; and
\item\label{it:forwhile3}
  Algorithm~\ref{alg-cr-direct}
  exits the \while\ loop, i.e., $\NeedForkWitnesses(\class{K})$ answers $\textrm{NO}$ for the first time,
  after polynomially many iterations.
\end{enumerate}
Statements \eqref{it:forwhile1}--\eqref{it:forwhile2} follow by induction
on the number of iterations by observing that
\eqref{eq-RinB} holds for $R=R_0$ at the start of the first iteration
of the \for\ loop; moreover, the descriptions of $\IsRepresentable(\class{K})$
and $\NeedForkWitnesses(\class{K})$ ensure that
whenever \eqref{eq-RinB} holds at the start of an iteration of 
the \for\ loop, $R$ together with $b\in\{a_1,\dots,a_k\}$ is a correct
input for $\IsRepresentable(\class{K})$, and the new version of $R$ produced
also satisfies \eqref{eq-RinB}; similarly,
whenever \eqref{eq-RinB} holds at the start of an iteration of 
the \while\ loop, $R$ is a correct
input for $\NeedForkWitnesses(\class{K})$, and the new version of $R$ produced
also satisfies \eqref{eq-RinB}.

For~\eqref{it:forwhile3}, notice that the \for\ loop is repeated exactly $k$
times, and the \while\ loop is iterated as long as new designated pairs of
fork witnesses can be added to $R$ by such an iteration. 
By statement~\eqref{it:forwhile1} above,
all these pairs of fork witnesses are witnessing forks in $\al{B}$
in coordinates $\ddd,\dots,n$,
therefore their number is at most $\aaa_\class{K}^2(n-\ddd)$, where
$\aaa_\class{K}:=\max\{|\al{A}|:\al{A}\in\class{K}\}$.
Hence, after at most $\aaa_\class{K}^2(n-\ddd)$ iterations of the \while\ loop
in Steps~6--8,
we have to get a $\textrm{NO}$ answer, and exit the loop.

Let $R_1$ denote the version of $R$
after Step~8, and let $R_2$ denote the $R$
returned at the end of Algorithm~\ref{alg-cr-direct} in Step~16.
We will prove that $R_2$ is a standardized representation for $\al{B}$
by showing that $R_1$ and $R_2$ satisfy
conditions~\eqref{it:char-cr1}--\eqref{it:char-cr4} of Lemma~\ref{lm-char-cr}.
Clearly, in Steps~9--15 
of Algorithm~\ref{alg-cr-direct} exactly
those new fork witnesses are added to $R_1$ which yield $R_2$ so that
condition~\eqref{it:char-cr4} of Lemma~\ref{lm-char-cr} is satisfied.
Since \eqref{eq-RinB} holds for $R=R_1$, and every element $R_2\setminus R_1$
is obtained from elements of $R_1$
by applying the term $p$, we have that
$R_0\subseteq R_1\subseteq R_2\subseteq\al{B}$.
This proves
condition~\eqref{it:char-cr1}
 of Lemma~\ref{lm-char-cr}.
It remains to prove condition~\eqref{it:char-cr3}.

Let $\bar{R}_1$ denote the set of all elements of
$\al{A}_1\times\dots\times\al{A}_n$ that are completely representable by $R_1$.
By Corollary~\ref{cor-representable}, 
$R_1\subseteq\al{B}$ implies that $\bar{R}_1\subseteq\al{B}$.
Next we want to prove that $\bar{R}_1$ is in fact a subalgebra of $\al{B}$.
Let $f\in F$ be an operation symbol, say $f$ is
$\ell$-ary, let $b_1,\dots,b_\ell\in\bar{R}_1$, and let
$b:=f(b_1,\dots,b_\ell)$. We need to show that $b\in\bar{R}_1$.
Since $b_1,\dots,b_\ell\in\bar{R}_1$, i.e., $b_1,\dots,b_\ell$
are completely representable by $R_1$, there exist
$(u_j^{(\ddd)},\hat{u}_j^{(\ddd)},\dots,u_j^{(n)},\hat{u}_j^{(n)},\bar{w_j^I})$
for $j\in[\ell]$ such that 
conditions~\eqref{it:needforks2}--\eqref{it:needforks4}
in the description of $\NeedForkWitnesses(\class{K})$ hold for $R=R_1$.
By the choice of $f$ and definition of $b$ above, we have that
actually, all
conditions~\eqref{it:needforks1}--\eqref{it:needforks5}
in the description of $\NeedForkWitnesses(\class{K})$ hold for $R=R_1$.
Since $R_1$ is the version of $R$ in Step~8 
of Algorithm~\ref{alg-cr-direct},
the last iteration of the \while\ loop ran with $R=R_1$ and produced the
answer $\textrm{NO}$. Thus, $b$ must be completely representable by $R_1$,
which proves that $b\in\bar{R}_1$.
Finally, notice that the \for\ loop in Steps~2--5 of
Algorithm~\ref{alg-cr-direct} ensure that the generators $a_1,\dots,a_k$
of $\al{B}$ belong to $\bar{R}_1$. Thus, $\bar{R}_1=\al{B}$, so
condition~\eqref{it:char-cr3} of Lemma~\ref{lm-char-cr} also holds.
This completes the proof of the correctness of Algorithm~\ref{alg-cr-direct}.

To estimate the running time of Algorithm~\ref{alg-cr-direct},
recall that Algorithm~\ref{alg-locrep} runs in $O(n^\ddd)$ time,
so this is the running time of Step~1.
In Steps~2--5, Algorithm~\ref{alg-ir} is called $k$ times, and each
iteration runs in $O(n^{\ddd+2})$ time, therefore the completion of Steps~2--5
requires $O(kn^{\ddd+2})$ time.
As we saw earlier in this proof, 
in Steps~6--8
the oracle $\NeedForkWitnesses(\class{K})$ is called
$O(n)$ times, therefore the running time of these steps is $O(n)$.
Finally, in Steps~9--15, 
the double \for\ loop is iterated
$O(n)$ times, and each iteration requires $O(n)$ time. Therefore,
Steps~9--15 
require $O(n^2)$ time. Hence, the overall running time of
Algorithm~\ref{alg-cr-direct} is $O(kn^{\ddd+2})$.
\end{proof}

\subsection{$\SMP(\class{K})$ is in $\NP$.}
\label{ssec-smp-np}

The main result of this subsection is the following theorem.

\begin{thm}
  \label{thm-short-rep}
  Let $a_1,\dots,a_k,b\in\al{A}_1\times\dots\times\al{A}_n$
  $(\al{A}_1,\dots,\al{A}_n\in\class{K})$ be an input for
  $\SMP(\class{K})$. If $b$ is in the subalgebra of
  $\al{A}_1\times\dots\times\al{A}_n$ generated by $\{a_1,\dots,a_k\}$,
  then there exists
  a term $t$ in the language $F$
  of $\class{K}$ such that
  \begin{itemize}
  \item
    $t(a_1,\dots,a_k)=b$ and
  \item
    $\crcf{F}{t}$ has size $O(kn^{\ddd+2})$.
  \end{itemize}
  Consequently, $\SMP(\class{K})\in\NP$.
\end{thm}  

The proof relies on Lemma~\ref{lm-short-cr} below, which shows
that by recording the computations  
Algorithm~\ref{alg-cr-direct} performs on a given input $a_1,\dots,a_k$,
we can build a polynomial size circuit that produces a
standardized representation for
the algebra with generators $a_1,\dots,a_k$.
The precise statement is as follows.   

\begin{lem}
  \label{lm-short-cr}
  Let $a_1,\dots,a_k\in\al{A}_1\times\dots\times\al{A}_n$ with
  $\al{A}_1,\dots,\al{A}_n\in\class{K}$ $(n\ge\ddd)$, and let
  $\al{B}$ be the subalgebra of $\al{A}_1\times\dots\times\al{A}_n$
  generated by $\{a_1,\dots,a_k\}$.
  There exists an $F$-circuit $C$
  with the following properties:
  \begin{enumerate}
  \item\label{it:short-cr3}
    $C$ has $k$ input nodes and $O(n^{\ddd-1})$ labeled output nodes
    such that for the input $a_1,\dots,a_k$, the collection of outputs of $C$
    --- together with the labeling of the outputs --- is 
    a standardized representation $R$ for $\al{B}$ computed by
    Algorithm~\ref{alg-cr-direct}.
  \item\label{it:short-cr2}
    The size of $C$
    is $O(kn^{\ddd+2})$.
  \end{enumerate}  
\end{lem}

\begin{proof}
  $C$ will be constructed in five phases, the first four of which produce \mbox{$F\cup\{P\}$}-circuits,
  which correspond to Step~1, Steps~2--5, Steps~6--8 and Steps~9--15 of  Algorithm~\ref{alg-cr-direct}.

  In Step~1, i.e., during the execution of Algorithm~\ref{alg-locrep}
  on the input $a_1,\dots,a_k$,
  for every choice of $I\in\binom{[n]}{\ddd-1}$ and
  $\bar{b}\in\al{B}|_I$, the designated local witness $r_{I,\bar{b}}\in R$ is
  obtained as $r_{I,\bar{b}}:=t_{I,\bar{b}}(a_1,\dots,a_k)$ for a term
  $t_{I,\bar{b}}$ (in the language $F$)
  such that $\bar{b}=t_{I,\bar{b}}(a_1|_I,\dots,a_k|_I)$.
  Since $\bigl|\al{B}|_I\bigr|\le\aaa_{\class{K}}^{\ddd-1}$ for every such $I$,
  where $\aaa_{\class{K}}:=\max\{|\al{A}|:\al{A}\in\class{K}\}$ is a constant
  independent of $a_1,\dots,a_k$, the terms $t_{I,\bar{b}}$ can be chosen
  to be of constant size ($\le \aaa_{\class{K}}^{\ddd-1}$).
  Therefore, an $F$-circuit $C_1$ which, for the input $a_1,\dots,a_k$, 
  computes the partial standardized representation
  obtained in Step~1 of Algorithm~\ref{alg-cr-direct}
  (i.e., a full set of designated local witnesses for $\al{B}$) 
  can be constructed as follows: $C_1$ is the disjoint union of the
  $F$-circuits $\crc{t_{I,\bar{b}}}$, except that they all have the same input
  nodes for $a_1,\dots,a_k$. The output $r_{I,\bar{b}}$
  of each $\crc{t_{I,\bar{b}}}$ is an output of $C_1$ as well, with
  the appropriate label recording its designation.
  Clearly, $C_1$ has size $O(n^{\ddd-1})$.
  
  In Steps~2--5,  Algorithm~\ref{alg-ir} is run on each input tuple
  $b\in\{a_1,\dots,a_k\}$
  together with
  the partial standardized representation $R$ constructed
  by that point in the
  algorithm (which  contains a full set
  of designated local witnesses for $\al{B}$ by
  Step~1). By Remark~\ref{rem-alg1}, each of these $k$
    applications of Algorithm~\ref{alg-ir} is a computation that can
    be performed by the $P$-circuit $\crcplus{T_n}$ on
    designated local witnesses in $R$ and on appropriately chosen
    designated fork witnesses, some of which are added to $R$
    during the computation.
    Therefore,
  an $F\cup\{P\}$-circuit $C_2$ which, for the inputs $a_1,\dots,a_k$, 
  computes the partial standardized representation for $\al{B}$
  obtained in Step~5 of Algorithm~\ref{alg-cr-direct}
  can be constructed as follows:
  add to $C_1$ $k$ disjoint copies of the $P$-circuit $\crcplus{T_n}$
  described in Remark~\ref{rem-alg1},
  one for each choice of $b\in\{a_1,\dots,a_k\}$, so that
  their input nodes are the appropriate designated local witnesses and
  fork witnesses computed earlier or during the evaluation of $T_n$ as described in Remark~\ref{rem-alg1},
  and they have appropriately
  labeled output nodes for every pair of
  designated
  fork witnesses added to $R$.
  It follows from Remark~\ref{rem-alg1} that $C_2$
  has size $O(n^{\ddd-1})+k\cdot O(n^{\ddd+1})=O(kn^{\ddd+1})$.
  
Similarly, in Steps~6--8,
in every iteration of the \while\ loop where
  the output of 
  $\NeedForkWitnesses(\class{K})$ has the form
  \[
  X=\bigl(\textrm{YES},f,
  \bigl((u_j^{(\ddd)},\hat{u}_j^{(\ddd)},\dots,u_j^{(n)},\hat{u}_j^{(n)},
  \bar{w_j^I})\bigr)_{j=1}^{\ell},b_1,\dots,b_\ell,b,\tilde{R}\bigr)
  \]
  such that conditions~\eqref{it:needforks1}--\eqref{it:needforks7}
  in the description of $\NeedForkWitnesses(\class{K})$ are satisfied,
  the fork witnesses $u_j^{(m)},\hat{u}_j^{(m)}$ and local witnesses
  $w_j^I$ had already been included in the current version of $R$,
  $b_1,\dots,b_\ell$ and $b$ can be computed as stated in
  conditions~\eqref{it:needforks4}--\eqref{it:needforks5}, and
  Algorithm~\ref{alg-ir} can be run on $b$ (and the current version of $R$)
  to find the fork witnesses that have to be added to $R$ to get $\tilde{R}$
  (the new $R$).
  We saw in the proof of Theorem~\ref{thm-cr-direct} that the \while\ loop
  is iterated at most $O(n)$ times. Therefore, 
  an $F\cup\{P\}$-circuit $C_3$ which, for the input $a_1,\dots,a_k$, 
  computes the partial standardized representation for $\al{B}$
  obtained in Step~8 of Algorithm~\ref{alg-cr-direct}
  can be constructed as follows:
  add to $C_2$, one after the other, as many as necessary (at most $O(n)$)
  circuits of the following form:
  \begin{itemize}
    \renewcommand{\labelitemi}{$\circ$}
    \item
  $\ell$ disjoint copies of the $P$-circuit $\crc{T_n}$ 
  whose inputs are (already computed) fork and local witnesses
  $u_j^{(\ddd)},\hat{u}_j^{(\ddd)},\dots,u_j^{(n)},\hat{u}_j^{(n)}, \bar{w_j^I}$
  ($j\in[\ell]$), and
  \item
  whose outputs $b_1,\dots,b_\ell$ are the inputs for an
  $f$-gate, which produces $b$,
  \item
followed by a copy of the $P$-circuit $\crcplus{T_n}$ from Remark~\ref{rem-alg1}
with appropriately labeled output nodes for every pair of designated fork witnesses added to
$R$, as described in Remark~\ref{rem-alg1}.
  \end{itemize}
  Since the
 $P$-circuits $\crc{T_n}$ and $\crcplus{T_n}$ have size $O(n^{\ddd+1})$ by \eqref{eq:cTn} and  Remark~\ref{rem-alg1},
  and the arities $\ell$ of the operation symbols $f$ are bounded by a constant
  independent of the input,
  we get that the size of $C_3$ is
  \[
  O(kn^{\ddd+1})+O(n)\cdot\bigl(\ell\cdot O(n^{\ddd+1})+1+O(n^{\ddd+1})\bigr)
  =O(kn^{\ddd+2}).
  \]

  Finally, to compute the additional fork witnesses included in $R$ in
  Steps~ 9--15
  of Algorithm~\ref{alg-cr-direct}, we have to add at most
  $O(n)$ many $P$-circuits to $C_3$,
  each one of size $1$ or $3$
   (see the definition of $p$ in \eqref{eq-ps}).

  The resulting $F\cup\{P\}$-circuit, $C_4$, is therefore of size $O(kn^{\ddd+2})$.
  Finally, since $P$ is a term in $F$, gates of type $P$ can be easily eliminated
  from $C_4$ 
  by replacing each of them with the corresponding $F$-circuit (of constant size).
 Thus we obtain an $F$-circuit $C$ from $C_4$ which satisfies
  conditions~\eqref{it:short-cr3}--\eqref{it:short-cr2}
   of the lemma.
\end{proof}

\begin{proof}[Proof of Theorem~\ref{thm-short-rep}]
  To prove the first statement of the theorem,    
  consider an input $a_1,\dots,a_k,b\in\al{A}_1\times\dots\times\al{A}_n$
  $(\al{A}_1,\dots,\al{A}_n\in\class{K})$ for
  $\SMP(\class{K})$, let $\al{B}$ denote the subalgebra of
  $\al{A}_1\times\dots\times\al{A}_n$ generated by $\{a_1,\dots,a_k\}$,
  and assume that $b\in\al{B}$.
  For $n<\ddd$ the assertion of the theorem is trivial ---
  by the same argument
  that we used in the second paragraph of the proof of
  Lemma~\ref{lm-short-cr} --- 
  since $|\al{B}|\le \aaa_{\class{K}}^{\ddd-1}$ for a constant $\aaa_{\class{K}}$
  that depends only on
  $\class{K}$. Therefore, we will assume from now on that $n\geq\ddd$.
  
  Let $C$ be a circuit satisfying
  properties~\eqref{it:short-cr3}--\eqref{it:short-cr2}
  of Lemma~\ref{lm-short-cr}, and let $R$ denote the standardized
  representation for $\al{B}$ that is computed by $C$ from $a_1,\dots,a_k$.
  As we noted in Remark~\ref{rm-representable}, 
  Lemma~\ref{lm-rep-generates}
  shows that if $b\in\al{B}$, then
  the equality \eqref{eq-rep-b},
  \[ b = T_n(u^{(\ddd)},\hat{u}^{(\ddd)},\dots,u^{(n)},\hat{u}^{(n)}, \bar{b^I)}, \]
  holds with all arguments of $T_n$ members of $R$; namely, each pair
  $u^{(m)},\hat{u}^{(m)}$ ($m=\ddd,\dots,n$) is
  designated to witness a derived fork in the $m$-th coordinate,
  and the tuple $\bar{b^I}$ ($I\in\binom{[n]}{\ddd-1}$)
  consists of some designated
  local witnesses for $\al{B}$.
Hence, $C$ has a subcircuit $U$ which outputs exactly the arguments of $T_n$ in the above expression for $b$. 
Since  $C$ is an $F$-circuit of size $O(kn^{d+2})$, so is $U$.
Identifying the outputs of $U$ with the inputs of an $F$-circuit for $T_n$, which is obtained from its $P$-circuit
by replacing each $P$-gate by the corresponding $F$-circuit, we get an $F$-circuit for a term
$t$ such that $t(a_1,\dots,a_k) = b$. 
 The size of $\crcf{F}{T_n}$ is constant times the size of $\crcf{P}{T_n}$, which is $O(n^{d+1})$ by~\eqref{eq:cTn}
 (for $e=1$).
Therefore the size of $\crcf{F}{t}$ is $O(kn^{\ddd+2})+O(n^{\ddd+1})=O(kn^{\ddd+2})$.
This proves the first statement of Theorem~\ref{thm-short-rep}.

 By this statement, we have a polynomial size certificate for
 ``$b$ is in the subalgebra generated by $\{a_1,\dots,a_k\}$''
 for every input $a_1,\dots,a_k,b$ for $\SMP(\class{K})$.
 This proves that $\SMP(\class{K})$ lies in $\NP$.
\end{proof}

\begin{rem}\label{rm-witness-term}
 In Theorem~\ref{thm-short-rep}, the term $t$ witnessing that $b$ belongs to the subalgebra generated by $a_1,\dots,a_k$
 clearly has depth polynomial in $n$.  
 In the case that the $\ddd$-cube term is a Mal'tsev term, this was already observed in~\cite[Theorem 2.2]{Ma:SMP}.
 If the $\ddd$-cube term is a near-unanimity term, then the fact that $b$ belongs to the subalgebra generated by
 $a_1,\dots,a_k$ can be witnessed by polynomially many terms of constant length by the
 Baker--Pixley Theorem~\cite{BP:PIAT}.
 It is however open in the general case, whether a polynomial bound can be imposed on the length of a representation 
 of $t$ by a term rather than by a circuit, as in Theorem~\ref{thm-short-rep}. 
\end{rem}

\subsection{Finding compact representations:
The relationship between $\SMP(\class{K})$ and $\CompactRep(\class{K})$}
\label{ssec-finding-crs2} 

 The main result of this section is the following theorem.

\begin{algorithm}
\caption{Reduction of $\CompactRep(\class{K})$ to $\SMP(\class{K})$}
\label{alg-smp-red-cr} 
\begin{algorithmic}[1] 
\REQUIRE $a_1,\dots,a_k\in\al A_1\times\dots\times\al A_n$ with 
$\al A_1,\dots,\al A_n\in\class{K}$, $n\geq\ddd$
\ENSURE Standardized representation $R$ for the subalgebra $\al B$ of
$\al A_1\times\dots\times\al A_n$ generated by $a_1,\dots,a_k$
\STATE
Run Algorithm~\ref{alg-locrep} for
  $\LocalRep(\class{K})$ with input
  $a_1,\dots, a_k\in\al{A}_1\times\dots\times\al{A}_n$ and
  $\theta_i=0$ ($i\in[n]$) to get output $R_0$
\STATE $R:=R_0$
\FOR{$i=\ddd,\dots,n$ and $\gamma\in\al B|_i$}
\STATE find $b\in R_0$ with $b|_i=\gamma$
\FOR{$\beta\in\al B|_i$}
\STATE let $c\in\al A_1\times\dots\times\al A_i$ be such that  
$c|_{[i-1]}=b|_{[i-1]}$ and $c|_i=\beta^\gamma$
\STATE run $\SMP(\class{K})$ with input 
$a_1|_{[i]},\dots,a_k|_{[i]},c\in\al A_1\times\dots\times\al A_i$
\IF{answer is YES}
\FOR{$j=i+1,\dots,n$}
\STATE find $c_j\in A_j$ such that $\SMP(\class{K})$ with input
$a_1|_{[j]},\dots,a_k|_{[j]},(c,c_j)\in\al A_1\times\dots\times\al A_j$ answers YES
\STATE $c:=(c,c_j)$
\ENDFOR
\STATE add $b,c$ to $R$, and designate them to witness the fork $(\gamma,\beta^\gamma)\in\fork_i(R)$.
\ENDIF
\ENDFOR
\ENDFOR 
\STATE \return\ $R$
\end{algorithmic}
\end{algorithm}

\begin{algorithm}
\caption{Reduction of $\SMP(\class{K})$ to $\CompactRep(\class{K})$}
\label{alg-cr-red-smp} 
\begin{algorithmic}[1] 
\REQUIRE $a_1,\dots,a_k,b\in\al A_1\times\dots\times\al A_n$ with 
$\al A_1,\dots,\al A_n\in\class{K}$, $n\geq\ddd$
\ENSURE Is $b$ in the subalgebra $\al B$ of
$\al A_1\times\dots\times\al A_n$ generated by $a_1,\dots,a_k$?
\STATE Run $\CompactRep(\class{K})$ with input 
$a_1,\dots,a_k$, and let $R$ be its output (a standardized
representation for $\al B$)
\FOR{$I\in\binom{[n]}{\ddd-1}$}
\IF{$R$ contains no designated witness for $b|_I$}
\STATE \return\ NO
\ENDIF
\ENDFOR
\STATE Run Algorithm~1 for $\IsRepresentable(\class{K})$ with input
$b,R$, to get output $(\textsf{A},S',S)$ with 
$\textsf{A}=\textrm{YES}$
or $\textsf{A}=\textrm{NO}$
\STATE \return\  $\textsf{A}$
\end{algorithmic}
\end{algorithm}

\begin{thm}
\label{thm-smp-cr-equiv}
The problems
$\SMP(\class{K})$ and
$\CompactRep(\class{K})$ are polynomial time reducible to one another.
 In more detail,
  Algorithm~\ref{alg-smp-red-cr} solves $\CompactRep(\class{K})$
  in polynomial time by repeated calls of $\SMP(\class{K})$, and
  Algorithm~\ref{alg-cr-red-smp} solves $\SMP(\class{K})$
  (for inputs satisfying $n\ge\ddd$) in
  polynomial time by a single call of $\CompactRep(\class{K})$.
\end{thm}

\begin{proof}
  As we mentioned earlier, $\SMP(\class{K})$ for inputs
    $a_1,\dots,a_k,b$ that are $n$-tuples with $n<\ddd$ can be solved in
    constant time, since the algebra $\al{B}$ generated by $a_1,\dots,a_k$
    has size bounded by a constant that depends only on $\class{K}$.
    Therefore it will not restrict generality to assume $n\ge\ddd$
    for the inputs of $\SMP(\class{K})$ in Algorithm~\ref{alg-cr-red-smp},
    and the first statement of Theorem~\ref{thm-smp-cr-equiv} will follow
    from the second statement on Algorithms~\ref{alg-smp-red-cr}
    and~\ref{alg-cr-red-smp}.

    First we will discuss Algorithm~\ref{alg-smp-red-cr}.
 To prove its correctness, note that Step~1 produces a
  partial standardized representation $R_0$ for $\al{B}$, which
  is a full set of designated local witnesses for $\al{B}$.
Elements without designations are not added to $R=R_0$
later on in Algorithm~\ref{alg-smp-red-cr},
therefore we will be done if we show that
Steps~3--16 add to $R$ a pair of designated witnesses (from $\al B$) 
for all derived forks of $\al B$ in coordinates $\ge\ddd$,
  and possibly for some non-derived forks as well, but nothing else is added
  to $R$.

Lines 3--7 show that the loop in Steps~3--16 examines
each pair $(\gamma,\beta)\in\al B|_i\times\al B|_i$ for every
$i=\ddd,\dots,n$, finds 
$b=(b_1,\dots,b_n)\in R_0\,(\subseteq R\subseteq B)$
such that $b_i=\gamma$, and checks --- using $\SMP(\class{K})$ ---
whether or not the tuple $c=(b_1,\dots,b_{i-1},\beta^\gamma)$ is in the
subalgebra $\al B|_{[i]}$ generated by the elements 
$a_1|_{[i]},\dots,a_k|_{[i]}$.
If the answer is YES, then $B$ contains a tuple of the form 
$(c,c_{i+1},\dots,c_n)=(b_1,\dots,b_{i-1},\beta^\gamma,c_{i+1},\dots,c_n)$
for some $c_j\in A_j$ ($j=i+1,\dots,n$),
which will be found, coordinate-by-coordinate, by repeated applications
of $\SMP(\class{K})$ in Steps~9--12.
It is clear that in this case $b,c$ are in $\al{B}$ and witness that $(\gamma,\beta^\gamma)$ is a fork in $B$,
so these witnesses are correctly added to $R$  in Step~13.
This is the only step when $R$ changes, so this shows that the $R$ returned
in Step~17 is a partial standardized representation for $\al{B}$.

It remains to show that $R$ contains designated witnesses for all derived
forks of $\al{B}$ in coordinates $\ge\ddd$. Let $\ddd\le i\le n$, and let
$(\gamma,\beta^\gamma)\in\fork_i'(B)$ with
$(\gamma,\beta)\in\fork_i(B)$.
By statement~\eqref{it:df} of Lemma~\ref{lm-der-forks}, for every choice of
$b\in\al{B}$ with $b|_i=\gamma$ there exists an element $\hat{b}$ such that
$b$ and $\hat{b}$ are witnesses for the derived fork
$(\gamma,\beta^\gamma)\in\fork_i'(B)$, that is, such that
$\hat{b}|_{[i-1]}=b|_{[i-1]}$ and $\hat{b}|_i=\beta^\gamma$.
Hence, for the $b$ and for the $i$-tuple $c=\hat{b}|_{[i]}$
found in Steps~4 and~6 of Algorithm~\ref{alg-smp-red-cr},
$\SMP(\class{K})$ in Step~7 has to give
the answer YES, and the $n$-tuple $c$ produced in Steps~8--12 and
added to $R$ in Step~13, along with $b$,
to witness $(\gamma,\beta^\gamma)\in\fork_i'(R)$ is such a $\hat{b}$.
This finishes the proof of the correctness of Algorithm~\ref{alg-smp-red-cr}.

In Step~1, Algorithm~\ref{alg-locrep} runs in $O(n^\ddd)$ time.
Steps~3--16 require running $\SMP(\class{K})$ $O(n^2)$ times
on inputs not larger than the input for Algorithm~\ref{alg-smp-red-cr}, and adding one pair of
witnesses to $R$ no more than $O(n)$ times.
This show that Algorithm~\ref{alg-smp-red-cr} reduces $\CompactRep(\class{K})$ 
to $\SMP(\class{K})$ in $O(n^\ddd)$ time.

Now we turn to Algorithm~\ref{alg-cr-red-smp}.
The algorithm starts with a call to $\CompactRep(\class{K})$ to compute
a standardized representation $R$ for the algebra $\al B$
generated by the input tuples $a_1,\dots,a_k$.
A necessary condition for the input tuple
$b$ to be in $\al B$ is that
$b|_I\in\al B|_I$ for all $I\in\binom{[n]}{\ddd-1}$.
Since $R$ contains a full set of designated local witnesses for $\al{B}$,
$b$ will satisfy this necessary condition if and only if
$R$ contains designated local witnesses for all projections $b|_I$
$\bigl(I\in\binom{[n]}{\ddd-1}\bigr)$ of $b$.
This is being checked in 
Steps~2--6 of Algorithm~\ref{alg-cr-red-smp}; if the condition fails
for some $I\in\binom{[n]}{\ddd-1}$, the algorithm returns the correct answer
NO, meaning, $b\notin\al B$.

If the algorithm passes 
Steps~2--6 without returning NO, then
$b,R$ is a correct input for $\IsRepresentable(\class{K})$, which
checks in 
Step~7 whether $b$ is representable by $R$.
Since every tuple representable by $R$ must be in $\al B$, and conversely,
by Remark~\ref{rm-representable}, every element of $\al B$ is representable
by $R$, we get that the YES/NO answer provided by $\IsRepresentable(\class{K})$
is the correct answer to $\SMP(\class{K})$ for the given input.
This shows the correctness of Algorithm~\ref{alg-cr-red-smp}.

Steps~2--6 of Algorithm~\ref{alg-cr-red-smp} run in 
$O(n^{\ddd-1})$ time, while 
$\IsRepresentable(\class{K})$ in 
Step~7 requires $O(n^{\ddd+2})$ time.
Thus, Algorithm~\ref{alg-cr-red-smp} reduces
$\SMP(\class{K})$ to $\CompactRep(\class{K})$ in $O(n^{\ddd+2})$ time.
\end{proof}

\section{Structure Theory and the Subpower Membership Problem}
\label{sec-subpowers}

Our global assumption for this section is the following:

\begin{asm}
    $\var{V}$ is a variety with a $\ddd$-cube term, or equivalently,
    with a $(1,\ddd-1)$-parallelogram term ($\ddd>1$). 
\end{asm}    

As in Section~\ref{sec-cRep},
we do not assume that the algebras we are considering
are finite, or have a finite language,
so the results in this section hold for arbitrary algebras.

One of the main results of
\cite{KS:CAPT} is a structure theorem for the 
critical subalgebras of finite powers of algebras with a
cube (or parallelogram) term.
In this section we adapt this structure theorem to
find a new representation (different from compact representations)
for subalgebras of products of algebras in a variety $\var{V}$
with a cube term. In the next section, this representation will be used
to prove the main result of the paper.

To restate the results from \cite{KS:CAPT} that we need here, 
we introduce some terminology and notation.
Let $\al R$ be a subalgebra of a product 
$\al A^{(1)}\times\dots\times\al A^{(n)}$
of some algebras
$\al A^{(1)},\dots,\al A^{(n)}\in\var{V}$.
Let
$\al A_i:=\al{R}|_i$ for each $i\in[n]$, and
let $\al C:=\al A_1\times\dots\times\al A_n$.
So, $\al R$ is a \emph{subdirect subalgebra} of $\al C$.

We say that $\al R$ is a 
\emph{critical subalgebra} of $\al A^{(1)}\times\dots\times\al A^{(n)}$
if it has the following two properties:
\begin{itemize}
\item
$\al R$ is \emph{completely $\cap$-irreducible} in the lattice of subalgebras of
  $\al A^{(1)}\times\dots\times\al A^{(n)}$,
that is,
  whenever $\al{R}$ is an intersection of a family of subalgebras $\al{S}_j$
  ($j\in J$) of $\al A^{(1)}\times\dots\times\al A^{(n)}$, we have that
  $\al{R}=\al{S}_j$ for some $j\in J$; and
\item
$\al R$ is \emph{directly indecomposable} in the following sense:
$[n]$ cannot be partitioned 
into two nonempty sets $I$ and $J$ such that $\al R$ and 
$\al{R}|_I\times\al{R}|_J$ differ only by a permutation of 
coordinates.
\end{itemize}
  Note that if $\al{R}$ is completely $\cap$-irreducible,
  it will not have a direct decomposition
 $\al{R}|_I\times\al{R}|_J$ (up to a permutation of coordinates) 
  where $\al{R}|_I$ and $\al{R}|_J$ are both proper subalgebras of the
  corresponding products $\prod_{i\in I}\al{A}^{(i)}$ and
  $\prod_{i\in J}\al{A}^{(i)}$, respectively. However, $\al{R}$ may be
  completely $\cap$-irreducible and still directly decomposable
as $\al{R}|_I\times\al{R}|_J$ (up to a permutation of coordinates) if, say,
  $\al{R}|_J=\prod_{i\in J}\al{A}^{(i)}$
  and $\al{R}|_I$ is a completely $\cap$-irreducible subalgebra of
  $\prod_{i\in I}\al{A}^{(i)}$.

Now let us assume that $\al R$ is a critical subalgebra of 
$\al A^{(1)}\times\dots\times\al A^{(n)}$.
Continuing with the notation introduced
in the second to last paragraph, let us choose
and fix $\delta_i\in\Con(\al A_i)$ ($i\in[n]$)
such that 
$\delta:=\delta_1\times\dots\times\delta_n$
is the largest product congruence of 
$\al C$ with the property that $\al R$ is a $\delta$-saturated
subalgebra of $\al C$.
(Such a congruence exists, because the join of product congruences
of $\al{C}$ is a
product congruence, and if $\al R$ is saturated with respect to a family
of congruences of $\al C$, then it is saturated with respect to their join.)
With this notation, let
$\bar{\al R}:=\al R/\delta_{\al R}$, and let 
$\bar{\al A}_i:=\al A_i/\delta_i$ $(i\in[n])$; 
we call $\bar{\al R}$ 
\emph{the reduced representation of $\al R$}.

Theorems~2.5 and 4.1 of \cite{KS:CAPT} yield a structure theorem
for the critical subalgebras of finite 
powers $\al A^n$ of an arbitrary algebra $\al A\in\var{V}$.
The relevant proofs in \cite{KS:CAPT}, namely 
the proofs of Theorem~2.5 (and its preparatory Lemmas~2.1, 2.3, 2.4)
and Theorem~3.6 (part (3), implication $\Rightarrow$), 
carry over without any essential changes to the more general situation
when instead of subalgebras of powers $\al A^n$  with $\al A\in\var{V}$ 
we consider subalgebras of products 
$\al A^{(1)}\times\dots\times\al A^{(n)}$ with
$\al A^{(1)},\dots,\al A^{(n)}\in\var{V}$.
Thus, we get the theorem below, where
we state only those parts
of the structure theorem that we need later on, 
retaining the numbering from \cite[Theorem~2.5]{KS:CAPT}.
The superscript ${}^\flat$ in (6)$^\flat$
 indicates that 
instead of the original condition (6) we state a weaker condition
which is sufficient for our purposes. 

\begin{thm}[Cf.\ \cite{KS:CAPT}]
\label{thm-paralg}
Let $\al A^{(1)},\dots,\al A^{(n)}\in\var{V}$, and
let $\bar{\al R}$ be the reduced representation of a 
critical subalgebra $\al R$ of 
$\al A^{(1)}\times\dots\times\al A^{(n)}$.
If $n\ge \ddd$, then the following hold.
\begin{enumerate}
\item
$\bar{\al R}\le\prod_{i=1}^n\bar{\al A}_i$ is a representation of
$\bar{\al R}$ as a subdirect product of subdirectly irreducible
algebras $\bar{\al A}_i$.
\setcounter{enumi}{6}
\item[{\rm(6)}$^\flat$] \label{it:sim}
$\bar{\al A}_i$ and $\bar{\al A}_j$ are similar 
for any $i,j\in[n]$ (see Section~\ref{sec:cm}).
\item
If $n>2$, then each $\bar{\al A}_i$ has abelian monolith $\mu_i$
$(i\in[n])$.
\item \label{it:centiso} 
For the centralizers $\rho_\ell:=(0:\mu_\ell)$ of the monoliths $\mu_\ell$
$(\ell\in[n])$, the image of the composite map
\[
\bar{\al R} \stackrel{\proj_{ij}}{\rightarrow}
\bar{\al A}_i\times\bar{\al A}_j
\twoheadrightarrow\bar{\al A}_i/\rho_i\times\bar{\al A}_j/\rho_j.
\]
is the graph of an isomorphism $\bar{\al A}_i/\rho_i\to\bar{\al A}_j/\rho_j$
for any distinct $i,j\in[n]$.
\end{enumerate}
\end{thm}

Note that the homomorphism in part~\eqref{it:centiso} is the same as in
\cite[Theorem~2.5~(8)]{KS:CAPT}; the slightly different description
presented here will be more convenient later on.

\medskip

Now we are ready to discuss our representation theorem.
Let $\al B_1,\dots,\al B_n$ be
nontrivial algebras in $\var{V}$, and let $\al B$ be a
subdirect subalgebra of $\al B_1\times\dots\times\al B_n$.

First, we replace each $\al B_j$ ($j\in[n]$)
by its image under the embedding 
\begin{equation}\label{eq-repBj}
\al B_j
\hookrightarrow\prod_{\sigma\in\Irr(\al B_j)}\al B_j/\sigma,
\quad\qquad
x_j\mapsto(x_j/\sigma)_{\sigma\in\Irr(\al B_j)};
\end{equation}
thus, each $\al B_j$ is replaced by
a subdirect product of all of its subdirectly irreducible quotients
$\al B_j/\sigma$.
To set up a more convenient notation, let
\[
W:=\{(j,\sigma):j\in[n],\sigma\in\Irr(\al B_j)\},
\] 
and for each $j\in[n]$,
let $W_j:=\{j\}\times\Irr(\al B_j)$; thus, $W=W_1\cup\dots\cup W_n$.
Furthermore, for each $w=(j,\sigma)\in W$ let $\hat{\al B}_w:=\al B_j/\sigma$,
and for every element $x_j$ in $\al B_j$ let $\hat{x}_w:=x_j/\sigma$.
Then the product of the embeddings \eqref{eq-repBj} for all $j\in[n]$ 
yields an embedding
\begin{equation*}
\hat{\phantom{\al B}}
\colon \prod_{j\in[n]}\al B_j
\hookrightarrow
\prod_{w\in W}\hat{\al B}_w,
\qquad
x=(x_j)_{j\in[n]}
\mapsto \hat x:=(\hat{x}_w)_{w\in W}.
\end{equation*}
We will denote the image of $\al B$ under this embedding $\hat{\phantom{n}}$
by $\hat{\al B}$. By construction, $\hat{\al B}$ is a subdirect product
of the subdirectly irreducible algebras $\hat{\al B}_w$ ($w\in W$);
in particular,
$\hat{\al B}_w=\hat{\al B}|_w$ for all $w\in W$.
For each $w\in W$ let $\mu_w$ denote the monolith of $\hat{\al B}_w$ and
$\rho_w$ its centralizer $(0:\mu_w)$. 

Next we define a relation $\sim$ on $W$ as follows: we require $\sim$
to be reflexive, and   
for distinct $v,w\in W$ we define $v \sim w$ to hold if and only if
\begin{itemize}
\item 
the subdirectly irreducible algebras
$\hat{\al B}_v$ and $\hat{\al B}_w$ are similar with abelian monoliths
$\mu_v$ and $\mu_w$, and
\item
the image of $\hat{\al B}|_{vw}$ under the natural map
$\hat{\al B}_v\times\hat{\al B}_w\twoheadrightarrow
(\hat{\al B}_v/\rho_v)\times(\hat{\al B}_w/\rho_w)$
is the graph of an isomorphism
$\hat{\al B}_v/\rho_v\to\hat{\al B}_w/\rho_w$.
\end{itemize}
It is easy to see that $\sim$ is an equivalence relation on $W$.

Our representation theorem describes the algebra $\al B$ in terms of its
image $\hat{\al B}$, namely, it shows that $\hat{\al B}$ is determined by
its projections onto small sets of coordinates (i.e., small subsets of $W$)
and by its projections onto the blocks of $\sim$. 
A block of $\sim$ may be large,
but the image of $\hat{\al B}$ under a projection onto a block of $\sim$
has a special structure.

\begin{thm}
\label{thm-repr}
Let $\var{V}$ be a variety with a $\ddd$-cube term, 
let $\al B_1,\dots,\al B_n$ be
nontrivial algebras in $\var{V}$, and let $\al B$ be a
subdirect subalgebra of $\al B_1\times\dots\times\al B_n$.
Furthermore, let $W$, $\hat{\phantom{\al B}}$, $\hat{\al B}_w$ $(w\in W)$,
$\hat{\al B}$, and $\sim$ be as defined above.
Then, for any tuple $c\in\al B_1\times\dots\times\al B_n$,
the following conditions are equivalent:
\begin{enumerate}
\item \label{it:cB}
$c\in\al B$.
\item \label{it:c2}
  $c$ satisfies
  \begin{itemize}
  \item
    $c|_I\in\al{B}|_I$ for all $I\subseteq[n]$ with
    $|I|<\max\{\ddd,3\}$, and
  \item
  $\hat{c}|_U\in\hat{\al B}|_U$
for all blocks $U\,(\subseteq W)$ of $\sim$
of size $|U|\ge\max\{\ddd,3\}$.
\end{itemize}
\item \label{it:c3}
 $\hat{c}|_U\in\hat{\al B}|_U$
for all $U\subseteq W$ such that
\begin{itemize}
\item
$|U|<\mm\cdot\max\{\ddd,3\}$ where $\mm=\max\{|\Irr(\al B_j)|:j\in[n]\}$, or
\item
$U$ is a block of $\sim$ of size $|U|\ge\mm\cdot\max\{\ddd,3\}$.
\end{itemize}
\end{enumerate}
\end{thm}

\begin{rem}
The equivalence of conditions \eqref{it:cB} and~\eqref{it:c3} in Theorem~\ref{thm-repr}
can be restated as follows: $\hat{\al B}$ is the intersection of the 
subalgebras 
\[
\proj_U^{-1}\bigl(\proj_U(\hat{\al B})\bigr)
=\proj_U^{-1}\bigl(\hat{\al B}|_U\bigr)
\]
of $\prod_{w\in W}\hat{\al B}_w$ as $U$ runs over the subsets of $W$ 
listed in~\eqref{it:c3}.
\end{rem}

\begin{proof}[Proof of Theorem~\ref{thm-repr}]
Since $\hat{\phantom{\al B}}$ and $|_U$ 
($U\subseteq W$) are
homomorphisms, it is clear that $c\in\al B$ implies
$\hat{c}|_U\in\hat{\al B}|_U$ for all $U\subseteq W$.
This proves \eqref{it:cB} $\Rightarrow$~\eqref{it:c3}.

For the implication~\eqref{it:c3} $\Rightarrow$~\eqref{it:c2}, assume that~\eqref{it:c3} holds.
Then $\hat{c}|_U\in\hat{\al B}|_U$ for all blocks
$U$ of $\sim$, so the second statement in~\eqref{it:c2} holds.
To establish the first statement, choose $I\subseteq[n]$ such that
$|I|<\max\{\ddd,3\}$, and let $W_I:=\bigcup_{j\in I} W_j$.
Since the product of the isomorphisms 
$\al B_j\to\hat{\al B}|_{W_j}$, $ x_j\mapsto(\hat{x}_w)_{w\in W_j}$
(induced by the embeddings in \eqref{eq-repBj}) yields
an isomorphism
\[
\prod_{j\in I}\al B_j\to\prod_{j\in I}\hat{\al B}|_{W_j}
\,\Bigl(\le\prod_{w\in W_I}\hat{\al B}_w\Bigr)
,\quad
(x_j)_{j\in I}\mapsto(x_w)_{w\in W_I}\qquad  (x_j\in\al B_j),
\] 
which maps $\al{B}|_I$ onto $\hat{\al B}|_{W_I}$,
we get that $c|_I\in\al{B}|_I$ holds if and only if
$\hat{c}|_{W_I}\in\hat{\al B}|_{W_I}$.
The latter follows from assumption~\eqref{it:c3}, because
$|W_I|\le\sum_{j\in I}|W_j|=\sum_{j\in I}|\Irr(\al B_j)|\le\mm|I|$.
This completes the proof of~\eqref{it:c3} $\Rightarrow$~\eqref{it:c2}.

The remaining implication~\eqref{it:c2}~$\Rightarrow$~\eqref{it:cB} is the main statement
in Theorem~\ref{thm-repr}, which we will prove now. 
Assume that $c\notin\al B$, but $c|_I\in\al{B}|_I$ for all
$I\subseteq[n]$ with $|I|<\max\{\ddd,3\}$.
We have to show  that $\hat{c}|_U\notin\hat{\al B}|_U$
for some block $U\,(\subseteq W)$ of $\sim$
of size $|U|\ge\max\{\ddd,3\}$.

Using the assumption $c\notin\al B$
and Zorn's Lemma,
we first choose and fix a subalgebra 
$\al M$ of $\al B_1\times\dots\times\al B_n$ containing 
$\al B$, which is maximal for the property that it fails to contain $c$.
Then $\al M$ is completely $\cap$-irreducible in the
lattice of subalgebras of $\al B_1\times\dots\times\al B_n$.
Let $\{T_1,\dots,T_\ell\}$ be a partition of $[n]$ such that
$\al{M}|_{T_i}$ is directly indecomposable for every $i\in[\ell]$, and
$\al M$ differs from $\al{M}|_{T_1}\times\dots\times\al{M}|_{T_\ell}$
by a permutation of variables only; we will denote this fact by
$\al M\approx\al{M}|_{T_1}\times\dots\times\al{M}|_{T_\ell}$.
We must have $c|_T\notin\al{M}|_T$ for at least 
one block $T:=T_i$, because otherwise 
$\al M\approx\al{M}|_{T_1}\times\dots\times\al{M}|_{T_\ell}$
would imply that $c\in\al M$, contradicting the choice of $\al M$. 
Let us fix such a $T$ for the rest of the proof. Note that $|T|>1$, 
because $\al M$ is a subdirect product of 
$\al B_1,\dots,\al B_n$ (as $\al B\le\al M$), so we have
$c|_j\in\al B_j=\al{M}|_j$
for every one-element set $\{j\}\subseteq[n]$.

It follows from
$\al M\approx\al{M}|_{T_1}\times\dots\times\al{M}|_{T_\ell}$
that $\al M$ is the intersection of two subalgebras of 
$\al B_1\times\dots\times\al B_n$, as shown below:
\[
\al M\approx 
\Bigl(\al{M}|_T\times\prod_{j\in[n]\setminus T}\al B_j\Bigr)
\cap
\Bigl(\prod_{j\in T}\al B_j\times\al{M}|_{[n]\setminus T}\Bigr).
\]
Since $\al M$ is a completely $\cap$-irreducible subalgebra of 
$\al B_1\times\dots\times\al B_n$, we get that 
$\al M\approx\al{M}|_T\times\prod_{j\in[n]\setminus T}\al B_j$
or
$\al M\approx\prod_{j\in T}\al B_j\times\al{M}|_{[n]\setminus T}$.
The latter is impossible, because $\al{M}|_T\lneq\prod_{j\in T}\al B_j$
(as witnessed by $c|_T$).
Hence, for the subalgebra 
$\al R:=\al{M}|_T$ of $\prod_{j\in T}\al B_j$ we get that 
$\al M\approx\al R\times\prod_{j\in[n]\setminus T}\al B_j$. 
Furthermore,
$\al R$ is both completely
$\cap$-irreducible (because $\al M$ is)
and directly indecomposable (by construction),
so
$\al R$ is a critical subalgebra of $\prod_{j\in T}\al B_j$.
Our construction also implies that $\al{B}|_T\le\al R$
and $c|_T\notin\al R$.
Hence, $c|_T\notin\al{B}|_T$.
Thus, our assumption that $c|_I\in\al{B}|_I$ holds for all
$I\subseteq[n]$ with $|I|<\max\{\ddd,3\}$ forces that
$|T|\ge \max\{\ddd,3\}$.

Now we can apply Theorem~\ref{thm-paralg} to the
algebras $\al B_t$ ($t\in T$) in $\var{V}$,
and the critical subalgebra $\al R$ of $\prod_{t\in T}\al B_t$
where the number of factors in the product is $|T|\ge\max\{\ddd,3\}$.
Since $\al R$ is a subdirect product of the algebras $\al B_t$ ($t\in T$),
the reduced representation $\bar{\al R}$ of $\al R$ is the quotient algebra
$\al R/\delta_{\al R}$
where $\delta=\prod_{t\in T}\delta_t$
($\delta_t\in\Con(\al B_t)$) is the largest product congruence of
$\prod_{t\in T}\al B_t$ for which $\al R$ is $\delta$-saturated.
Let $\bar{\al B}_t:=\al B_t/\delta_t$ for every $t\in T$.
So, the conclusions of Theorem~\ref{thm-paralg} can be restated as follows:
\begin{enumerate}[align=left]
\item\label{it:barR} 
$\bar{\al R}\le\prod_{t\in T}\bar{\al B}_t$ is a representation of
$\bar{\al R}$ as a subdirect product of subdirectly irreducible
algebras $\bar{\al B}_t$.
\setcounter{enumi}{6}
\item[{\rm(6)}$^\flat$]
$\bar{\al B}_s$ and $\bar{\al B}_t$ are similar 
for any $s,t\in T$.
\item \label{it:abmu} 
Each $\bar{\al B}_t$ has abelian monolith $\mu_t$ $(t\in T)$.
\item
For the centralizers $\rho_\ell:=(0:\mu_\ell)$ of the monoliths $\mu_\ell$
$(\ell\in T)$, the image of the composite map
\begin{equation}
\label{eq-comp8}
\bar{\al R} \stackrel{\proj_{st}}{\rightarrow}
\bar{\al B}_s\times\bar{\al B}_t
\twoheadrightarrow\bar{\al B}_s/\rho_s\times\bar{\al B}_t/\rho_t
\end{equation}
is the graph of an isomorphism $\bar{\al B}_s/\rho_s\to\bar{\al B}_t/\rho_t$
for any distinct $s,t\in T$.
\end{enumerate}

By conclusion~\eqref{it:barR}, we have for each $t\in T$ that
$\bar{\al B}_t=\al B_t/\delta_t$
is subdirectly irreducible, so $\delta_t\in\Irr(\al B_t)$ and
$(t,\delta_t)\in W$.
Hence, the algebra $\bar{\al B}_t=\al B_t/\delta_t$
is one of the subdirect factors of $\hat{\al B}$, namely,
$\bar{\al B}_t=\al B_t/\delta_t=\hat{\al B}_w$ for
$w=(t,\delta_t)$.
This implies also that 
$\mu_t, \rho_t$ are the congruences of $\bar{\al B}_t=\hat{\al B}_w$ 
that we denoted earlier by
$\mu_w,\rho_w$. 

Let $\hat{T}:=\{(t,\delta_t):t\in T\}$.
Next we want to show that any two elements of $\hat{T}$ are
related by $\sim$.
Let $s,t\in T$, and let $v:=(s,\delta_s)$,
$w:=(t,\delta_t)$. As we noticed in the preceding paragraph,
we have that $\bar{\al B}_s=\hat{\al B}_v$,
$\mu_s=\mu_v$, $\rho_s=\rho_v$,
and 
$\bar{\al B}_t=\hat{\al B}_w$,
$\mu_t=\mu_w$, $\rho_t=\rho_w$.
If $s=t$, then $v=w$, and hence $v\sim w$ holds because $\sim$
is reflexive.
So, assume from now on that $s\not=t$. Hence $v\not= w$.
In this case, checking whether $v\sim w$ holds involves two conditions.
One is that the subdirectly irreducible algebras
$\hat{\al B}_v$ and $\hat{\al B}_w$ are similar with abelian monoliths
$\mu_v$ and $\mu_w$, which follows from conclusions  (6)$^\flat$--\eqref{it:abmu}.  
The other is that the image of 
$\hat{\al B}|_{vw}$ under the natural map
$\phi\colon\hat{\al B}_v\times\hat{\al B}_w\twoheadrightarrow
(\hat{\al B}_v/\rho_v)\times(\hat{\al B}_w/\rho_w)$,
is the graph of an isomorphism
$\hat{\al B}_v/\rho_v\to\hat{\al B}_w/\rho_w$.
We will establish this property by proving that
$\tilde{\al B}_{vw}:=\phi(\hat{\al B}|_{vw})$ 
is equal to the image of $\bar{\al R}$ 
under the homomorphism in~\eqref{eq-comp8},
which is $\phi(\bar{\al{R}}|_{st})$, because the map 
$\twoheadrightarrow$ in~\eqref{eq-comp8} is 
$\phi$. Let $\tilde{\al R}_{st}:=\phi(\bar{\al R}|_{st})$.
Since $\hat{\al B}|_{vw}$ is a subdirect product of
$\hat{\al B}_v$ and $\hat{\al B}_w$, we get that
$\tilde{\al B}_{vw}$ is a subdirect product of
$\hat{\al B}_v/\rho_v$ and $\hat{\al B}_w/\rho_w$.
By the construction of $\al R$, we have that $\al R\ge\al{B}|_T$, 
and $\ge$ is preserved under the natural homomorphism
\begin{equation}
\label{eq-nathomT}
\prod_{r\in T}\al B_r\twoheadrightarrow\prod_{r\in T}\al B_r/\delta_r
=\prod_{u\in\hat{T}}\hat{\al B}_u.
\end{equation}
The images of $\al R$ and $\al{B}|_T$ under this homomorphism are
$\bar{\al R}$ and $\hat{\al B}|_{\hat{T}}$, respectively, hence
we conclude that $\bar{\al R}\ge\hat{\al B}|_{\hat{T}}$.
Projecting further onto the coordinates $s,t$ in $T$, and the corresponding
coordinates $v=(s,\delta_s)$, $w=(t,\delta_t)$ in $\hat{T}$, we get that
$\bar{\al R}|_{st}\ge\hat{\al B}|_{vw}$.
Hence, it follows that
$\tilde{\al R}_{st}=\phi(\bar{\al R}|_{st})
\ge\phi(\hat{\al B}|_{vw})=\tilde{\al B}_{vw}$.
By conclusion~\eqref{it:centiso} above,
$\tilde{\al R}_{st}$ is the graph of an isomorphism 
$\bar{\al B}_s/\rho_s\to\bar{\al B}_t/\rho_t$, or equivalently, the graph of an
isomorphism $\hat{\al B}_v/\rho_v\to\hat{\al B}_w/\rho_w$.
Combining this fact with the earlier observation that
$\tilde{\al B}_{vw}$ is a subdirect product of 
$\hat{\al B}_v/\rho_v$ and $\hat{\al B}_w\rho_w$, 
we obtain that  
$\tilde{\al R}_{st}$ and $\tilde{\al B}_{vw}$ must be equal.
In particular, $\tilde{\al B}_{vw}$ is the graph of an
isomorphism $\hat{\al B}_v/\rho_v\to\hat{\al B}_w/\rho_w$, and hence $v\sim w$.

Our arguments in the last two paragraphs show that 
$\hat{T}$ is contained in one of the blocks $U$ of $\sim$.
We have $|U|\ge|\hat{T}|=|T|\ge\max\{\ddd,3\}$. 
It remains to verify that $\hat{c}|_U\notin\hat{\al B}|_U$.

Assume, for a contradiction, that 
$\hat{c}|_U\in\hat{\al B}|_U$.
Then projecting further to $\hat{T}\subseteq U$ yields that
$\hat{c}|_{\hat{T}}\in\hat{\al B}|_{\hat{T}}$.
As we saw earlier, $\hat{\al B}|_{\hat{T}}\le\bar{\al R}$, therefore
we get that the tuple $\hat{c}|_{\hat{T}}=(c_r/\delta_r)_{r\in T}$ lies in
$\bar{\al R}$. Hence the tuple
$c|_T=(c_r)_{r\in T}$ lies in the full inverse image of $\bar{\al R}$
under the natural homomorphism (\ref{eq-nathomT}).
This inverse image is $\al R$, because 
$\bar{\al R}=\al R/\delta_\al R$
with $\delta=\prod_{r\in T}\delta_r$,
and $\al R$ is $\delta$-saturated in $\prod_{r\in T}\al B_r$.
Thus, we obtain that $c|_T\in\al R$, which 
is impossible, since $\al R$ was chosen so that
$c|_T\notin\al R$. 
This contradiction proves that
$\hat{c}|_U\notin\hat{\al B}|_U$, and
completes the proof of Theorem~\ref{thm-repr}.
\end{proof}

\section{Algorithms Based on Theorem~\ref{thm-repr}}
\label{sec-str-alg}

Throughout this section we will work under the following global
assumptions.

\begin{asm} \label{ass-6}\quad 
  \begin{itemize}
  \item
    $\var{V}$ is a variety in a finite language with a $\ddd$-cube term ($\ddd>1$),
  \item
    $\class{K}$ is a finite set of finite algebras in $\var{V}$.
  \end{itemize}  
\end{asm}    

\begin{defi}
\label{def-dcoh}
Let $a_1,\dots,a_k,b\in\al A_1\times\dots\times\al A_n$
($\al A_1,\dots,\al A_n\in\class{K}$) be an input for $\SMP(\class{K})$
where $a_r=(a_{r1},\dots,a_{rn})$ ($r\in[k]$) and $b=(b_1,\dots,b_n)$.
We call this input \emph{$\ddd$-coherent} if the following
conditions are satisfied:
\begin{enumerate}
\item \label{it:nd3} 
$n\ge\max\{\ddd,3\}$;
\item 
 $a_1,\dots,a_k$ generate a subdirect subalgebra of $\al A_1\times\dots\times\al A_n$;
\item \label{it:ssi} 
$\al A_1,\dots,\al A_n$ are similar subdirectly irreducible algebras, 
and each $\al A_\ell$ has abelian monolith $\mu_\ell$;
\item \label{it:bId3} 
$b|_I$ is in the subalgebra of $\prod_{i\in I}\al A_i$
generated by $\{a_1|_I,\dots,a_k|_I\}$
for all $I\subseteq[n]$, $|I|<\max\{\ddd,3\}$; and
\item \label{it:cent} 
for the centralizers $\rho_\ell:=(0:\mu_\ell)$ of the monoliths $\mu_\ell$,
the subalgebra of $\al A_i/\rho_i\times\al A_j/\rho_j$ generated by
$\{(a_{1i}/\rho_i,a_{1j}/\rho_j),\dots,(a_{ki}/\rho_i,a_{kj}/\rho_j)\}$
is the graph of an isomorphism $\al A_i/\rho_i\to\al A_j/\rho_j$
for any distinct $i,j\in[n]$.
\end{enumerate}
\end{defi}

\begin{defi}
\label{def-dcohSMP}
We define $\SMPd(\class{K})$ to be 
the restriction of $\SMP(\class{K})$ to $\ddd$-coherent inputs.
\end{defi}

It is clear from Definition~\ref{def-dcoh} that $\ddd$-coherence for inputs
of $\SMP(\class{K})$ can be checked in polynomial time.

\begin{algorithm}
\caption{Reduction of $\SMP(\class{K})$ to $\SMPd(\HH\SSS\class{K})$}
\label{alg-smp-red-smpd} 
\begin{algorithmic}[1] 
\REQUIRE $b_1,\dots,b_k,b_{k+1}\in\al A_1\times\dots\times\al A_n$ 
with $\al A_1,\dots,\al A_n\in\class{K}\ (n\geq\ddd)$
\ENSURE Is $b_{k+1}$ in the subalgebra of 
$\al A_1\times\dots\times\al A_n$ generated by $b_1,\dots,b_k$?
\STATE $\ddd^-:=\max(\ddd-1,2)$,
\FOR{$I\in\binom{[n]}{\ddd^-}$}
\STATE generate $\al B|_I$ by $b_1|_I,\dots,b_k|_I$
\STATE \return\ NO if $b_{k+1}|_I\notin\al B|_I$ 
\ENDFOR
\STATE generate $\al B_i$ by $b_1|_i,\dots,b_k|_i$ for $i\in[n]$
\STATE omit all coordinates $i$ from the input for which $|\al B_i|=1$
(but keep earlier notation)
\STATE \return\ YES if $n\le\ddd^-$
\STATE $W:=\{(j,\sigma):j\in[n],\sigma\in\Irr(\al B_j)\}$
\STATE $\hat{\al B}_w:=\al B_j/\sigma$,\ \  
$\mu_w:=\text{monolith of }\hat{\al B}_w$,\ \ 
$\rho_w:=(0:\mu_w)$ for each $w=(j,\sigma)\in W$
\STATE $\hat{b}_i:=(\hat{b}_{iw})_{w\in W}$ where $\hat{b}_{iw}:=b_i|_j/\sigma$ for each $i\in[k+1]$ and $w=(j,\sigma)$

\FOR{distinct $v,w\in W$}
\STATE generate $\tilde{\al B}_{vw}$ by the pairs
$(\hat{b}_{iv}/\rho_v,\hat{b}_{iw}/\rho_w)$, $i\in[k]$
\ENDFOR
\STATE compute the equivalence relation $\sim$ on $W$ determined 
by the following condition: for distinct $v,w\in W$, 
\[\qquad\quad
v\sim w
\quad\Leftrightarrow\quad
\begin{cases}
\mu_v\le\rho_v,\ \mu_w\le\rho_w,\ 
\text{$\hat{B}_v$ and $\hat{B}_w$ are similar, and}\\
\text{$\tilde{\al B}_{vw}$ is the graph of an isomorphism
$\hat{\al B}_v/\rho_v\to\hat{\al B}_w/\rho_w$}
\end{cases}
\]
\STATE find the equivalence classes $E_1,\dots,E_\kappa$ of $\sim$ of size $>\ddd^-$
\FOR{$\lambda=1,\dots,\kappa$}
\STATE run $\SMPd(\HH\SSS\class{K})$ with input
$\hat{b}_1|_{E_\lambda},\dots,\hat{b}_{k}|_{E_\lambda},
\hat{b}_{k+1}|_{E_\lambda}\in\prod_{w\in E_\lambda}\hat{\al B}_w$
to get answer $\sf{A}\in\{\text{YES},\text{NO}\}$
\STATE \return\ NO if $\sf{A}=\text{NO}$\\
\ENDFOR 
\STATE \return\ YES
\end{algorithmic}
\end{algorithm}

\begin{thm}
\label{thm-polyequiv}
The decision problems $\SMP(\class{K})$ and 
$\SMPd(\HH\SSS\class{K})$ are
polynomial time equivalent.
\end{thm}

\begin{proof}
By Theorem~\ref{thm-smp-hom}, $\SMP(\class{K})$ is polynomial time equivalent
to $\SMP(\HH\SSS\class{K})$. Clearly, 
$\SMPd(\HH\SSS\class{K})$ is polynomial time reducible to 
$\SMP(\HH\SSS\class{K})$,
because it is a subproblem of $\SMP(\HH\SSS\class{K})$. 
Therefore it only remains to show that Algorithm~\ref{alg-smp-red-smpd} reduces
$\SMP(\class{K})$ to
$\SMPd(\HH\SSS\class{K})$ in polynomial time.

The correctness of Algorithm~\ref{alg-smp-red-smpd} is based on Theorem~\ref{thm-repr}.
Let $\al B$ denote the subalgebra of $\al A_1\times\dots\times\al A_n$
generated by $b_1,\dots,b_k$.
Without loss of generality we may assume that $n\geq\ddd$.
In Steps~1--5, it is checked whether $b_{k+1}|_I$ is contained in the projection $\al B|_I$ of $\al B$ for each
 $I\in\binom{[n]}{\ddd^-}$.
If not,
then clearly $b_{k+1}\notin\al B$ and the algorithm correctly
returns the answer NO in Step~4.
In Step~6 the algebras $\al B_i:=\al B|_i$ are computed for every 
$i\in[n]$. Note that by this time, 
$b_{k+1}$ had passed the tests in 
Step~4, which implies that
$b_{k+1}|_i\in\al B_i$ holds for every $i\in[n]$. 
Therefore, when in 
Step~7
all coordinates with $|B_i|=1$ are omitted from
the input tuples (but the earlier notation is kept for simplicity),
this deletion of the trivial coordinates does no affect
whether or not $b_{k+1}\in\al B$.
It follows that by the end of 
Step~7 we have that
\begin{enumerate}
\item \label{it:Bsd} 
$\al B$ is a subdirect subalgebra of $\al B_1\times\dots\times\al B_n$
where $\al B_1,\dots,\al B_n$ are nontrivial, and 
\item \label{it:bII}
$b_{k+1}|_I\in\al B|_I$ 
for all $I\in\binom{[n]}{\ddd^-}$.
\end{enumerate}

By~\eqref{it:Bsd} the algebra $\al B$ satisfies the assumptions of
Theorem~\ref{thm-repr}, and~\eqref{it:bII} shows that the tuple
$c:=b_{k+1}\in\al B_1\times\dots\times\al B_n$
satisfies the first part of condition~\eqref{it:c2} in Theorem~\ref{thm-repr}.

In Steps~9--16, Algorithm~\ref{alg-smp-red-smpd} computes 
the data needed to check whether or not
$c=b_{k+1}$ satisfies the second part
of condition~\eqref{it:c2} as well.
In more detail, the algorithm computes the set $W$,
the subdirectly irreducible 
algebras $\hat{\mathbf{B}}_w$ ($w\in W$), their congruences $\mu_w$ (the monolith) and
$\rho_w=(0:\mu_w)$,
the images $\hat{b}_i$ of the input tuples $b_i$ ($i\in[k+1]$) 
under the homomorphism $\hat{\phantom{d}}$, and finally, the equivalence relation $\sim$
on $W$, and its equivalence classes $E_1,\dots,E_\kappa$ of size ${}>\ddd^-$.
All computations follow the definitions exactly, except Step~15.
We explain now why the conditions used in 
Step~15 to compute $\sim$ are 
equivalent to the conditions in the definition of $\sim$ (stated 
right before Theorem~\ref{thm-repr}):
\begin{itemize}
\item
The inclusion $\mu_w\le\rho_w$ ($w\in W$) 
is equivalent to the condition that $\mu_w$ is abelian, because
the inclusion is true if $\mu_w$ is abelian, whereas
$\rho_w=0$ (and hence the inclusion fails)
if $\mu_w$ is nonabelian. 
\item
As discussed in the proof of Theorem~\ref{thm-repr}, 
$\tilde{B}_{vw}$ is the image of $\hat{B}|_{vw}$ under the 
natural map 
$\hat{\al B}_v\times\hat{\al B}_w\twoheadrightarrow
(\hat{\al B}_v/\rho_v)\times(\hat{\al B}_w/\rho_w)$.
\end{itemize}

Since $c=b_{k+1}$ satisfies the first part of condition~\eqref{it:c2}
in Theorem~\ref{thm-repr}, it follows from the
equivalence of conditions \eqref{it:cB} and~\eqref{it:c2} in Theorem~\ref{thm-repr}
that $b_{k+1}\in\al B$ if and only if 
$c=b_{k+1}$ satisfies the second part of condition~\eqref{it:c2} as well, that is,
\begin{enumerate}
\item[(3)]
$\hat{b}_{k+1}|_{E_\lambda}\in\hat{\al B}|_{E_\lambda}$
for all $\lambda\in[\kappa]$.
\end{enumerate} 
This is clearly equivalent to the condition that
\begin{enumerate}[align=left]
\item[(3)$'$]
$\hat{b}_{k+1}|_{E_\lambda}$ belongs to the subalgebra of 
$\prod_{w\in E_\lambda} \hat{B}_w$ generated by the tuples\break
$\hat{b}_1|_{E_\lambda},\dots,  \hat{b}_k|_{E_\lambda}$
for all $\lambda\in[\kappa]$.
\end{enumerate} 
It follows from the construction that for each $\lambda\in[\kappa]$,
$\hat{b}_1|_{E_\lambda},\dots,  \hat{b}_k|_{E_\lambda}, \hat{b}_{k+1}|_{E_\lambda}$
is a $\ddd$-coherent input for $\SMP(\HH\SSS\class{K})$, so
condition (3)$'$ can be checked using $\SMPd(\HH\SSS\class{K})$.
This is exactly what Algorithm~\ref{alg-smp-red-smpd} does in 
Steps~17--20, and
returns the correct answer: YES if (3)$'$ holds 
(including the case
when $\kappa=0$), and NO otherwise.
This completes the proof of the correctness of Algorithm~\ref{alg-smp-red-smpd}.

Now we show that Algorithm~6 reduces $\SMP(\class{K})$
to $\SMPd(\HH\SSS\class{K})$ in polynomial time.
Clearly, Steps~1,~8, and~21 require constant time.
Since each $\al B_i$ ($i\in[n]$) is a subalgebra of some member of
$\class{K}$, we have $|\al B_i|\le\aaa_{\class{K}}$ where
the constant $\aaa_{\class{K}}$ is independent of the input.
The parameter $\ddd$ is also independent of the input, and
so is $\sss:=\max\{|\Irr(\al A)|:\al A\in\SSS\class{K}\}$.
It follows that $|W|\le n\sss$ and that 
each iteration of the \for\ loops in Steps~2--5 and~12--14 require time $O(k)$.
Hence, Steps~2--5, 6--7, 9--10, 11, and~12--14
run in $O(kn^{\ddd^-})$, $O(kn)$, $O(n)$, $O(kn)$, and $O(kn^2)$ time, respectively.

In 
Step~15, to determine whether $v\sim w$ holds for a particular pair
of elements $v,w\in W$ requires constant time, because the condition
only involves data on algebras in $\HH\SSS\class{K}$ and on products
of two such algebras. (In particular, recall from Section~2 that similarity of 
$\hat{B}_v$ and $\hat{B}_w$ can be checked by looking at congruences
of subalgebras of $\hat{B}_v\times\hat{B}_w$.)
Thus, Step~15 runs in $O(n^2)$ time as does Step~16.
Since $E_1,\ldots,E_\kappa$ are disjoint subsets of $W$ and $|W|\le n\sss$,
we get that $\kappa\le n\sss$ and each $E_\lambda$ has size $|E_\lambda|\le n\sss$.
Thus, in 
Steps~17--20, $\SMPd(\HH\SSS\class{K})$ has to be run at most $O(n)$
times, and the input size of each run is $O(kn)$, approximately the 
same as the size of the original input.

This proves that Algorithm~\ref{alg-smp-red-smpd} reduces $\SMP(\class{K})$ to
$\SMPd(\HH\SSS\class{K})$ in polynomial time.
\end{proof}

\begin{algorithm}
\caption{For $\SMPd(\class{K})$ if $\class{K}$ 
is in a residually small variety}
\label{alg-smp-red-sm-v}
\begin{algorithmic}[1] 
\REQUIRE $\ddd$-coherent $a_1,\dots,a_k,b\in\al A_1\times\dots\times\al A_n$ 
$(\al A_1,\dots,\al A_n\in\class{K})$ where \\
$a_i=(a_{i1},\dots,a_{in})$ for all $i\in[k]$ 
\ENSURE Is $b$ in the subalgebra of $\al A_1\times\dots\times\al A_n$
generated by $a_1,\dots,a_k$?
\STATE let $\mu_j$ be the monolith of $\al A_j$,  
$\rho_j:=(0:\mu_j)$ for every $j\in[n]$ and $\rho := \rho_1\times\dots\times\rho_n$
\STATE reindex $a_1,\dots,a_k$ so that
$a_{11}/\rho_1,\dots,a_{r1}/\rho_1$ are pairwise distinct
and $\{a_{11}/\rho_1,\dots,a_{r1}/\rho_1\}=\{a_{11}/\rho_1,\dots,a_{k1}/\rho_1\}$\\
let $\mathcal{O} :=\{a_1,\dots,a_r\}$ 
\STATE generate $\al A_1/\rho_1$ by $a_{11}/\rho_1,\dots,a_{r1}/\rho_1$, and simultaneously,
\FOR{each new 
$a/\rho_1=t(a_{11}/\rho_1,\dots,a_{r1}/\rho_1)\in\al A_1/\rho_1$ 
($t$ a term)}
\STATE $\mathcal{O}:=\mathcal{O}\cup\{t(a_1,\dots,a_r)\}$  
\ENDFOR 
\STATE find the equivalence relation $\equiv$ on $[n]$ defined by
\[
s\equiv t\ \ \Leftrightarrow\ \ \al A_s=\al A_t \text{ and }
o|_s=o|_t \text{ for all $o\in\mathcal{O}$}
\qquad\text{($s,t\in[n]$)}
\]
let $T$ be a transversal for the blocks of $\equiv$, and let
$\al A_T:=\prod_{j\in T}\al A_j$
\STATE enumerate the elements of the subalgebra $P$ of $\al A_T^{A_T}$ generated by 
the identity function $A_T\to A_T$ and by the constant functions
with value $o|_T$ ($o\in\mathcal{O}$) 
(so every $p\in P$ is of the form $A_T\to A_T, (x_t)_{t\in T} \mapsto (p_t(x_t))_{t\in T}$ for some function $p_t\colon A_t\to A_t$)
\STATE find $o\in\mathcal{O}$ such that $b\in o/\rho$
\STATE $H:=\emptyset$
\FOR{$p\in P$}
\FOR{$c\in\{a_1,\dots,a_k\}$}
\FOR{$j\in[n]$}
\STATE find $t\in T$ with $t\equiv j$ 
\STATE let $d_j:=p_t(c|_j)$
\ENDFOR
\STATE $d := (d_1,\dots,d_n)$
\IF{$d\in o/\rho$}
\STATE $H:=H\cup\{d\}$
\ENDIF
\ENDFOR
\ENDFOR
\STATE run Sims' algorithm for $\SMP(\mathcal{G})$ with the input
$H\cup\{b\}\subseteq \al G_1\times\dots\times\al G_n$ to check whether $b\in\langle H\rangle$ where
$\al G_j$ is the group $(o_j/\rho_j;+_{o_j},-_{o_j},o_j)$ for each $j\in[n]$
and $\mathcal{G}$ is the family of all induced abelian groups on blocks
of abelian congruences of algebras in $\mathcal{K}$;
get answer $\sf{A}\in\{\text{YES},\text{NO}\}$
\STATE \return\ $\sf{A}$
\end{algorithmic}
\end{algorithm}

\begin{thm}
\label{thm-RScase}
If, in addition to Assumption~\ref{ass-6}, the variety $\var{V}$ is residually small, then
\[
\SMP(\class{K})\in\PP.
\]
\end{thm}

\begin{proof}
By Theorem~\ref{thm-polyequiv}, $\SMP(\class{K})$ is polynomial time
equivalent to $\SMPd(\HH\SSS\class{K})$. Also, if $\class{K}$ is finite, so
is $\HH\SSS\class{K}$. Therefore it suffices to prove that 
$\SMPd(\class{K})\in\PP$ holds for every $\class{K}$ 
as in the theorem.
We will prove this by showing that, under the assumptions of the theorem,
Algorithm~\ref{alg-smp-red-sm-v} solves $\SMPd(\class{K})$ in polynomial time.

First we discuss the correctness of Algorithm~\ref{alg-smp-red-sm-v}.
Let $a_1,\dots,a_k,b\in\al A_1\times\dots\times\al A_n$ 
($\al A_1,\dots,\al A_n\in\class{K}$) be a correct input for
$\SMPd(\class{K})$ (i.e., a $\ddd$-coherent input for $\SMP(\class{K})$).
Then conditions~\eqref{it:nd3}--\eqref{it:cent} in Definition~\ref{def-dcoh} hold.
By condition~\eqref{it:ssi},
the algebras $\al A_j$ ($j\in[n]$) are subdirectly irreducible
with abelian monoliths, so in Step~1 of Algorithm~\ref{alg-smp-red-sm-v}
the monoliths $\mu_j$ and their
centralizers $\rho_j$ will be found. Moreover, since $\class{K}$ is
assumed to be in a residually small variety, we get from 
Theorem~\ref{thm-cube-par}\eqref{it:cubecm} and
Corollary~\ref{cor-rs} that
\begin{equation} \label{eq:rhoj}
\rho_j \text{ is an abelian congruence of } \al A_j \text{ for every } j\in[n].
\end{equation}

Let $\al B$ denote the subalgebra of $\al A_1\times\dots\times\al A_n$
generated by the input tuples $a_1,\dots,a_k$, and let 
$\rho$ denote the product congruence $\rho_1\times\dots\times\rho_n$
on $\al A_1\times\dots\times\al A_n$. The restriction of $\rho$ to $\al B$
will be denoted by $\rho_{\al B}$.
Condition~\eqref{it:cent} in Definition~\ref{def-dcoh} implies that 
\begin{equation} \label{eq:BAj}
\text{the map } \al B/\rho_{\al B}\to \al A_j/\rho_j, 
(x_1,\dots,x_n)/\rho_{\al B}\mapsto x_j/\rho_j
\text{ is a bijection for every } j\in[n].
\end{equation}
Finally, conditions~\eqref{it:nd3} and~\eqref{it:bId3} together imply that for the input tuple
$b$ we have $b|_I\in\al B|_I$ for all sets $I\in\binom{[n]}{2}$.
Hence, we have 
$b|_{i,j}/(\rho_i\times\rho_j)\in \al B|_{i,j}/(\rho_i\times\rho_j)$
for all $i,j\in[n]$.
Since, by condition~\eqref{it:cent}, $\al B|_{i,j}/(\rho_i\times\rho_j)$ is the graph
of an isomorphism $\al A_i/\rho_i\to\al A_j/\rho_j$
for every pair $i,j\in[n]$ ($i\neq j$), it follows that 
the tuple $b/\rho=(b|_1/\rho_1,\dots,b|_n/\rho_n)$ belongs to
$\al B[\rho]/\rho$.
 Hence, 
\begin{equation} \label{eq:Brho}
 b\text{ is an element of the algebra } \al B[\rho], \text{ the } \rho\text{-saturation of } \al B.
\end{equation}

We return to the analysis of Algorithm~\ref{alg-smp-red-sm-v}.
After appropriately reindexing $a_1,\dots,a_k$,
Steps~2--6 produce a subset $\mathcal{O}$ of $\al B$
such that the first coordinates of the tuples in $\mathcal{O}$
form a transversal for the $\rho_1$-classes of $\al A_1$.
Thus, it follows from~\eqref{eq:BAj} that the tuples in
$\mathcal{O}$ form a transversal for the $\rho_{\al B}$-classes of
$\al B$, and hence also for the $\rho$-classes of $\al B[\rho]$. 
Let $|\mathcal{O}|=\ell$.
Since $|\mathcal{O}|=\bigl|\mathcal{O}|_1\bigr|$, we have
$\ell=|{\al A}_1/\rho_1|<|A_1|$.

Now let $\equiv$, $T$, and $\al A_T$ be as defined (and computed)
in 
Step~7, and let $\al A:=\al A_1\times\dots\times\al A_n$.
During our discussion of 
Steps~8--24 we will use computations
where the same term is used in different algebras. 
To make it easier for the reader to keep track of where each 
computation takes place, we will use superscripts to
indicate the relevant algebras.

It is easy to see that 
the set $P$ of functions $A_T\to A_T$ computed in 
Step~8 is 
\begin{equation}
\label{eq-setP}
P=\{\ttt^{{\al A}_T}(x,\mathcal{O}|_T):
\text{$\ttt$ is a $(1+\ell)$-ary term}\}
\end{equation}
where we assume that an ordering of $\mathcal{O}$ has been fixed to
ensure that its elements are always substituted into terms in that 
fixed order.
Our observation~\eqref{eq:Brho} shows that in 
Step~9 Algorithm~\ref{alg-smp-red-sm-v} will find
an element $o\in\mathcal{O}$ such that $b\in o/\rho$.

\begin{clm}
\label{clm-setH}
The set $H$ obtained by Algorithm~\ref{alg-smp-red-sm-v} after completing 
Steps~10--22 is  
\begin{equation}
\label{eq-setH}
H=\{\ttt^{\al A}(a_i,\mathcal{O})\in o/\rho: i\in[k],\ 
\text{$\ttt$ is a $(1+\ell)$-ary term}\}.
\end{equation}
\end{clm}

\begin{proof}[Proof of Claim~\ref{clm-setH}.]
We will use the notation from
Steps~10--22.
Let $p\in P$ and $c\in\{a_1,\dots,a_k\}$,
say $c=(c_1,\dots,c_n)$. 
By \eqref{eq-setP}, $p\in P$ if and only if $p$
is a unary polynomial operation of $\al A_T$ of the form
$\ttt^{{\al A}_T}(x,\mathcal{O}|_T)$ 
for some $(1+\ell)$-ary term $\ttt$. 
Our goal is to show that for every choice of $p$ and $c$
the tuple $d$ computed in Steps~13--17
is 
\begin{equation}
\label{eq-element-d}
d=\ttt^{\al A}(c,\mathcal{O}).
\end{equation}
This will prove that the tuples $d$ computed in 
Steps~13--17
are exactly the elements of $\al A$ of the form
$\ttt^{\al A}(c,\mathcal{O})$ where $c\in\{a_1,\dots,a_k\}$
and $\ttt$ is a $(1+\ell)$-ary term. 
Since such a tuple $d$ is added to $H$ in 
Steps~18--20
if and only if $d$ also satisfies $d\in o/\rho$,
the equality \eqref{eq-setH} will follow.

To verify \eqref{eq-element-d} let $j\in[n]$ and let $t\in T$ 
be the unique transversal element such that $t\equiv j$.
By the definition of $\equiv$ we have that 
 $o'|_t=o'|_j$ for all $o'\in\mathcal{O}$.
The latter condition may be written as $\mathcal{O}|_t=\mathcal{O}|_j$
(with the fixed ordering of $\mathcal{O}$ in mind,
these are tuples of elements in $\al A_t=\al A_j$).
The function $p_t\colon \al A_t\to\al A_t$ computed in 
 Step~8
is the polynomial function
$\ttt^{{\al A}_t}(x,\mathcal{O}|_t)$ of $\al A_t$.
Thus, using the equalities $\al A_t=\al A_j$
and $\mathcal{O}|_t=\mathcal{O}|_j$ we get that
\[
d_j=p_t(c_j)
=\ttt^{{\al A}_t}(c_j,\mathcal{O}|_t)
=\ttt^{{\al A}_j}(c_j,\mathcal{O}|_j)
=\ttt^{\al A}(c,\mathcal{O})|_j.
\]
This holds for every $j\in[n]$, so the proof of 
\eqref{eq-element-d}, and hence the proof of Claim~\ref{clm-setH},
is complete.
\renewcommand{\qedsymbol}{$\diamond$}
\end{proof}

To establish the correctness of the last two steps of
Algorithm~\ref{alg-smp-red-sm-v} recall from~\eqref{eq:rhoj} that 
$\rho_j$ is an abelian congruence of $\al A_j$ for every
$j\in[n]$. Therefore, $\rho=\rho_1\times\dots\times\rho_n$ 
is an abelian congruence of 
$\al A=\al A_1\times\dots\times\al A_n$. 
Hence, by Theorem~\ref{thm-abelian-congr}~\eqref{it:abelian1}, 
there is an induced abelian group
$\al G:=(o/\rho;+_0,-_o,o)$ on the $\rho$-class
$o/\rho$.
Moreover, since $\rho$ is the product congruence 
$\rho_1\times\dots\times\rho_n$ of $\al A$, we get that
$\al G=\al G_1\times\dots\times\al G_n$ where
$\al G_j$ is the group $(o_j/\rho_j;+_{o_j},-_{o_j},o_j)$
for every $j\in[n]$.
Recall also that $b\in o/\rho$, that is, 
$b$ is contained in $\al G=\al G_1\times\dots\times\al G_n$.

\begin{clm}
\label{clm-red-to-groups}
The following conditions on $b$ are equivalent:
\begin{enumerate}
\item \label{it:binB} 
$b$ is in the subalgebra $\al B$ of $\al A=\al A_1\times\dots\times\al A_n$ 
generated by $\{a_1,\dots,a_k\}$;
\item \label{it:binS} 
$b$ is in the subgroup of $\al G=\al G_1\times\dots\times\al G_n$
generated by the set $H$.
\end{enumerate}
\end{clm}

\begin{proof}[Proof of Claim~\ref{clm-red-to-groups}.]
To prove the implication 
 \eqref{it:binB}~$\Rightarrow$~\eqref{it:binS}
  assume that 
$b\in\al B$, that is, $b=g(a_1,\dots,a_k)$ for some
$k$-ary term $g$.
For each $i\in[k]$, let $o^{(i)}$
denote the unique element of $\mathcal{O}$
in the $\rho_{\al B}$-class of $a_i$.
Since $g(a_1,\dots,a_k)=b\in o/\rho$,
it follows from Theorem~\ref{thm-abelian-congr}~\eqref{it:abelian2}
that 
\begin{multline*}
g(a_1,\dots,a_k)=
g(a_1,o^{(2)}\dots,o^{(k)})
+_o g(o^{(1)},a_2,o^{(3)}\dots,o^{(k)})
+_o\dots\\
+_o g(o^{(1)},\dots,o^{(k-1)},a_k)
-_o(k-1)g(o^{(1)},\dots,o^{(k-1)},o^{(k)}).
\end{multline*}
All $+_o$-summands on the right hand side belong to $H$, therefore
$b$ is in the subgroup of 
$\al G=\al G_1\times\dots\times\al G_n$ generated by $H$.

For the reverse implication 
\eqref{it:binS}~$\Rightarrow$~\eqref{it:binB}
notice first that
$H\subseteq B$, because $\al B$ is a subalgebra of 
$\al A=\al A_1\times\dots\times\al A_n$ and 
the elements of $H$ are obtained from 
$a_1,\dots,a_k\in B$ by unary polynomial operations of $\al A$
that are constructed from term
operations by
using parameters from $\mathcal{O}\subseteq B$ only.
Since the group operations $+_o$, $-_o$, $o$ of 
$\al G$ are also polynomial operations
of $\al A$ obtained from term operations
using parameters from $\mathcal{O}$ only, we get that
the subgroup of $\al G$ generated by $H$ is contained in $\al B$.
\renewcommand{\qedsymbol}{$\diamond$}
\end{proof}

As in 
 Step~23 of Algorithm~\ref{alg-smp-red-sm-v}, let $\class{G}$ denote the set of all
induced abelian groups on blocks of 
abelian congruences of algebras in $\class{K}$.
Then, clearly, $\al G_1,\dots,\al G_n\in\class{G}$.
Therefore,
Claim~\ref{clm-red-to-groups} shows that $\SMPd(\class{K})$ run with the input
$a_1,\dots,a_k,b$ in $\al A_1\times\dots\times\al A_n$ 
($\al A_1,\dots,\al A_n\in\class{K}$)
has the same answer as $\SMP(\class{G})$ run with the input $H$ and $b$
in $\al G_1\times\dots\times\al G_n$ 
($\al G_1,\dots,\al G_n\in\class{G}$).
Hence, Algorithm~\ref{alg-smp-red-sm-v} finds the correct answer in 
Steps~23, 24. The proof
of the correctness of Algorithm~\ref{alg-smp-red-sm-v} is complete.

To prove that Algorithm~\ref{alg-smp-red-sm-v} runs in polynomial time, we will estimate
the time complexity of each step separately.
Recall that $\aaa_{\class{K}}$ denotes the maximum size of an algebra 
in $\class{K}$.

Steps~1 and~3--6 run in $O(n)$ time, while Step~2 runs in $O(kn)$ time.
In 
Step~7, $\equiv$ and a transversal $T$ for the
blocks of $\equiv$
can be found in $O(n^2)$ time, 
since $\ell=|\mathcal{O}|$ is bounded above by a constant 
($< |A_1|\le \aaa_{\class{K}}$) which is
independent of the size of the input.
Since $\equiv$ is the kernel of the map 
$[n]\to\class{K}\times \bigcup_{i\in[n]}A_i^\ell$, 
$j\mapsto (\al A_j,\mathcal{O}|_j)$, the number $|T|$
of the $\equiv$-blocks
is at most $|\class{K}|\aaa_{\class{K}}^\ell$,
independent of the input. 
Therefore, $\al A_T$ can be computed in constant time, and
Step~7 
altogether requires $O(n^2)$ time.
For the same reason, $|\al A_T^{A_T}|$ is also bounded above by a constant,
independent of the input, therefore 
Step~8 runs in constant time.
Step~9 also runs in constant time, because, in view of~\eqref{eq:BAj},
$b\in o/\rho$ is equivalent to $b|_1\in o|_1/\rho_1$.
Clearly, 
Step~10 also runs in constant time.

Using the previous estimates on $|P|\,(\le |\al A_T^{A_T}|)$ and $|T|$
we see that the number of iterations of the outer 
\for\ loop (line~11) is bounded above by a constant, while the inner \for\ loops starting in Step~12 and Step~13
  iterate $k$ and $n$ times, respectively. In each iteration, Steps 14--15 and 18--20 require constant time.
Thus Steps 11--22 run in $O(kn)$ time.

In 
Steps~11--22 at most one element is added to $H$ for each choice of
$p\in P$ and $c\in\{a_1,\dots,a_k\}$. Since $|P|$ is bounded above
by a constant independent of the input, we get that $|H|$ has size 
$O(k)$.
Thus, in 
Step~23 the size of the input $H\cup\{b\}$ for
$\SMP(\class{G})$ is $O(kn)$. 
Moreover, the size of each group in $\class{G}$ is $\le\aaa_{\class{K}}$.
Since Sims' algorithm for $\SMP(\class{G})$ 
runs in $O(kn^3)$ time~\cite[Corollary 3.7]{Zw:CDP} on an input
$H\cup\{b\}\subseteq\al G_1\times\dots\times\al G_n$ 
with $|H|=O(k)$, we get that 
Step~23 of Algorithm~\ref{alg-smp-red-sm-v} requires $O(kn^3)$
time. Clearly, 
Step~24 runs in constant time.

Combining the time complexities of 
Steps~1--24 we get that Algorithm~\ref{alg-smp-red-sm-v}
runs in $O(kn^3)$ time. This completes the proof of 
Theorem~\ref{thm-RScase}.
\end{proof}

\subsection*{Acknowledgement.}
The authors are thankful to the anonymous referees for their careful reading and suggestions.

\def\cprime{$'$}

\end{document}